\documentclass{lmcs} 

\pdfoutput=1

\usepackage{centernot}
\usetikzlibrary{arrows,automata,shapes,decorations,calc,patterns,positioning,math}
\tikzmath{
    \distX=40;
    \distY=28;
    \lineGap=(0.25*\distX) - \distX;
    \lineCut=(0.5*\distX);
    \barX=6;
    \barY=3;
    \textY=8;
    \shiftX=15;
    \shiftXbig=2*\shiftX;
}
\usepackage{algcompatible}
\usepackage[noend]{algpseudocode}
\usepackage{algorithm}
\usepackage{longtable, booktabs}
\usepackage{multirow}
\usepackage{hyperref}
\usepackage[nameinlink,capitalise,noabbrev]{cleveref}
\usepackage[normalem]{ulem}

\newcommand{\lFor}[2]{%
    \State\algorithmicfor\ {#1}\ \algorithmicdo\ {#2}%
  }
\newcommand{\lWhile}[2]{%
    \State\algorithmicwhile\ {#1}\ \algorithmicdo\ {#2}%
  }
\newcommand{\varname}[1]{\ensuremath{\mathit{#1}}}
\newcommand{\varidx}[2]{\ensuremath{\varname{#1}[#2]}}
\newcommand{\writeop}{\ensuremath{\gets}}
\newcommand{\optwriteop}{\ensuremath{~^{\star}\!\!\writeop^{\star}}}
\newcommand{\defeq}{\stackrel{\text{def}}{\equiv}}
\newcommand{\VBar}[2]{($ (0,#1) + #2 $) -- ($(0,-#1)+ #2 $)}
\newcommand{\Above}[2]{($ (0,#1) + #2 $)}
\newcommand{\Below}[2]{($ (0,-#1) + #2 $)}
\newcommand{\concsym}{\smile^{\hspace{-.5ex}\raisebox{-.2ex}{\tiny$\bullet$}}}
\newcommand{\nconcsym}{{\centernot\smile}^{\hspace{-.5ex}\raisebox{-.4ex}{\tiny$\bullet$}}}
\newcommand{\conc}{\ensuremath{\mathbin{\concsym}}}
\newcommand{\nconc}{\ensuremath{\mathbin{\nconcsym}}}
\newcommand{\comp}[1]{\ensuremath{\overline{\mathit{#1}}}}
\newcommand{\co}{\ensuremath{\cdot}}
\newcommand{\states}{\ensuremath{\mathcal{S}}}
\newcommand{\initstate}{\ensuremath{\mathit{init}}}
\newcommand{\actionset}{\ensuremath{\mathit{Act}}}
\newcommand{\transrel}{\ensuremath{\mathit{Trans}}}

\newcommand{\thrlocacts}{\ensuremath{\mathit{TLoc}}}
\newcommand{\block}{\ensuremath{\mathcal{B}}}
\newcommand{\nonblock}{\ensuremath{\comp{\block}}}

\newcommand{\elimf}[2]{\ensuremath{\#_{#1}\{#2\}}}
\newcommand{\thrsym}{\ensuremath{\mathit{thr}}}
\newcommand{\thrmap}[1]{\ensuremath{\thrsym(#1)}}
\newcommand{\regsym}{\ensuremath{\mathit{reg}}}
\newcommand{\regmap}[1]{\ensuremath{\regsym(#1)}}
\newcommand{\justact}[2]{\ensuremath{\mathit{JA}^{#1}_{#2}}}
\newcommand{\truevalsym}{\ensuremath{\mathit{stor}}}
\newcommand{\trueval}[1]{\ensuremath{\truevalsym(#1)}}
\newcommand{\safRep}{\ensuremath{\mathit{saf}}}
\newcommand{\regRep}{\ensuremath{\mathit{reg}}}
\newcommand{\atoRep}{\ensuremath{\mathit{ato}}}
\newcommand{\TID}{\ensuremath{\mathbb{T}}}
\newcommand{\RID}{\ensuremath{\mathbb{R}}}
\newcommand{\Data}[1][\rid]{\ensuremath{\mathbb{D}_{#1}}}
\newcommand{\startread}[1][i,x]{\ensuremath{\mathit{sr_{#1}}}}
\newcommand{\finishread}[2][i,x]{\ensuremath{\mathit{fr_{#1}(#2)}}}
\newcommand{\readop}[2][i,x]{\ensuremath{\mathit{rd_{#1}(#2)}}}
\newcommand{\readopempty}[1][i,x]{\ensuremath{\mathit{rd_{#1}}}}
\newcommand{\startwrite}[2][i,x]{\ensuremath{\mathit{sw_{#1}(#2)}}}
\newcommand{\finishwrite}[1][i,x]{\ensuremath{\mathit{fw_{#1}}}}
\newcommand{\orderwrite}[1][i,x]{\ensuremath{\mathit{ow_{#1}}}}
\newcommand{\orderread}[1][i,x]{\ensuremath{\mathit{or_{#1}}}}
\newcommand{\overlapsym}{\ensuremath{\mathit{ovrl}}}
\newcommand{\overlap}[1]{\ensuremath{\overlapsym(#1)}}
\newcommand{\posvalsym}{\ensuremath{\mathit{posv}}}
\newcommand{\posval}[1]{\ensuremath{\posvalsym(#1)}}
\newcommand{\readerssym}{\ensuremath{\mathit{rds}}}
\newcommand{\readers}[1]{\ensuremath{\readerssym(#1)}}
\newcommand{\writerssym}{\ensuremath{\mathit{wrts}}}
\newcommand{\writers}[1]{\ensuremath{\writerssym(#1)}}
\newcommand{\pendingsym}{\ensuremath{\mathit{pend}}}
\newcommand{\pending}[1]{\ensuremath{\pendingsym(#1)}}
\newcommand{\usrsym}{\ensuremath{\mathit{usr}}}
\newcommand{\usr}[1]{\ensuremath{\usrsym(#1)}}
\newcommand{\ufrsym}{\ensuremath{\mathit{ufr}}}
\newcommand{\ufr}[1]{\ensuremath{\ufrsym(#1)}}
\newcommand{\ursym}{\ensuremath{\mathit{urd}}}
\newcommand{\ur}[1]{\ensuremath{\ursym(#1)}}
\newcommand{\uswsym}{\ensuremath{\mathit{usw}}}
\newcommand{\usw}[1]{\ensuremath{\uswsym(#1)}}
\newcommand{\ufwsym}{\ensuremath{\mathit{ufw}}}
\newcommand{\ufw}[1]{\ensuremath{\ufwsym(#1)}}
\newcommand{\uowsym}{\ensuremath{\mathit{uow}}}
\newcommand{\uow}[1]{\ensuremath{\uowsym(#1)}}
\newcommand{\uorsym}{\ensuremath{\mathit{uor}}}
\newcommand{\uor}[1]{\ensuremath{\uorsym(#1)}}
\newcommand{\valssym}{\ensuremath{\mathit{rec}}}
\newcommand{\vals}[1]{\ensuremath{\valssym(#1)}}
\newcommand{\StatusAll}{\ensuremath{\mathbb{S}}}
\newcommand{\Reg}[1][m]{\ensuremath{\mathit{Reg}_{#1}}}
\newcommand{\undefsymb}{\ensuremath{\bot}}
\newcommand{\false}{\ensuremath{\mathit{false}}}
\newcommand{\true}{\ensuremath{\mathit{true}}}
\newcommand{\critsym}{\ensuremath{\mathit{c}}}
\newcommand{\crit}[1][i]{\ensuremath{\critsym_{#1}}}
\newcommand{\noncritsym}{\ensuremath{\mathit{nc}}}
\newcommand{\noncrit}[1][i]{\ensuremath{\noncritsym_{#1}}}
\newcommand{\idx}[1]{\ensuremath{\#({#1})}}
\newcommand{\tid}{\ensuremath{t}}
\newcommand{\tidtwo}{\ensuremath{t'}}
\newcommand{\rid}{\ensuremath{r}}
\newcommand{\data}{\ensuremath{d}}
\newcommand{\regtype}{\ensuremath{\gamma}}
\newcommand{\parcomp}{\ensuremath{\parallel}}
\newcommand{\then}{\ensuremath{\rightarrow}}
\newcommand{\satnone}{\ensuremath{\mathrm{X}}}
\newcommand{\satmutex}{\ensuremath{\mathrm{M}}}
\newcommand{\satdf}{\ensuremath{\mathrm{D}}}
\newcommand{\satsf}{\ensuremath{\mathrm{S}}}
\newcommand{\tp}{\ensuremath{\mathit{tt}}}

\newcommand{\allact}{\ensuremath{\actionset}}
\newcommand{\diam}[1]{\ensuremath{\langle \mathit{#1} \rangle}}
\newcommand{\boxm}[1]{\ensuremath{[ \mathit{#1} ]}}
\newcommand{\clos}[1]{\ensuremath{\mathit{#1}^\star}}
\newcommand{\imps}{\ensuremath{\Rightarrow}}
\newcommand{\WCT}{\textit{WCT}}
\newcommand{\isreadsym}{\ensuremath{\textit{read?}}}
\newcommand{\isread}[1]{\ensuremath{\isreadsym(#1)}}
\newcommand{\iswritesym}{\ensuremath{\textit{write?}}}
\newcommand{\iswrite}[1]{\ensuremath{\iswritesym(#1)}}
\newcommand{\modeltype}{\ensuremath{m}}
\newcommand{\irRep}{\ensuremath{\mathit{IR}}}
\newcommand{\frRep}{\ensuremath{\mathit{FR}}}
\newcommand{\norm}[1]{\ensuremath{\mathit{norm}(#1)}}
\DeclareMathAlphabet{\mathbbm}{U}{bbm}{m}{n}          
\definecolor{highlightColour}{named}{orange}
\newcommand{\colourname}{orange}
\newcommand{\highlight}[1]{\textcolor{highlightColour}{#1}}
\newcommand{\plat}[1]{\raisebox{0pt}[0pt][0pt]{#1}}  
\newcommand{\actequivsym}{\ensuremath{\equiv}}
\newcommand{\actequiv}[2]{\ensuremath{#1 \mathbin{\actequivsym} #2}}

\newcommand{\proj}{\ensuremath{\mathit{proj}}}
\newcommand{\finish}{\textit{end}}
\newcommand{\ftoi}{\ensuremath{\textit{\frRep2\irRep}}}
\newcommand{\itof}{\ensuremath{\textit{\irRep2\frRep}}}
\newcommand{\pftoi}{\ensuremath{\textit{pre-\ftoi}}}
\newcommand{\ftoistate}{\ensuremath{\textit{\frRep2\irRep}}}
\newcommand{\itofstate}{\ensuremath{\textit{\irRep2\frRep}}}

\crefname{line}{line}{lines}
\crefname{defi}{Definition}{Definitions}
\crefname{lem}{Lemma}{Lemmas}
\crefname{obs}{Observation}{Observations}
\crefname{cor}{Corollary}{Corollaries}
\crefname{prop}{Proposition}{Propositions}
\crefname{exa}{Example}{Examples}
\crefname{thm}{Theorem}{Theorems}
\crefrangeformat{line}{lines~#3#1#4--#5#2#6}
\crefrangeformat{figure}{#3Figures~#1#4--#5#2#6}
\crefrangeformat{defi}{#3Definitions~#1#4--#5#2#6}
\crefrangeformat{lem}{#3Lemmas~#1#4--#5#2#6}
\crefrangeformat{prop}{#3Propositions~#1#4--#5#2#6}
\crefrangeformat{exa}{#3Examples~#1#4--#5#2#6}
\crefrangeformat{algorithm}{#3Algorithms~#1#4--#5#2#6}
\crefname{threadprop}{requirement}{requirements}
\creflabelformat{threadprop}{#2(#1)#3}
\crefrangeformat{threadprop}{#3requirements~(#1)#4--#5(#2)#6}

\begin{document}

\title[Just Verification of Mutual Exclusion Algorithms]
      {Just Verification of Mutual Exclusion Algorithms with (Non-)Blocking and (Non-)Atomic Registers}
\titlecomment{This paper is an extension of \cite{GlabbeekLS25}.}

\author[R.J. van Glabbeek]{Rob {van Glabbeek}\lmcsorcid{0000-0003-4712-7423}}[a]
\address{School of Informatics, University of Edinburgh, UK\newline
School of Computer Science and Engineering, University of New South Wales, Sydney, Australia}
\thanks{Supported by Royal Society Wolfson Fellowship RSWF\textbackslash R1\textbackslash 221008}
\email{rvg@cs.stanford.edu}

\author[B. Luttik]{Bas Luttik\lmcsorcid{0000-0001-6710-8436}}[b]

\address{Eindhoven University of Technology, The Netherlands}
\email{s.p.luttik@tue.nl, m.s.c.spronck@tue.nl}

\author[M.S.C. Spronck]{Myrthe Spronck\rsuper*\lmcsorcid{0000-0003-2909-7515}}[b]

\begin{abstract}
  We verify the correctness of a variety of mutual exclusion algorithms through model checking.
    We look at algorithms where communication is via shared read/write registers, where those registers can be atomic or non-atomic.
    For the verification of liveness properties, it is necessary to assume a completeness criterion to eliminate spurious counterexamples. 
    We use justness as completeness criterion.
    Justness depends on a concurrency relation; we consider several such relations, modelling different assumptions on the working of the shared registers.
    We present executions demonstrating the violation of correctness properties by several algorithms, and in some cases suggest improvements.
\end{abstract}

\maketitle
\renewcommand{\thefootnote}{\fnsymbol{footnote}}
\footnotetext[1]{Corresponding author.\\\mbox{}}
\renewcommand{\thefootnote}{\arabic{footnote}}
\setcounter{footnote}{0}

\section{Introduction}\label{sec:introduction}

The mutual exclusion problem is a fundamental problem in concurrent programming.
Given $N\geq 2$ \emph{threads},\footnote{What we call threads are in the literature frequently referred to as \emph{processes} or \emph{computers}. We use \emph{threads} to distinguish between the real systems and our models of them, expressed in a process algebra.} each of which may occasionally wish to access a \emph{critical section}, a \emph{mutual exclusion algorithm} seeks to ensure that at most one thread accesses its critical section at any given time. Ideally, this is done in such a way that whenever a thread wishes to access its critical section, it eventually succeeds in doing so.
Many mutual exclusion algorithms have been proposed in the literature, and in general their correctness depends on assumptions one can make on the environment in which these algorithms will be running. The present paper aims to make these assumptions explicit, and to verify the correctness of some of the most popular mutual exclusion algorithms as a function of these assumptions.

\paragraph{Correctness properties of mutual exclusion algorithms.}

A thread that does not seek to execute its critical section is said to be executing its \emph{non-critical section}. We regard \emph{leaving the non-critical section} as getting the desire to enter the critical section. After this happens, the thread is executing its \emph{entry protocol}, the part of the mutual exclusion algorithm in which it negotiates with other threads who gets to enter the critical section first. The critical section occurs right after the entry protocol, and is followed by an \emph{exit protocol}, after which the thread returns to its non-critical section.
When in its non-critical section, a thread is not expected to communicate with the other threads for the purpose of negotiating access to the critical section. Moreover, a thread may choose to remain in its non-critical section forever after. However, once a thread gains access to its critical section, it must leave it within a finite time, so as to make space for other threads.

The most crucial correctness property of a mutual exclusion algorithm is \emph{mutual exclusion}: at any given time, at most one thread will be in its critical section.
This is a safety property. In addition, a hierarchy of liveness properties have been considered. The weakest one is \emph{deadlock freedom}: 
Whenever at least one thread is running its entry protocol, eventually some thread will enter its critical section. This need not be one of the threads that was observed to be in its entry protocol.
A stronger property is \emph{starvation freedom}: whenever a thread leaves its non-critical section, it will eventually enter its critical section.
A yet stronger property, called \emph{bounded bypass}, augments starvation freedom with a bound on the number of times other threads can gain access to the critical section before any given thread in its entry protocol.

In this paper we check for over a dozen mutual exclusion protocols, and for six possible assumptions on the environment in which they are running, whether they satisfy mutual exclusion, deadlock freedom and starvation freedom.
We will not investigate bounded bypass, nor other desirable properties of mutual exclusion protocols, such as \emph{first-come-first-served}, \emph{shutdown safety},  \emph{abortion safety}, \emph{fail safety} and \emph{self-stabilisation} \cite{Lamport86Mutex2}.

\paragraph{Memory models.}\hspace{-5pt}\footnote{A \href{https://en.wikipedia.org/wiki/Memory_model_(programming)}{\emph{memory model}} describes the interactions of threads through memory and their shared use of the data. 
The models reviewed here differ in the degree in which different register accesses exclude each other, and in what values a register may return in case of overlapping reads and writes. In this paper, we do not consider \emph{weak memory models}, e.g., those that allow for compiler optimisations or for reads to sometimes fetch values that were already changed by another thread. In \cite{AttiyaGHKMV11} it has been shown that mutual exclusion cannot be realised in weak memory models, unless those models come with \emph{memory fences} or \emph{barriers} that can be used to undermine their weak nature.
}\hspace{6pt}
In the mutual exclusion algorithms considered here, the threads communicate with each other solely by reading from and writing to shared registers. The main assumptions on the environment in which mutual exclusion algorithms will be running concern these registers.
It is frequently assumed that (read and write) operations on registers are ``undividable'', meaning that they cannot overlap or interleave each other: if two threads attempt to perform an operation on the same register at the same time, one operation will be performed before the other. This assumption, sometimes referred to as \emph{atomicity}, is explicitly made in Dijkstra's first paper on mutual exclusion \cite{dijkstra65}.
Atomicity is sometimes conceptualised as operations occurring at a single moment in time. We instead acknowledge that operations have duration: there is a period of time between the \emph{invocation} (start) and \emph{response} (completion) of an operation.
Consequently, if operations cannot overlap in time, then, when multiple operations are attempted simultaneously, the one successfully invoked first must postpone the occurrence of the others until at least its response.
One operation postponing another is called \emph{blocking} \cite{CDV09}.

Deviating from Dijkstra's original presentation, several authors have considered a variation of the mutual exclusion problem where the atomicity assumption is dropped \cite{Lamport74,peterson1983new,Lamport86Mutex1,Lamport86Mutex2,Szy88,Szy90,anderson1993fine,aravind2010yet}.
Attempted operations can then occur immediately, without blocking each other. We say these operations are \emph{non-blocking}. In this context, read and write operations may be \emph{concurrent}, i.e.\ the intervals between their respective invocations and responses may overlap. 
We must then consider the consequences of operations overlapping each other. To address this, Lamport \cite{Lamport86IPCbasic,Lamport86IPCalg} proposes a hierarchy of three memory models, specifically for single-writer multi-reader (SWMR) registers; such registers are owned by one thread, and only that thread is capable of writing to them.
Crucial for these definitions is the assumption that every register has a domain, and a read of that register always yields a value from that domain. It is also important that threads can only perform a single operation at a time, meaning that a thread's operations can never overlap each other.
\begin{itemize}
\item A \textbf{safe} register guarantees merely that a read that is not concurrent with any write returns the most recently written value.
    \item A \textbf{regular} register guarantees that any read returns either the last value written before it started, or the value of any overlapping write, if there is one. 
    \item An \textbf{atomic} register guarantees that reads and writes behave as though they occur in some total order. This total order must comply with the real-time order of the operations: if operation $a$ has its response before operation $b$ is invoked, then $a$ must be ordered before~$b$.
\end{itemize}
These three memory models form a hierarchy, in the sense that any atomic register is regular, and any regular one is safe. 
When we merely know that a register is safe,
a read that overlaps with any write might return any value in the domain of the register.
In the cases of safe and regular registers, it is possible that two reads that overlap with the same write first return the new and then the old value of the register; this phenomenon is called \emph{new-old inversion}.

\begin{exa}\label{ex:SWMR}   
    Consider the execution illustrated in \cref{fig:SWMRexample}, which has three threads, with id's $0$, $1$ and $2$, operating on a register $x$ with domain $\{0,1,2\}$.
    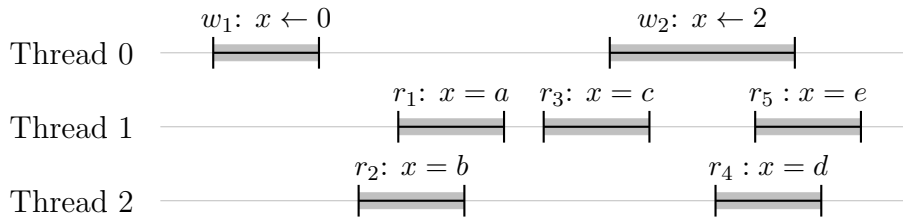
\begin{figure}[ht]
    \centering
    \begin{tikzpicture}[initial text=,inner sep=0pt, outer sep=0pt, minimum size=0pt, node distance=\distY pt and \distX pt]
        \node [draw=none] (0-name) {\large Thread $0$};

        \node [draw=none, right=of 0-name, xshift=\lineGap pt] (0-line-start) {};
        \node [draw=none, below=of 0-line-start] (1-line-start) {};
        \node [draw=none, below=of 1-line-start] (2-line-start) {};

        \node [draw=none, left=of 1-line-start, xshift=-\lineGap pt] (1-name) {\large Thread $1$};
        \node [draw=none, left=of 2-line-start, xshift=-\lineGap pt] (2-name) {\large Thread $2$};

        \node [draw=none, right=of 0-line-start, xshift=-\lineCut pt] (0-op1-s) {};
        \node [draw=none, right=of 0-op1-s] (0-op1-e) {};
        \node (0-op1-m) at ($(0-op1-s)!0.5!(0-op1-e)$) {};
        \node [draw=none, below=of 0-op1-e, xshift=\shiftXbig pt] (1-op2-s) {};
        \node [draw=none, right=of 1-op2-s] (1-op2-e) {};
        \node (1-op2-m) at ($(1-op2-s)!0.5!(1-op2-e)$) {};
        \node [draw=none, below=of 1-op2-s, xshift=-\shiftX pt] (2-op3-s) {};
        \node [draw=none, right=of 2-op3-s] (2-op3-e) {};
        \node (2-op3-m) at ($(2-op3-s)!0.5!(2-op3-e)$) {};
        \node [draw=none, above=of 2-op3-e, xshift= \shiftXbig pt] (1-op4-s) {};
        \node [draw=none, right=of 1-op4-s] (1-op4-e) {};
        \node (1-op4-m) at ($(1-op4-s)!0.5!(1-op4-e)$) {};
        \node [draw=none, above=of 1-op4-e, xshift=-\shiftX pt] (0-op5-s) {};
        \node [draw=none, right=of 0-op5-s, xshift=\shiftXbig pt] (0-op5-e) {}; 
        \node (0-op5-m) at ($(0-op5-s)!0.5!(0-op5-e)$) {};
        \node [draw=none, below=of 0-op5-e, xshift=-\shiftX pt] (1-op6-s) {};
        \node [draw=none, right=of 1-op6-s] (1-op6-e) {};
        \node (1-op6-m) at ($(1-op6-s)!0.5!(1-op6-e)$) {};
        \node [draw=none, below=of 1-op6-s, xshift=-\shiftX pt] (2-op7-s) {};
        \node [draw=none, right=of 2-op7-s] (2-op7-e) {};
        \node (2-op7-m) at ($(2-op7-s)!0.5!(2-op7-e)$) {};

        \node [draw=none, right=of 1-op6-e, xshift=-\lineCut pt] (1-line-end) {};
        \node [draw=none, above=of 1-line-end] (0-line-end) {};
        \node [draw=none, below=of 1-line-end] (2-line-end) {};

        \path[-, opacity=0.25]
            (0-line-start) edge (0-line-end)
            (1-line-start) edge (1-line-end)
            (2-line-start) edge (2-line-end);

        \draw[thick]
            \VBar{\barX pt}{(0-op1-s)}
            \VBar{\barX pt}{(0-op1-e)}
            \VBar{\barX pt}{(1-op2-s)}
            \VBar{\barX pt}{(1-op2-e)}
            \VBar{\barX pt}{(2-op3-s)}
            \VBar{\barX pt}{(2-op3-e)}
            \VBar{\barX pt}{(1-op4-s)}
            \VBar{\barX pt}{(1-op4-e)}
            \VBar{\barX pt}{(0-op5-s)}
            \VBar{\barX pt}{(0-op5-e)}
            \VBar{\barX pt}{(1-op6-s)}
            \VBar{\barX pt}{(1-op6-e)}
            \VBar{\barX pt}{(2-op7-s)}
            \VBar{\barX pt}{(2-op7-e)};

        \filldraw[opacity=0.25]
            \Above{\barY pt}{(0-op1-s)} rectangle \Below{\barY pt}{(0-op1-e)}
            \Above{\barY pt}{(1-op2-s)} rectangle \Below{\barY pt}{(1-op2-e)}
            \Above{\barY pt}{(2-op3-s)} rectangle \Below{\barY pt}{(2-op3-e)}
            \Above{\barY pt}{(1-op4-s)} rectangle \Below{\barY pt}{(1-op4-e)}
            \Above{\barY pt}{(0-op5-s)} rectangle \Below{\barY pt}{(0-op5-e)}
            \Above{\barY pt}{(1-op6-s)} rectangle \Below{\barY pt}{(1-op6-e)}
            \Above{\barY pt}{(2-op7-s)} rectangle \Below{\barY pt}{(2-op7-e)};

        \path[-,thick]
            (0-op1-s) edge (0-op1-e)        
            (1-op2-s) edge (1-op2-e)  
            (2-op3-s) edge (2-op3-e) 
            (1-op4-s) edge (1-op4-e)     
            (0-op5-s) edge (0-op5-e)    
            (1-op6-s) edge (1-op6-e)    
            (2-op7-s) edge (2-op7-e);      

        \draw[]
            (0-op1-m) node[above,yshift=\textY pt] {$w_1$: $x \writeop 0$}
            (1-op2-m) node[above,yshift=\textY pt] {$r_1$: $x = a$}
            (2-op3-m) node[above,yshift=\textY pt] {$r_2$: $x = b$}
            (1-op4-m) node[above,yshift=\textY pt] {$r_3$: $x = c$}
            (0-op5-m) node[above,yshift=\textY pt] {$w_2$: $x \writeop 2$}
            (1-op6-m) node[above,yshift=\textY pt] {$r_5: x = e$}
            (2-op7-m) node[above,yshift=\textY pt] {$r_4: x = d$};
    \end{tikzpicture}
    \caption{Example behaviour of SWMR safe, regular and atomic registers.}
    \label{fig:SWMRexample}
    \end{figure}

    Let us consider which values are possible for $a, b, c, d$ and $e$, depending on the type of register $x$:
    \begin{itemize}
        \item \textbf{Regardless of the type of $x$}, $a = 0$ and $b=0$. This is because both $r_1$ and $r_2$ are not concurrent with any write, hence both reads will return the last written value of the register, regardless of the type of $x$. That the reads overlap each other has no consequences.
        \item \textbf{If $x$ is a safe register}, then $c, d$ and $e$ may all be any value of $0, 1$ or $2$. This is because all three reads overlap with a write. These three assignments are independent of each other, so all 27 different possible combinations of choices for $c, d$ and $e$ are allowed.
        \item \textbf{If $x$ is a regular register}, then $c, d$ and $e$ may all be any value of $0$ or $2$: the old value or the value that is being written concurrently. Here too, the three returned values are independent, so there are 8 different possibilities. In particular, it can be that $c = 2$ but $e=0$; this is due to new-old inversion.
        \item \textbf{If $x$ is an atomic register}, then the values returned by the reads must reflect some order on all operations in the execution, and this order must respect the real-time ordering of the operations. Consequently, there are only five combinations of values of $c, d$ and $e$ allowed: if $c$ is $0$, then any combination of $d$ and $e$ being $0$ or $2$ is accepted, since $r_3$ can be ordered before $w_2$ without affecting whether $r_4$ and $r_5$ are before or after $w_2$. If instead we have $c = 2$, then $r_3$ must be ordered after $w_2$ and hence $r_4$ and $r_5$, which in real time come after $r_3$, must also be ordered after $r_2$, giving us $d=e=2$.  
    \end{itemize}
\end{exa}

In \cref{sec:mwmr-regs} we discuss the generalisation of these memory models to multi-writer multi-reader (MWMR) registers, ones that can be written and read by all threads.

Besides blocking and non-blocking registers, as explained above, we consider two intermediate memory models. The \emph{blocking model with concurrent reads} requires (1) any scheduled read or write to await the completion of any write that is in progress, and (2) any scheduled write to await the completion of any unfinished read. However, reads from different threads need not wait for each other and may overlap in time without ill effects. In the model of \emph{non-blocking reads},\footnote{In this terminology, from \cite{CDV09}, a \emph{blocking read} blocks a write; it does not refer to a read that is blocked.} we have (1) but not (2).
This model, where writes block reads but reads do not block writes, may apply when writes can abort in-progress reads, superseding them.

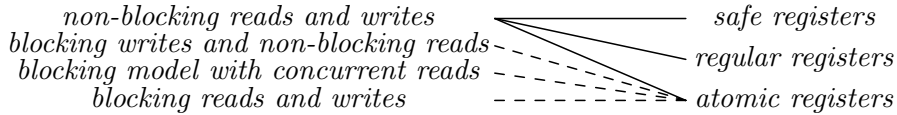
\expandafter\ifx\csname graph\endcsname\relax
   \csname newbox\expandafter\endcsname\csname graph\endcsname
\fi
\ifx\graphtemp\undefined
  \csname newdimen\endcsname\graphtemp
\fi
\expandafter\setbox\csname graph\endcsname
 =\vtop{\vskip 0pt\hbox{%
    \graphtemp=.5ex
    \advance\graphtemp by 0.429in
    \rlap{\kern 0.000in\lower\graphtemp\hbox to 0pt{\hss \emph{blocking reads and writes}\hss}}%
    \graphtemp=.5ex
    \advance\graphtemp by 0.286in
    \rlap{\kern 0.000in\lower\graphtemp\hbox to 0pt{\hss \emph{blocking model with concurrent reads}\hss}}%
    \graphtemp=.5ex
    \advance\graphtemp by 0.143in
    \rlap{\kern 0.000in\lower\graphtemp\hbox to 0pt{\hss \emph{blocking writes and non-blocking reads}\hss}}%
    \graphtemp=.5ex
    \advance\graphtemp by 0.000in
    \rlap{\kern 0.000in\lower\graphtemp\hbox to 0pt{\hss \emph{non-blocking reads and writes}\hss}}%
    \graphtemp=.5ex
    \advance\graphtemp by 0.429in
    \rlap{\kern 2.857in\lower\graphtemp\hbox to 0pt{\hss \emph{atomic registers}\hss}}%
    \graphtemp=.5ex
    \advance\graphtemp by 0.214in
    \rlap{\kern 2.857in\lower\graphtemp\hbox to 0pt{\hss \emph{regular registers}\hss}}%
    \graphtemp=.5ex
    \advance\graphtemp by 0.000in
    \rlap{\kern 2.857in\lower\graphtemp\hbox to 0pt{\hss \emph{safe registers}\hss}}%
\pdfliteral{
q [] 0 d 1 J 1 j
0.576 w
0.576 w
q [3.6 4.] 0 d
92.592 -30.888 m
164.592 -30.888 l
S Q
q [3.6 4.081382] 0 d
92.592 -20.592 m
164.592 -30.888 l
S Q
q [3.6 3.528678] 0 d
92.592 -10.296 m
164.592 -30.888 l
S Q
92.592 0 m
164.592 -30.888 l
S
92.592 0 m
164.592 -15.408 l
S
92.592 0 m
164.592 0 l
S
Q
}%
    \hbox{\vrule depth0.429in width0pt height 0pt}%
    \kern 2.857in
  }%
}%

\centerline{\box\graph}
\vspace{1em}

We model six different memory models, which are illustrated above. The blocking aspect of our memory models is captured via different concurrency relations (\cref{sec:concurrency}). The distinction between safe, regular and atomic registers is captured via three different process-algebraic models (\cref{sec:modelling}). \hypertarget{justification}{Since the safe/regular/atomic distinction is only relevant in models that allow writes to overlap reads and writes, we only make it for the non-blocking model; for the other three memory models we reuse our atomic register models}.

\paragraph{Completeness criteria.}

In previous work \cite{spronck2023process}, we checked the mutual exclusion property of several algorithms, with safe, regular and atomic MWMR registers, through model checking with the mCRL2 toolset~\cite{mCRL2toolset}.
We did not check the liveness properties at that time; the presence of certain infinite loops in our models introduced spurious counterexamples to such properties, which hindered our verification efforts.
As an example of what we call a ``spurious counterexample'', we frequently found violations to starvation freedom where one thread, $i$, never obtained access to its critical section because a different thread, $j$, was endlessly repeating a busy wait, or some other infinite cycle which should reasonably not prevent $i$ from progressing to its critical section.
Yet, the model checker does not know this, and can therefore only conclude that the property is not satisfied.
In this paper, we extend our previous work by addressing this problem and checking liveness properties as well.

One method for discarding spurious counterexamples from verification results is applying completeness criteria: rules for determining which paths in the model represent real executions of the modelled system.
By ensuring that all spurious paths are classified as incomplete and only taking complete paths into consideration when verifying liveness properties, we can circumvent the spurious counterexamples.
Of course, one must take care not to discard true system executions by classifying those as incomplete.
The completeness criterion must therefore be chosen with care.
Examples of well-known completeness criteria are weak fairness and strong fairness.
Weak fairness assumes that every task\footnote{What constitutes a \emph{task} differs from paper to paper; hence there are multiple flavours of strong and weak fairness; here a task could be a read or write action of a certain thread on a certain register.} that eventually is perpetually enabled must occur infinitely often; strong fairness assumes that if a task is infinitely often enabled it must occur infinitely often \cite{lehmann1981impartiality,apt1983proof,glabbeek2019progress}.
In effect, making a fairness assumption amounts to assuming that if something is tried often enough, it will always eventually succeed \cite{glabbeek2019progress}.
In that sense, these assumptions, even weak fairness, are rather strong, and may well result in true system executions being classified as incomplete.
In this paper, we therefore use the weaker completeness criterion \emph{justness} \cite{glabbeek2019progress,glabbeek2019justness,bouwman2020off}.

Unlike weak and strong fairness, justness takes into account how different actions in the model relate to each other. 
Informally, it says that if an action $a$ can occur, then eventually $a$ occurs itself, or a different action occurs that interferes with the occurrence of $a$.
The underlying idea
of justness 
is that the different components that make up a system must all be capable of making progress: if thread $i$ wants to perform an action entirely independent of the actions performed by thread $j$, then there can be no interference.
However, if both threads are interacting with a shared register, then we may decide that one thread writing to the register can prevent the other from reading it at the same time, or vice versa.
Which actions interfere with each other is a modelling decision, dependent on our understanding of the real underlying system.
It is formalised through a \emph{concurrency relation}, which must adhere to some restrictions. In this paper we propose four concurrency relations, each modelling one of the four major memory models reviewed above: 
non-blocking reads and writes, blocking writes and non-blocking reads, the blocking model with concurrent reads, and blocking reads and writes.

\paragraph{Model checking.}
Traditionally, mutual exclusion algorithms have been verified by pen-and-paper proofs using behavioural reasoning. As remarked by Lamport \cite{Lamport86Mutex2},
``the behavioral reasoning used in our correctness proofs, and in most other published correctness proofs of concurrent algorithms, is inherently unreliable''.  This is especially the case when dealing with the intricacies of non-atomic registers.\footnote{A good illustration of unreliable behavioural reasoning is given in \cite[Section 21]{glabbeek2023modelling}, through a short but fallacious argument that the mutual exclusion property of Peterson's mutual exclusion protocol, which is known to hold for atomic registers, would also hold for safe registers. We challenge the reader to find the fallacy in this argument before looking at the solution.} 
This problem can be alleviated by automated formal verification; here we employ model checking.

While the precise modelling of the algorithms, the registers and the employed completeness criterion requires great care, the subsequent verification requires a mere button-push and some patience.
Since our model checker traverses the entire state-space of a protocol, the verified protocols and all their registers need to be finite.\footnote{Using symbolic tools~\cite{neele2018bakery} might alleviate this problem; we did not do this here.}
 This prevented us from checking the bakery algorithm \cite{Lamport74}, as it is one of the few mutual exclusion protocols that employs an unbounded state space. Moreover, those algorithms that work for $N$ threads, for any $N\mathbin\in\mathbbm{N}$, could be checked for small values of $N$ only; in this paper we take $N\mathbin=3$. Consequently, any failure of a correctness property that shows up only for $>3$ threads will not be caught here.

As stated, we employed these methods in previous work to check mutual exclusion algorithms.
Although there we checked only safety properties, and did not consider the blocking aspects of memory, this already gave interesting results.
For instance, we showed that Szymanski's flag algorithm from \cite{Szy88}, even when adapted to use Booleans, violates mutual exclusion with non-atomic registers.
Here, we expand this previous work by checking deadlock freedom and starvation freedom in addition to mutual exclusion, and by including blocking into our memory models.
In total, we check the three correctness properties of over a dozen mutual exclusion algorithms, for six different memory models.
Among others, we cover Aravind's BLRU algorithm \cite{aravind2010yet}, Dekker's algorithm \cite{dijkstra1962over,alagarsamy2003some} and its RW-safe variant \cite{buhr2016dekker}, and Szymanski's 3-bit linear wait and 4-bit FCFS robust algorithms \cite{Szy90}.
In some cases where we find property violations, we suggest fixes to the algorithms so that the properties are satisfied.

\section{Preliminaries}\label{sec:preliminaries}

A \emph{labelled transition system} (LTS) is a tuple $(\states, \actionset, \initstate, \transrel)$ in which $\states$ is a finite set of states, $\actionset$ is a finite set of actions, $\initstate \in \states$ is the initial state, and $\transrel \subseteq \states \times \actionset \times \states$ is a transition relation.
We write $s \xrightarrow{a} s'$ for $(s, a, s') \in \transrel$.
We say an action $a$ is \emph{enabled} in a state $s$ if there exists a state $s'$ such that $(s, a, s') \in \transrel$.

A \emph{path} $\pi$ is a non-empty, potentially infinite alternating sequence of states and actions
$s_0 a_1 s_{1} a_2 \ldots$, with $s_0, s_1, \ldots \in \states$ and $a_1, a_2, \ldots \in \actionset$,
such that if $\pi$ is finite, then its last element is a state, and for all relevant $i \in \mathbbm{N}$, $s_i \xrightarrow{a_{i+1}} s_{i+1}$. 
The first state of $\pi$ is its \emph{initial state}.
The \emph{length} of $\pi$ is the number of transitions in it.

\hypertarget{parallel}{We use a notion of parallel composition that is taken from Hoare's CSP \cite{Ho85}, where synchronisation between components is enforced on all shared actions. 
It is defined as follows:
    For some $k\geq 1$, let $P_1, \ldots, P_k$ be LTSs, where $P_i = (\states_i, \actionset_i, \initstate_i, \transrel_i)$ for all $1 \leq i \leq k$.
    The \emph{parallel composition} $P_1 \parcomp \ldots \parcomp P_k$ of $P_1, \ldots, P_k$ is the LTS $P = (\states, \actionset, \initstate, \transrel)$ in which $\states = \states_1 \times \ldots \times \states_k$, $\actionset = \bigcup_{1 \leq i \leq k}\actionset_i$, $\initstate = (\initstate_1, \ldots, \initstate_k)$, and a transition $((s_1, \ldots, s_k), a, (s_1', \ldots, s_k'))$ is in $\transrel$ if, and only if, $a \in \actionset$ and the following are true for all $1 \leq i \leq k$:
                   if $a \notin \actionset_i$, then $s_i = s_i'$, and
                   if $a \in \actionset_i$, then $(s_i, a, s_i') \in \transrel_i$.}

Note that by this definition of $\transrel$, if an action is in the action set of a component but not enabled by that component in a particular state of the parallel composition, then the composition cannot perform a transition labelled with that action.

As mentioned in the introduction, the completeness criterion we use for our liveness verification is a variant of justness \cite{glabbeek2019justness,glabbeek2019progress}.
Specifically, while justness is originally defined on transitions, we here define it on action labels, an adaptation we take from \cite{bouwman2020off}. 
As stated earlier, the definition of justness relies on the notion of a concurrency relation.

\hypertarget{second}{
\begin{defi}\rm\label{def:conc}
    Given an LTS $(\states, \actionset, \initstate, \transrel)$ and an equivalence relation $\actequivsym$ on $\actionset$, a relation $\conc \subseteq \actionset \times \actionset$ is a \emph{concurrency relation} if, and only if:
    \begin{itemize}
        \item $\actequivsym$ and $\conc$ are disjoint; in particular, $\conc$ is irreflexive.
        \item For all $a \in \actionset$, if $\pi$ is a path from a state $s \in \states$ to a state $s' \in \states$ such that $a$ is enabled in $s$ and $a \conc b$ for all $b \in \actionset$ occurring on $\pi$, then there exists an action $c \in \actionset$ such that $\actequiv{a}{c}$ and $c$ is enabled in $s'$.
    \end{itemize}
\end{defi}}\vspace{1ex}\noindent
A concurrency relation may be asymmetric. We often reason about the complement of $\conc$, $\nconc$. Read $a \conc b$ as ``$a$ is independent from $b$'' and $a \nconc b$ as ``$b$ interferes with/postpones $a$''. 
The second condition can be understood as: if an action $a$ is enabled in some state, and then some actions occur that do not interfere with $a$, then afterwards an action equivalent with $a$ should still be enabled.

\begin{obs}\rm\label{obs:subset}
Concurrency relations can be refined by removing pairs; a subset of a concurrency relation is still a concurrency relation.
\end{obs}

Informally, justness says that a path is complete if whenever an action $a$ is enabled along the path, there is eventually an occurrence of an action (possibly $a$ itself) that interferes with it.
This can be weakened by defining a set of \emph{blockable actions}, for which this restriction does not hold; a blockable action may be enabled on a complete path without there being a subsequent occurrence of an interfering action.
In this paper, the action of a thread to leave its non-critical section will be blockable. This way we model that a thread may choose to never take that option.
We give the formal definition of justness, incorporating the blockable actions. We represent the set of blockable actions as $\block$. Its complement, $\nonblock$, is defined as $\actionset \setminus \block$, given a set of actions $\actionset$.

\hypertarget{just}{
\begin{defi}\rm\label{def:justness}
    A path $\pi$ in an LTS $(\states, \actionset, \initstate, \transrel)$
    satisfies $\block$-$\conc$-\emph{justness of actions} ($\justact{\conc}{\block}$) if, and only if, for each suffix $\pi'$ of $\pi$, if an action $a \in \nonblock$ is enabled in the initial state of $\pi'$, then an action $b \in \actionset$ occurs in $\pi'$ such that $a \nconc b$.
\end{defi}}\noindent
We say that a property is satisfied on a model under $\justact{\conc}{\block}$ if it is satisfied on every path of that model, starting from the model's initial state, that satisfies $\justact{\conc}{\block}$. If $\block$ and $\conc$ are clear from the context, we simply say that a path that satisfies $\justact{\conc}{\block}$ is \emph{just}.

\section{MWMR registers}\label{sec:mwmr-regs}
In \cite{spronck2023process}, we presented process-algebraic models of MWMR safe, regular and atomic registers, which through the semantics of the process algebra determine an LTS for each register of a given kind. 
To this end, we also had to define the behaviour of an MWMR register of each type.
We use those same MWMR behaviours in this paper, which we present in this section, supported by a running example. Our modelling decisions for MWMR registers will be discussed in \cref{sec:modelling}.
\begin{exa}\label{ex:MWMR}
    The definitions of MWMR registers we use all reduce to the SWMR case when writes do not overlap. Hence, \cref{ex:SWMR} is still applicable.
    For this running example, we focus on the behaviour when two writes overlap, see \cref{fig:MWMRexample}.
    We still consider the case of three threads with id's $0$, $1$ and $2$, and a register $x$ with domain $\{0,1,2\}$.
    \begin{figure}[ht!]
    \centering
    \begin{tikzpicture}[initial text=,inner sep=0pt, outer sep=0pt, minimum size=0pt, node distance=\distY pt and \distX pt]
        \node [draw=none] (0-name) {\large Thread $0$};

        \node [draw=none, right=of 0-name, xshift=\lineGap pt] (0-line-start) {};
        \node [draw=none, below=of 0-line-start] (1-line-start) {};
        \node [draw=none, below=of 1-line-start] (2-line-start) {};

        \node [draw=none, left=of 1-line-start, xshift=-\lineGap pt] (1-name) {\large Thread $1$};
        \node [draw=none, left=of 2-line-start, xshift=-\lineGap pt] (2-name) {\large Thread $2$};

        \node [draw=none, right=of 0-line-start, xshift=-\lineCut pt] (0-op1-s) {};
        \node [draw=none, right=of 0-op1-s,xshift=-\shiftX pt] (0-op1-e) {};
        \node (0-op1-m) at ($(0-op1-s)!0.5!(0-op1-e)$) {};
        \node [draw=none, below=of 0-op1-e, xshift=\shiftXbig*1.8 pt] (1-op2-s) {};
        \node [draw=none, right=of 1-op2-s, xshift=\shiftX pt] (1-op2-e) {};
        \node (1-op2-m) at ($(1-op2-s)!0.5!(1-op2-e)$) {};
        \node [draw=none, below=of 1-op2-s, xshift=-\shiftX*1.8 pt] (2-op3-s) {};
        \node [draw=none, right=of 2-op3-s] (2-op3-e) {};
        \node (2-op3-m) at ($(2-op3-s)!0.5!(2-op3-e)$) {};
        \node [draw=none, above=of 2-op3-e, xshift=\shiftXbig*2.4 pt] (1-op4-s) {};
        \node [draw=none, right=of 1-op4-s] (1-op4-e) {};
        \node (1-op4-m) at ($(1-op4-s)!0.5!(1-op4-e)$) {};
        \node [draw=none, above=of 1-op4-e, xshift=-\shiftXbig*3.3 pt] (0-op5-s) {};
        \node [draw=none, right=of 0-op5-s] (0-op5-e) {}; 
        \node (0-op5-m) at ($(0-op5-s)!0.5!(0-op5-e)$) {};
        \node [draw=none, right=of 0-op5-e, xshift=\shiftXbig pt] (0-op6-s) {};
        \node [draw=none, right=of 0-op6-s] (0-op6-e) {};
        \node (0-op6-m) at ($(0-op6-s)!0.5!(0-op6-e)$) {};
        \node [draw=none, right=of 2-op3-e, xshift=-\shiftX*1 pt] (2-op7-s) {};
        \node [draw=none, right=of 2-op7-s, xshift=\shiftX*1.5 pt] (2-op7-e) {};
        \node (2-op7-m) at ($(2-op7-s)!0.5!(2-op7-e)$) {};
        \node [draw=none, right=of 2-op7-e, xshift=\shiftXbig*0.8 pt] (2-op8-s) {};
        \node [draw=none, right=of 2-op8-s] (2-op8-e) {};
        \node (2-op8-m) at ($(2-op8-s)!0.5!(2-op8-e)$) {};

        \node [draw=none, right=of 2-op8-e, xshift=-\lineCut pt] (2-line-end) {};
        \node [draw=none, above=of 2-line-end] (1-line-end) {};
        \node [draw=none, above=of 1-line-end] (0-line-end) {};

        \path[-, opacity=0.25]
            (0-line-start) edge (0-line-end)
            (1-line-start) edge (1-line-end)
            (2-line-start) edge (2-line-end);

        \draw[thick]
            \VBar{\barX pt}{(0-op1-s)}
            \VBar{\barX pt}{(0-op1-e)}
            \VBar{\barX pt}{(1-op2-s)}
            \VBar{\barX pt}{(1-op2-e)}
            \VBar{\barX pt}{(2-op3-s)}
            \VBar{\barX pt}{(2-op3-e)}
            \VBar{\barX pt}{(1-op4-s)}
            \VBar{\barX pt}{(1-op4-e)}
            \VBar{\barX pt}{(0-op5-s)}
            \VBar{\barX pt}{(0-op5-e)}
            \VBar{\barX pt}{(0-op6-s)}
            \VBar{\barX pt}{(0-op6-e)}
            \VBar{\barX pt}{(2-op7-s)}
            \VBar{\barX pt}{(2-op7-e)}
            \VBar{\barX pt}{(2-op8-s)}
            \VBar{\barX pt}{(2-op8-e)};

        \filldraw[opacity=0.25]
            \Above{\barY pt}{(0-op1-s)} rectangle \Below{\barY pt}{(0-op1-e)}
            \Above{\barY pt}{(1-op2-s)} rectangle \Below{\barY pt}{(1-op2-e)}
            \Above{\barY pt}{(2-op3-s)} rectangle \Below{\barY pt}{(2-op3-e)}
            \Above{\barY pt}{(1-op4-s)} rectangle \Below{\barY pt}{(1-op4-e)}
            \Above{\barY pt}{(0-op5-s)} rectangle \Below{\barY pt}{(0-op5-e)}
            \Above{\barY pt}{(0-op6-s)} rectangle \Below{\barY pt}{(0-op6-e)}
            \Above{\barY pt}{(2-op7-s)} rectangle \Below{\barY pt}{(2-op7-e)}
            \Above{\barY pt}{(2-op8-s)} rectangle \Below{\barY pt}{(2-op8-e)};

        \path[-,thick]
            (0-op1-s) edge (0-op1-e)        
            (1-op2-s) edge (1-op2-e)  
            (2-op3-s) edge (2-op3-e) 
            (1-op4-s) edge (1-op4-e)     
            (0-op5-s) edge (0-op5-e)    
            (0-op6-s) edge (0-op6-e)    
            (2-op7-s) edge (2-op7-e)
            (2-op8-s) edge (2-op8-e);      

        \draw[]
            (0-op1-m) node[above,yshift=\textY pt] {$w_1$: $x \writeop 0$}
            (1-op2-m) node[above,yshift=\textY pt] {$w_2$: $x \writeop 1$}
            (2-op3-m) node[above,yshift=\textY pt] {$r_1$: $x = a$}
            (1-op4-m) node[above,yshift=\textY pt] {$r_3$: $x = c$}
            (0-op5-m) node[above,yshift=\textY pt] {$r_2$: $x = b$}
            (0-op6-m) node[above,yshift=\textY pt] {$r_4$: $x = d$}
            (2-op7-m) node[above,yshift=\textY pt] {$w_3$: $x \writeop 2$}
            (2-op8-m) node[above,yshift=\textY pt] {$r_5$: $x = e$};
    \end{tikzpicture}
    \caption{Running example to illustrate the behaviour of MWMR registers.}
    \label{fig:MWMRexample}
    \end{figure}
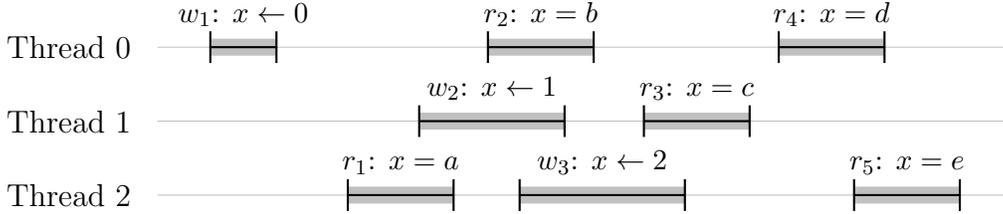
    Note that $r_1$ overlaps with $w_2$, $r_3$ overlaps with $w_3$ and $r_2$ overlaps with both. The other two reads do not overlap with writes, but they come after two writes that overlap each other: $w_2$ and $w_3$. Some of the reads overlap with each other, but this does not play a role in any of our definitions.
\end{exa}

\subsection{Safe MWMR registers}\label{sec:safe}
To extend the single-writer definition of safe registers to a multi-writer one, we follow Lamport in assuming that a concurrent read cannot affect the behaviour of a read or write. Lamport's SWMR definitions consider how concurrent writes affect reads, but not how concurrent writes affect writes.
Here we follow Raynal's approach for safe registers: a write that is concurrent with another write sets the value of the register to some arbitrary value in the domain of the register \cite{Raynal13}.
We can summarise the behaviour of MWMR safe registers with four rules:
\begin{enumerate}
    \item A read not concurrent with any writes on the same register returns the value most recently written into the register.
    \item A read concurrent with one or more writes on the same register returns an arbitrary value in the domain of the register.
    \item A write not concurrent with any other write on the same register results in the intended value being set.
    \item A write concurrent with one or more other writes on the same register results in an arbitrary value in the domain of the register being set.
\end{enumerate}

\begin{exa}\label{ex:MWMR-saf}
    Let us consider the possible return values $a, b, c, d$ and $e$ in \cref{ex:MWMR} when $x$ is a safe register.
    Note that $r_1$, $r_2$ and $r_3$ all overlap with at least one write, hence rule 2 applies and they can all return any value in the domain of the register, so in this case $a, b, c \in \{0,1,2\}$, and these values are independent. 
    The last two reads, $r_4$ and $r_5$, are not concurrent with any writes, so rule 1 applies and they return the most recently written value. However, the most recent write to complete was $w_3$, which overlapped with $w_2$. Thus, rule 4 applies to $w_3$ and so it writes an arbitrary value into the register. Consequently, $d, e \in \{0,1,2\}$ but it must be that $d = e$.
\end{exa}

\subsection{Regular MWMR registers}\label{sec:regular}

We wish to define regular MWMR registers as an extension of Lamport's definition of SWMR regular registers: a read returns either the last written value before the read began, or the value of any concurrent write, if there is one. The extension of this idea to the MWMR case is non-trivial; in \cite{spronck2023process} we present one extension and compare it to four different suggestions from \cite{Shao11}. The complexity comes from determining what the last written value is, given that writes may be concurrent. 
We here follow the proposal of \cite{spronck2023process}, where the reads must all agree on the ordering of the writes, but may still return the value of any write they overlap with.
For the purpose of determining the last written value, the register effectively behaves as though there is a global ordering on all operations, obtained by picking for each write action an arbitrary ordering moment between its invocation and response, while each read is ordered by its invocation.
Given such a global ordering, the \emph{last written value} for a read is the value of the last write occurring before it according to this ordering.

\begin{exa}\label{ex:MWMR-reg}
    Let us consider the possible return values  $a, b, c, d$ and $e$ in \cref{ex:MWMR} when $x$ is a regular register.
    Regardless of the ordering used in a particular execution of the register, $r_1$ can return either $0$ or $1$: the only write that can be ordered before $r_1$ is $w_1$, and $r_1$ overlaps with $w_2$. Thus, $a \in \{0,1\}$.
    For $r_2$, it is a bit more complicated. If $w_2$ is ordered after $r_2$, then $r_2$ can return $0$, $1$ or $2$. 
   For then, since $w_3$'s invocation comes after $r_2$'s invocation, the last written value must come from $w_1$ and thus be $0$. The values of the two overlapping writes are also possible, so $b \in \{0,1,2\}$. However, if $w_2$ is ordered before $r_2$ then $r_2$ can only return $1$ or $2$, since $0$ is no longer the last written value. Then $b \in \{1,2\}$.
    The case of $r_3$ is similarly complex: it can always return $2$ because it overlaps with $w_3$, and it can never return $0$ since $w_2$ has its response before $r_3$'s invocation and must therefore be ordered before it. Hence, $0$ can never be the last written value from the perspective of $r_3$. However, whether $r_3$ can return $1$ or not depends on the ordering used: if $w_3$ is ordered after $w_2$ but before $r_3$, then $c = 2$; otherwise $1$ is the last written value for $r_3$ and so $c \in \{1,2\}$.
    For $a, b$ and $c$, the returned values are independent except to the extend that the sets of possibilities may be smaller based on the ordering. Since the ordering is global, they must agree on that.
    For $r_4$ and $r_5$, it holds that they can both only return the last written value and they must agree on what this last written value is. Dependent on the ordering of $w_2$ and $w_3$, we get $d = e = 1$ or $d = e= 2$.
\end{exa}
    
\subsection{Atomic MWMR registers}\label{sec:atomic}

Lamport's definition of SWMR atomic registers, namely that the register must behave as though reads and writes occur in some strict order, is directly applicable to the MWMR case.

\begin{exa}\label{ex:MWMR-ato}
    Let us consider the possible return values $a, b, c, d$ and $e$ in \cref{ex:MWMR} when $x$ is an atomic register.
    In this case, we must consider all possible ways of ordering the operations in \cref{fig:MWMRexample} such that the real-time order is respected.
    In the case of $r_1$, it may be ordered before or after $w_2$: if it is before then $ a = 0$, if it is after then $a = 1$.
    For $r_2$, it does not only matter how it is ordered with respect to $w_2$ and $w_3$ but also how the two writes are ordered with respect to each other: if $r_2$ is before both writes then $b = 0$, otherwise $b$ is the value of the last write before $r_2$. However, the order that is generated is total, so what $r_2$ sees must agree with what $r_1$ sees. I.e., if $w_2$ is ordered before $r_1$, then since $r_1$ completes before $r_2$ begins, $w_2$ must also be ordered before $r_2$.
    Similarly, $r_3$, $r_4$ and $r_5$ will return $1$ or $2$ depending on how $w_2$ and $w_3$ are ordered, but this must agree with what $r_1$ and $r_2$ see. For example: if $a=1$ because $w_2$ is ordered before $r_1$ and $b = 2$ because $w_3$ is ordered before $r_2$, then because $r_1$ must be ordered before $w_3$ and $r_2$ must be ordered before $r_3$, it must be the case that $c=2$ as well. The return value of $r_3$ can differ from $r_4$ and $r_5$ because $r_3$ can be ordered before $w_3$, but the last two reads will always see the same value, so $d = e$.    
\end{exa}

\subsection{Threads}\label{sec:threads}

When we turn these definitions into models, as we will describe in \cref{sec:modelling}, we must combine models of the registers with models of threads executing the appropriate algorithm. 
Unlike the registers, which are mostly independent of the algorithm that we wish to analyse with them (the algorithm merely dictates the register's identifier, domain and initial value), threads are fully algorithm-dependent.
Yet, we do enforce several basic requirements on how the threads may interact with each other and with the registers:
\begin{enumerate}
    \item Threads may only communicate with each other by reading from and writing to shared registers.\label[threadprop]{thread-prop-1}
    \item Whenever a thread invokes a read, the only thing it does immediately afterwards is await the response of that read. A thread must always be willing to read all values in the domain of the register it is reading from. \label[threadprop]{thread-prop-2}
    \item Whenever a thread invokes a write, the only thing it does immediately afterwards is await the response of that write. A thread may never attempt to write a value outside of the domain of the register it is writing to. \label[threadprop]{thread-prop-3}
    \item A thread can only await the response of an operation if it has invoked that operation immediately before.\label[threadprop]{thread-prop-4}
\end{enumerate}

\section{Concurrency relations}\label{sec:concurrency}

Recall that we wish to use different concurrency relations to capture the models of blocking described in \cref{sec:introduction}.
By defining which operations may interfere with each other, we represent whether it is reasonable for one thread's operations on a register to interfere with (and thereby postpone) another thread's operations on that same register.
Since there are four models of blocking, we use four concurrency relations. These relations should ultimately be defined on the actions in the model, with a suitable equivalence relation, but we here present the idea of them in terms of operations.
The translation to actions, and the proof that the resulting concurrency relations are valid for our models (with respect to \cref{def:conc}), will be given in \cref{sec:justness}.

For the definition of our concurrency relations, we introduce the set $O$ of all operations performed by threads and registers; this includes read and write operations that are used for communication, but also any internal operations that threads and registers may perform.
Additionally, we introduce the following predicates and mappings on operations: $\isread{o}$ returns true if $o$ is a read operation and false otherwise; $\iswrite{o}$ returns true if $o$ is a write operation and false otherwise; $\thrmap{o}$ returns the id of the thread involved with operation $o$, or $\bot$ if no such thread exists; and $\regmap{o}$ returns the id of the register involved with operation $o$, or $\bot$ if no such register exists.
In practice, registers never do anything that does not in some way, even indirectly, involve a thread, so $\thrsym$ never maps to $\bot$.

The \emph{thread interference relation}, $\conc_T$, expresses that every operation is independent from every other operation \emph{unless} they belong to the same thread; every two operations by the same thread interfere with each other. It captures the memory model with non-blocking reads and writes. This is the largest (w.r.t.\ inclusion) concurrency relation  we will use.

\begin{defi}\rm\label{def:concT}
    $\conc_T = \{(a, b) \mid a,b \in O, \thrmap{a} \neq \thrmap{b}\}$
\end{defi}

The model with blocking writes and non-blocking reads is captured by the \emph{signalling reads relation}, $\conc_S$.
\begin{defi}\rm\label{def:concS}
    $\conc_S = \conc_T \setminus \{(a, b) \mid a,b \in O, \isread{a} \lor \iswrite{a},  \iswrite{b}, \regmap{a} = \regmap{b} \}$
\end{defi}
\noindent
Intuitively, this is the same as $\conc_T$ except that one thread writing to a register can interfere with a write to or a read from that same register by another thread. However, a read cannot interfere with another thread's read or write.
This concurrency relation has a precedent in \cite{glabbeeksignals,bouwman2020off}. There, reads are modelled as \emph{signals}, which cannot block anything. Hence the name of this relation.

The blocking model with concurrent reads is captured by the \emph{interfering reads relation}, $\conc_I$.
This is a further refinement from $\conc_S$, where a read operation can interfere with a write operation on the same register, but cannot interfere with another read.

\begin{defi}\rm\label{def:concI}
    $\conc_I = \conc_S \setminus \{(a, b) \mid a,b \in O, \iswrite{a}, \isread{b}, \regmap{a} = \regmap{b} \}$
\end{defi}
\noindent
This goes with the idea that performing a write on a memory location can only be done when the memory is reserved: repeated reads can prevent the memory from being reserved for a write, but as long as there is no write all the reads can take place concurrently.

Finally, the model of blocking reads and writes is captured by the \emph{all interfering relation}, $\conc_{\!A}$, a refinement of $\conc_I$ where a read operation can also interfere with another read operation.
\begin{defi}\rm\label{def:concA}
    $\conc_{\!A}= \conc_I \setminus \{(a, b) \mid a,b \in O, \isread{a},  \isread{b}, \regmap{a} = \regmap{b}\}$
\end{defi}
\noindent
In this model, every operation on a register fully reserves that register for only that operation, and hence can prevent any other operation from taking place at the same time.

\section{Modelling}\label{sec:modelling}
Up until now, we have kept the discussion of the registers, threads and their interaction fairly high-level.
In order to do verifications, we need to translate the definitions from the previous sections to LTSs.
We explain this translation here; the full process-algebraic models are given in \cref{app:registers}.
We will first describe register models that stick relatively closely to the high-level descriptions given thus far. In these models, which we call \emph{full-read} (for: fully modelled reads) models, we explicitly represent invocations and responses of the read and write operations.
Subsequently, we describe models in which read operations are only modelled by a single action. These models, called \emph{instant-read} (for: reads are modelled as occurring instantaneously) models, are further removed from the theory, but by reducing the number of actions each operation takes we reduce the number of transitions in our models, and thus speed up the verifications significantly.
In \cref{app:complete-same}, we prove that the results we obtain from the instant-read models are also valid for the full-read models.

For our definitions, we presuppose two disjoint finite sets: $\TID$ of thread identifiers (thread id's) and $\RID$ of register identifiers (register id's).
Additionally, for every $r\mathbin\in\RID$ we reference the set $\Data$ of all data values that the register $r$ can hold.

\subsection{Full-read models}\label{sec:fullread}

Our full-read models are based on the process-algebraic models of MWMR registers from  \cite{spronck2023process}; we made only slight alterations to make them more efficient, improve their presentation and better facilitate the definition of concurrency relations.

The invocation of a read by thread $\tid \in \TID$ of register $\rid \in \RID$ is modelled by the \emph{start read} action $\startread[\tid,\rid]$. Since at the invocation no value is read yet, there is no value passed along with this action.
The response to this read, returning the value $\data \in \Data[\rid]$, is modelled by the \emph{finish read} action $\finishread[\tid,\rid]{\data}$. 
Similarly, there are \emph{start write} and \emph{finish write} actions $\startwrite[\tid,\rid]{\data}$ and $\finishwrite[\tid,\rid]$.
Together, these actions form the \emph{interface actions} of the registers, which threads use to communicate with the registers.

Additional actions may exists in the models. Actions that belong solely to threads are called \emph{thread-local actions}, and actions that belong solely to registers are called \emph{register-local actions}. Such actions represent internal activity and cannot be used for communication.

A register model requires some finite amount of memory to store a representation of relevant past events. We store this in what we call the \emph{status object}, which features a finite set $\StatusAll$ of possible states. We abstract away from the exact implementation; for this presentation, all that is relevant is which information can be retrieved from it. Amongst others, we use the following \emph{access functions}, which are local to any given register $r$:
\begin{itemize}
    \item $\truevalsym: \StatusAll \rightarrow \Data$, the value that is currently stored in the register.
    \item $\writerssym: \StatusAll \rightarrow 2^{\TID}$, the set of thread id's of threads that have invoked a write operation on this register that has not yet had its response.
\end{itemize}
Any occurrence of a register action $a$ induces a state change $s \xrightarrow{a} s'$, resulting in an \emph{update} to these access functions. For instance, the actions $\startwrite[\tid,\rid]{\data}$ and $\finishwrite[\tid,\rid]$ cause the updates $\writers{s'} = \writers{s} \cup \{\tid\}$ and $\writers{s'} = \writers{s} \setminus \{\tid\}$, respectively.

We now briefly discuss how the three types of register are modelled.

\paragraph{Full-read safe registers.}
In our model of a safe register $r$, its status object maintains a Boolean variable $\overlapsym$ for each thread id, recording whether an ongoing read or write operation of this thread overlapped with a write by another thread. The value of $\overlapsym$ is updated in a straightforward way each time $r$ experiences a register interface action, aided by the access function $\writerssym$.
Using this function, our model can determine which of the four rules from \cref{sec:safe} applies when a read or write finishes, and behave accordingly.

\paragraph{Full-read regular registers.}
Recall that to determine for each read what the last written value is, a regular register needs to create an global ordering on the reads and writes.
Note that even though this ordering orders all operations, regular registers differ from atomic registers because reads can still return the values of all writes they are concurrent with; the ordering is only used to determine the last written value.
In our model, we generate this global ordering through the register-local \emph{order write} action $\orderwrite[\tid,\rid]$, which occurs at an arbitrary point between the start write action $\startwrite[\tid,\rid]{\data}$ and the finish write $\finishwrite[\tid,\rid]$.
The value $\truevalsym$ is updated with the value of a write when the order write action for that write operation happens, instead of at the finish write action. 
Hence, $\truevalsym$ contains the last written value according to the generated order.
While a read is active, so between its start and finish, the register model builds a set of the possible return values on the fly.
When the read starts, this set is initialised to $\truevalsym$ and the intended value of every active write.
Subsequently, whenever a write starts, its intended value is added to the set.
This way, at the finish read, the set will contain exactly those values that the read could return.

\paragraph{Full-read atomic registers}
For the atomic register model, we reuse the order write action from the regular register model, and add the similar \emph{order read} action $\orderread[\tid,\rid]$ for read operations.
This way, we generate an ordering on all operations. In our model, $\truevalsym$ is updated when the $\orderwrite[\tid,\rid]$ occurs, similar to regular registers. For read operations, the value of $\truevalsym$ when $\orderread[\tid,\rid]$ occurs is remembered, and returned at the matching response.

\subsection{Instant-read models}\label{sec:instantread}
In the full-read models we explicitly represent the duration of read and write operations through start and finish actions.
For verification purposes, however, it is sufficient to model a read operation with only a single action, even if this is a further abstraction from reality.
This was previously observed by Lamport for SWMR registers \cite{Lamport22deconstructing}. We now informally argue that it also holds for the MWMR registers; a formal proof can be found in \cref{app:complete-same}.
Whenever a read with duration returns a particular value $d$, there exists at least one instant between the invocation and response of that read such that if the read would occur at exactly that instant, it could also return $d$:
\begin{itemize}
    \item With safe registers, if there is a concurrent write during the read then if the read were to occur at any instant where that write is active, it could return any arbitrary value, including $d$. If there is no concurrent write during the read, then $d$ must be the stored value of the register at the start of the read, so the read may as well have occurred instantaneously at its invocation. 
    \item With regular registers, if $d$ is the value of a concurrent write, then the read may just as well have occurred at any instant where that write is active. If $d$ is instead the last written value when the read started, then the read may as well have occurred instantaneously at its invocation.
    \item With atomic registers, the register behaves as though all operations are totally ordered, so the read may occur instantaneously at the point where it was ordered.
\end{itemize}
Conversely, an instant read cannot return \emph{more} values than one with duration.
Thus, having a read occur instantaneously instead of giving it explicit duration has no impact on the possible range of values the read may return. 
Additionally, with all three register types the presence of a concurrent read has no effect on read or write operations, so other operations are also unaffected by modelling reads as occurring instantaneously.
Note that this is merely an abstraction for the purpose of modelling; in reality, read operations have duration.

To model reads as occurring instantaneously means replacing the actions $\startread[\tid,\rid]$, $\orderread[\tid,\rid]$ and $\finishread[\tid,\rid]{\data}$ from the full-read models with a single \emph{read} action: $\readop[\tid,\rid]{\data}$.
Additionally, we can simplify our status object, since less information needs to be tracked. 
For example, the safe register model no longer needs to use $\overlapsym$ to track whether there is a concurrent write during a read (though it must still track whether there is a concurrent write during a write); instead we can directly check $\writerssym$ when a read occurs. 
Similarly with the regular register model, we no longer need to track a set of possible return values for each active read that grows between its invocation and response; we can directly check the currently stored value and the values of all active writes when a read occurs, and return one of those values.
For the atomic model, the read action takes the place of the order read action.

\subsection{Thread-register models}\label{sec:thrregmodels}

To model a mutual exclusion algorithm, we must model both the threads executing the algorithm and the registers used for communication.
The behaviour of the threads will largely depend on the particular algorithm, but we do make some general assumptions on how threads are modelled. We presuppose the existence of an LTS $T_{\tid} = (\states_{\tid}, \actionset_{\tid}, \initstate_{\tid}, \transrel_{\tid})$ and a set of thread-local actions $\thrlocacts_{\tid}$ for every $\tid \mathbin\in \TID$ such that 
  $\actionset_{\tid} = \mathit{Interface}_{\tid} \cup \thrlocacts_{\tid}$.
Here $\mathit{Interface}_{\tid}$ depends on which type of register models the thread model is combined with:
For the full-read models, $\mathit{Interface}_{\tid} = \{\startread[\tid, \rid], \finishread[\tid, \rid]{\data}, \startwrite[\tid, \rid]{\data}, \finishwrite[\tid, \rid] \mid \rid \mathbin\in \RID, \data \mathbin\in \Data\}$, whereas for the instant-read models we have $\mathit{Interface}_{\tid} = \{\readop[\tid, \rid]{\data}, \startwrite[\tid, \rid]{\data}, \finishwrite[\tid, \rid] \mid \rid \mathbin\in \RID, \data \mathbin\in \Data\}$.
We assume that all sets of thread-local actions are pairwise disjoint; this enforces that threads can only communicate through interacting with the registers.
The other requirements listed in \cref{sec:threads} are also enforced.
When modelling a thread executing a particular algorithm, we ensure that the instant-read and full-read versions of the model are identical, besides using different actions for the read operations. In \cref{app:thread-transform}, we describe how a full-read thread LTS is transformed into an instant-read one.

We combine thread and register LTSs into \emph{thread-register models}.
This combination is possible as long as both the thread and register LTSs are of the same type, i.e.\ both full-read or both instant-read. We sometimes refer to these as full-read thread-register models and instant-read thread-register models.
For the following definition, we let $R_{\rid} = (\states_{\rid}, \actionset_{\rid}, \initstate_{\rid}, \transrel_{\rid})$ be the LTS associated with each $\rid \in \RID$.
\begin{defi}\rm\label{def:tr-model}
    A \emph{thread-register model} is a six-tuple $(\states, \actionset, \initstate, \transrel, \thrsym, \regsym)$, such that $(\states, \actionset, \initstate, \transrel)$ is a parallel composition of thread and register LTSs and
    \begin{itemize}
        \item $\thrsym \!: \actionset \rightarrow \TID$ is a mapping from actions to thread id's; and
        \item $\regsym \!: \actionset \rightarrow \RID \cup \{\undefsymb\}$ is a mapping from actions to register id's and the special value $\undefsymb \mathbin{\notin} \RID$.
    \end{itemize}
\noindent
$\thrsym$ and $\regsym$ are defined in the obvious way, e.g., 
$\thrmap{\startwrite[\tid,\rid]{\data}} = \tid$ and $\regmap{\startwrite[\tid,\rid]{\data}} = \rid$.
Crucially, for a thread-local action $a$, $\regmap{a} = \undefsymb$.
However, for the register-local actions, $\thrsym$ is mapped to a thread id, for example $\thrmap{\orderwrite[\tid,\rid]} = \tid$.
Note that these are the same $\thrsym$ and $\regsym$ we referred to in \cref{sec:concurrency}.
\end{defi}
By our construction of the thread and register LTSs, every action in $\actionset$ appears in at most two components of the parallel composition. Specifically, all thread-local actions appear only in their associated thread LTS, register-local actions appear in only their associated register LTS, and the interface actions appear in the associated thread and register LTS.

\subsection{Justness for thread-register models}\label{sec:justness}
In order to obtain a suitable notion of justness for our thread-register models, we need to choose both $\block$ and $\conc$. 
Which actions are blockable depends on the mutual exclusion problem; we address this in the next section.
For $\conc$, we apply the abstract concurrency relations we gave in \cref{sec:concurrency} to both the full-read and instant-read thread-register models.
To do this, we must adapt the definitions of $O$, $\thrsym$, $\regsym$, $\isreadsym$ and $\iswritesym$, and provide an equivalence relation $\actequivsym$ on actions.
For $O$, we use the actions $\actionset$ of the thread-register model we apply the concurrency relation to. For $\thrsym$ and $\regsym$, we use the mappings from \cref{def:tr-model} of the same name.
To determine $\isreadsym$ and $\iswritesym$, we must decide which actions in a model represent part of an  operation that may be interfered with.
Since all operations must eventually terminate, it would be incorrect for an operation to be invoked but never get its response because of blocking. Thus, we do not allow interference on finish actions or the register-local order actions.
Thus, only the start of an operation or the full operation itself (for the reads of the instant-read models) may be interfered with.
We therefore use the following definitions for the predicates:
\begin{align*}
\isread{a}  &= \exists_{\tid\in\TID,\rid\in\RID,\data\in\Data[\rid]}. (a=\startread[\tid,\rid] \lor a=\readop[\tid,\rid]{\data})\\
\iswrite{a} &= \exists_{\tid\in\TID,\tid\in\RID,\data\in\Data}. (a=\startwrite[\tid,\rid]{\data})
\end{align*}

The equivalence relation on actions we use is the identity relation $\mathit{id}$ for all actions except the instant-read model's read actions.
This is because, whenever an action $\readop[\tid,\rid]{\data}$ is enabled, the thread $\tid$ will continue attempting to read from this register until it is successful, but due to other operations on this register it may happen that the value $\data$ eventually cannot be read any more. Thus, we must allow reads of different values to be equivalent, as long as they are by the same thread and register. We expand on this argument in \cref{app:equivalence}.
We define $\actequivsym$ as $\mathit{id} \cup \{(\readop[\tid,\rid]{\data},\readop[\tid,\rid]{\data'}) \mid  \tid\in\TID, \, \rid \in \RID, \, \data,\data' \in \Data[\rid]\}$.

To establish that the resulting relations are indeed concurrency relations for these models (cf.\ \cref{def:conc}), we first establish a property of our models. We call this property \emph{thread-consistency}.
\begin{defi}\rm\label{def:thr-const}
    An LTS $(\states, \actionset, \initstate, \transrel)$ is \emph{thread-consistent} with respect to a mapping $\thrsym: \actionset \rightarrow \TID$ and an equivalence relation $\actequivsym$ on $\actionset$ if, and only if, for all states $s \in \states$, if an action $a \in \actionset$ is enabled in $s$ and there exists a transition \plat{$s \xrightarrow{b} s'$} for some $s' \in \states, \, b \in \actionset$ such that $\thrmap{a} \neq \thrmap{b}$, then in $s'$ an action $c$ is enabled such that $\actequiv{a}{c}$.
\end{defi}
\noindent
The subsequent lemmas apply to both types of thread-register model.
\begin{restatable}{lem}{thrconsist}\rm\label{lem:thr-consist}
    Let $M = (\states, \actionset, \initstate, \transrel, \thrsym, \regsym)$ be a thread-register model and $\actequivsym$ as above. 
    Then the LTS $(\states, \actionset, \initstate, \transrel)$ is thread-consistent with respect to $\thrsym$ and $\actequivsym$.
\end{restatable}
\noindent
This lemma is proven in \cref{app:thrconsist}.

For the subsequent lemmas, fix a thread-register model $M = (\states, \actionset, \initstate, \transrel, \thrsym, \regsym)$.%
\begin{restatable}{lem}{concTval}\rm\label{lem:concT-val}
   $\conc_T$ is a concurrency relation for $M$.
\end{restatable}
\noindent
This follows by a straightforward application of \cref{lem:thr-consist}.
The details are in \cref{app:concTproof}.

\begin{lem}\rm\label{lem:ai-val}
    $\conc_S$, $\conc_I$ and $\conc_{\!A}$ are concurrency relations for $M$.
\end{lem}
\begin{proof}
    This follows from \cref{lem:concT-val} and \cref{obs:subset}.
\end{proof}
\noindent
As stated \hyperlink{justification}{in the introduction}, we capture six memory models. We obtain three variants of the non-blocking model by combining the $\conc_T$ relation with the safe, regular and atomic register models. The remaining three memory models are represented by combining $\conc_S$, $\conc_I$ and $\conc_{\!A}$ with the atomic register model. 
In \cref{app:characterise}, we formally characterise just paths in both types of thread-register models for all six variants. 
In \cref{app:complete-same} we prove that using the instant-read and full-read models results in the same verification outcomes; we extend this in \cref{app:avoidoverlap} to prove that using the instant-read atomic register model, which allows overlapping writes, for the three memory models with blocking writes is sound for verification purposes.

\subsection{Modelling mutual exclusion algorithms}\label{sec:modelmutex}
Thus far, we have not touched on how we translate the pseudocode of the mutual exclusion algorithms to thread models. We here explain the crucial choices made.

The only operations on registers we allow are reading and writing.
Hence, more complicated statements that may have been present in the original presentation of an algorithm, such as compare-and-swap instructions, are converted into these primitive operations. A statement like ``\textbf{await} $\forall_{i \in \TID}: x(i)$'', where $x$ is some condition on a thread id $i$, is modelled as a recursive process that checks each thread id from smallest to largest, waiting for each until $x(i)$ is satisfied.
In general, ``\textbf{await}'' statements are modelled as busy waiting loops.

Besides the behaviour of the threads themselves, mutual exclusion algorithms also dictate the domains and initial values of the registers.
Where an algorithm does not specify the initial value of a register, we allow all possible initial values.

We use $\crit[\tid]$ and $\noncrit[\tid]$, with $\tid\in\TID$, as thread-local actions.
These actions represent a thread entering its critical section and leaving its non-critical section, respectively. We define $\block = \{\noncrit[\tid] \mid \tid \in \TID\}$ to capture that a thread may always choose to remain in its non-critical section indefinitely; this is an important assumption underlying the correctness of mutual exclusion protocols.

\section{Verification}\label{sec:verification}
Below, and in \cref{app:pseudocode}, we collect over a dozen mutual exclusion algorithms from the literature.
We also propose alternative versions of several algorithms satisfying more of the checked properties than the original ones.
We did our verifications with the mCRL2 toolset~\cite{mCRL2toolset}, using the instant-read models, which allow faster verification.
To this end, we had to encode our register and thread models in the mCRL2 process algebra \cite{mCRL2language}. Comments on this translation are given in \cref{app:mcrl2}.
We additionally encoded mutual exclusion, deadlock freedom, and starvation freedom in the modal $\mu$-calculus, the logic used to represent properties in the mCRL2 toolset.
We used the patterns from \cite{spronck2024progress} to incorporate the justness assumption into our formulae for deadlock freedom and starvation freedom.
The modal $\mu$-calculus formulae appear in \cref{app:mucalc}. \href{https://doi.org/10.5281/zenodo.19048046}{The mCRL2 and modal $\mu$-calculus files are available as supplementary material \cite{GLSzenodo}.\footnote{This link will only work once the record has been finalised, until then \uline{\href{https://zenodo.org/records/19048046?preview=1&token=eyJhbGciOiJIUzUxMiJ9.eyJpZCI6ImJkODMzZDcwLTU1ZDgtNDM3MS05MDNiLWM1NWVjZTg5MDI5NSIsImRhdGEiOnt9LCJyYW5kb20iOiIwNDliOTc2N2M2MzY1MzJjNjUyNDZkMDE3NjdhZDliZiJ9.DE1MxBd6bmOjV_g56IGKr0gYClKTL6JOmlH95Wh2nqayzCmkx_epk39faxFNhdN17THWpcPikYyxTGkBeBCmPw}{click here}} to preview the draft.}}

Besides the correctness properties discussed in \cref{sec:introduction}, Dijkstra \cite{dijkstra65} requires that the correctness of the algorithms may not depend on the relative speeds of the threads. This requirement is automatically satisfied in our approach, since we allow all possible interleavings of thread actions in our models.

We checked mutual exclusion, deadlock freedom, and starvation freedom. 
If mutual exclusion is not satisfied, we do not care about the other two properties.
Additionally, if deadlock freedom is not satisfied, we know that starvation freedom is not satisfied either.
We can therefore summarise our results in a single table: $\satnone$ if none of the three properties are satisfied, $\satmutex$ if only mutual exclusion is satisfied, $\satdf$ if only mutual exclusion and deadlock freedom hold, and $\satsf$ if all three are satisfied. See \cref{tab:results}.
As stated previously, we verify liveness properties under justness, where we employ $\conc_T$ for safe and regular registers and all four concurrency relations $\conc_T$, $\conc_S$, $\conc_I$ and $\conc_{\!A}$ for atomic registers.
We checked with 2 threads for algorithms designed for 2 threads, and with 3 for all others. We restrict ourselves to at most 3 threads because, due to the state-space explosion problem, even models with only 3 threads frequently take hours or even days to check these properties on. 

\begin{table}[tbp]
\caption{Verification results}
\label{tab:results}
\vspace{1em}
\resizebox{\textwidth}{!}{%
\begin{tabular}{@{}lc|cccccc@{}}
\toprule
\multirow{2}{*}{\textit{Algorithm}}              & \multirow{2}{*}{\textit{$\#$ threads}} & \textit{Safe} & \textit{Regular} & \multicolumn{4}{c}{\textit{Atomic}} \\
                                                 &                                                   & $T$             & $T$                & $T$       & $S$       & $I$      & $A$      \\ \midrule
    \hyperref[sec:anderson]{Anderson} \cite{anderson1993fine}                 & 2 &  $\satsf$     & $\satsf$ &  $\satsf$ & $\satsf$ & $\satmutex$ &  $\satmutex$ \\* \midrule
    \hyperref[sec:aravind]{Aravind BLRU} \cite{aravind2010yet}             & 3 & $\satsf$ & $\satsf$ & $\satsf$ & $\satmutex$ & $\satmutex$ & $\satmutex$ \\
    \hyperref[sec:aravind]{Aravind BLRU (alt.)}                             & 3  & $\satsf$ & $\satsf$ & $\satsf$ & $\satsf$ & $\satmutex$ & $\satmutex$ \\* \midrule
    \hyperref[sec:attiya-welch]{Attiya-Welch (orig.)} \cite{AttiyaWelch04}        & 2  &  $\satdf$     & $\satsf$    &  $\satsf$ & $\satdf$ & $\satmutex$ & $\satmutex $ \\
    \hyperref[sec:attiya-welch]{Attiya-Welch (orig., alt.)}                       & 2  &  $\satsf$     & $\satsf$    &   $\satsf$ & $\satdf$ & $\satmutex$ & $\satmutex $ \\
    \hyperref[sec:attiya-welch]{Attiya-Welch (var.)} \cite{Shao11}                & 2  &  $\satmutex$  & $\satmutex$ &  $\satsf$ & $\satdf$ & $\satmutex$ & $\satmutex $ \\
    \hyperref[sec:attiya-welch]{Attiya-Welch (var., alt.)}                        & 2  &  $\satsf$     & $\satsf$    & $\satsf$ & $\satdf$ & $\satmutex$ & $\satmutex$  \\* \midrule
    \hyperref[sec:burns-lynch]{Burns-Lynch} \cite{burns1993bounds}               & 3  &  $\satdf$     & $\satdf$    &  $\satdf$ & $\satdf$ & $\satmutex$ & $\satmutex$  \\* \midrule
    \hyperref[sec:dekker]{Dekker} \cite{alagarsamy2003some}                 & 2  &  $\satmutex$  & $\satmutex$ &  $\satsf$ & $\satdf$ & $\satmutex$ & $\satmutex $ \\
    \hyperref[sec:dekker]{Dekker (alt.)}                                   & 2  &  $\satmutex$  & $\satmutex$ &  $\satsf$ & $\satsf$ & $\satmutex$ & $\satmutex $ \\
    \hyperref[sec:dekker]{Dekker RW-safe} \cite{buhr2016dekker}           & 2  &  $\satsf$     & $\satsf$    &  $\satsf$ & $\satdf$ & $\satmutex$ & $\satmutex $ \\
    \hyperref[sec:dekker]{Dekker RW-safe} + \hyperref[sec:dftosf]{Bar-David} & 2 & $\satsf$ & $\satsf$ & $\satsf$ & $\satsf$ & $\satmutex$ & $\satmutex$  \\* \midrule
    \hyperref[sec:dijkstra]{Dijkstra} \cite{dijkstra65}                       & 3  &  $\satmutex$  & $\satdf$ &  $\satdf$ & $\satmutex$ & $\satmutex$ & $\satmutex$  \\
    \hyperref[sec:dijkstra]{Dijkstra (alt.)} & 3 & $\satmutex$ & $\satdf$ & $\satdf$ & $\satdf$ & $\satmutex$ & $\satmutex$ \\* \midrule
    \hyperref[sec:kessels]{Kessels} \cite{kessels1982arbitration}            & 2  &  $\satnone$   & $\satnone$  &  $\satsf$ & $\satsf$ & $\satmutex$ & $\satmutex $ \\* \midrule
    \hyperref[sec:knuth]{Knuth} \cite{knuth1966additional}                 & 3  &  $\satmutex$  &  $\satsf$ & $\satsf$ & $\satmutex$ & $\satmutex$ & $\satmutex$ \\* \midrule
    \hyperref[sec:burns-lynch]{Lamport 1-bit} \cite{Lamport86Mutex2}             & 3  &  $\satdf$     &   $\satdf$  &  $\satdf$ & $\satdf$ & $\satmutex$ & $\satmutex$  \\
    \hyperref[sec:burns-lynch]{Lamport 1-bit} + \hyperref[sec:dftosf]{Bar-David} & 3 & $\satsf$ & $\satsf$ & $\satsf$ & $\satsf$ & $\satmutex$ & $\satmutex$ \\
    \hyperref[sec:lamport-3bit]{Lamport 3-bit} \cite{Lamport86Mutex2}             & 3  &  $\satsf$ & $\satsf$ & $\satsf$ & $\satsf$ & $\satmutex$ & $\satmutex$ \\* \midrule
    \hyperref[sec:kessels]{Peterson} \cite{Peterson81}                       & 2  &  $\satnone$   & $\satnone$  &  $\satsf$ & $\satsf$ & $\satmutex$ & $\satmutex $ \\
    \hyperref[sec:peterson-new]{Peterson new solution (int.)} \cite{peterson1983new} & 3 & $\satdf$ & $\satsf$ & $\satsf$ & $\satdf$ & $\satmutex$ & $\satmutex$ \\
    \hyperref[sec:peterson-new]{Peterson new solution (bit)} \cite{peterson1983new} & 3 & $\satsf$& $\satsf$& $\satsf$& $\satdf$&$\satmutex$ & $\satmutex$\\*
    \midrule
    \hyperref[sec:szymanski-flag]{Szymanski flag (int.)} \cite{Szy88}                    & 3  &  $\satnone$   & $\satnone$  &  $\satsf$ & $\satsf$ & $\satmutex$ & $\satmutex$ \\
    \hyperref[sec:szymanski-flag]{Szymanski flag (bit)} \cite{Szy88}                & 3  &  $\satnone$   & $\satnone$ &   $\satnone$ & $\satnone$ & $\satnone$ & $\satnone$    \\
    \hyperref[sec:szymanski-flag]{Szymanski flag (bit, alt.)} & 3 &  $\satnone$   & $\satnone$  &  $\satsf$ & $\satsf$ & $\satmutex$ & $\satmutex$ \\
    \hyperref[sec:Szymanski]{Szymanski 3-bit lin.\ wait} \cite{Szy90}         & 3  &  $\satnone$   & $\satnone$ &   $\satnone$ & $\satnone$ & $\satnone$ & $\satnone$   \\
    \hyperref[sec:Szymanski]{Szymanski 3-bit lin.\ wait (alt.)}                & 2  &  $\satsf$     & $\satsf$    &  $\satsf$ & $\satsf$ & $\satmutex$ & $\satmutex$     \\
    \hyperref[sec:szymanski-4bit]{Szymanski 4-bit robust} \cite{Szy90} & 3 & $\satnone$ & $\satnone$ & $\satsf^{\star}$ &$\satsf^{\star}$ & $\satmutex$&$\satmutex$\\*
    \bottomrule
\end{tabular}%
}
\end{table}

We list the origin of each algorithm in the table; the results of verifying our proposed alternate versions are indicated by ``alt.''.
We give pseudocode for the algorithms.
Therein we merely present the entry and exit protocols of an algorithm, separated by the instruction \textbf{critical section}. Implicitly these instructions alternate with the non-critical section, and may be repeated indefinitely.
\hypertarget{threadorder}{We use $N$ for the number of threads.
As identifiers for threads}, we use the integers $0\ldots N{-}1$. So $\TID = \{0,\ldots,N{-}1\}$. When presenting pseudocode, we give the algorithm for an arbitrary thread $i$. When $N \mathop= 2$, we use the shorthand notation $j = 1\mathop{-}i$.

In the subsequent sections, we discuss the results summarised in the table.
Where relevant for the discussion, we include the pseudocode here.
The remaining pseudocode can be found in \cref{app:pseudocode}.

\subsection[Impossibility of liveness with reads interference]{Impossibility of liveness with \texorpdfstring{$\conc_I$}{reads interference}}
Perhaps the most notable result in \cref{tab:results} is that no algorithm satisfies either liveness property under $\justact{\conc_I}{\block}$ or $\justact{\conc_{\!A}}{\block}$.
Since $\conc_{\!A}$ is a refinement of $\conc_I$, we focus on the behaviour for $\conc_I$.
When we take $\conc_I$ as our concurrency relation, then one thread's read of a register can interfere with another thread's write to that same register.
It turns out that when this is the case, starvation freedom is impossible for algorithms that rely on communication via registers. The following argument is adapted from \cite{glabbeek2018speed,glabbeek2023modelling}.
Assume that $\mathit{Alg}$ is an algorithm that satisfies starvation freedom.
Let $i$ and $j$ be different threads, and assume that all other threads, if any, stay in their non-critical section forever.
Since $\mathit{Alg}$ is starvation-free, thread $i$ must be capable of freely entering the critical section if thread $j$ is not competing for access.
Hence, 
thread $j$ must communicate its interest in the critical section to thread $i$ as part of its entry protocol.
Since reading from and writing to registers is the only form of communication we allow, 
thread $j$ must, in its entry protocol, write to some register $\varname{reg}$, which $i$ must read in its own entry protocol.
As long as $i$ does not read $j$'s interest from $\varname{reg}$, thread $i$ can enter the critical section freely.
Therefore, if thread $i$'s read of $\varname{reg}$ can block thread $j$'s write to $\varname{reg}$, thread $i$ can infinitely often access the critical section without ever letting thread $j$ communicate its interest, thus never letting thread $j$ enter.

For this argument it is crucial that right after $i$'s read of $\varname{reg}$, thread $i$ enters and then leaves the critical section and returns to its entry protocol, where it engages in another read of $\varname{reg}$, so quickly that thread $j$ has not yet started its write to $\varname{reg}$ in the meantime. 
This uses the requirement on mutual exclusion protocols that their correctness may not depend on the relative speeds of the
threads.
Without that requirement one can easily achieve starvation freedom even with blocking reads, as demonstrated in \cite{glabbeek2023modelling}.

The argument above explains why starvation freedom is never satisfied for $\conc_I$ or $\conc_{\!A}$.
However, it does not explain why we also never observe deadlock freedom.
After all, in the execution sketched above, while thread $j$ is stuck in its entry protocol, thread $i$ infinitely often accesses the critical section.
While we do not (yet) have an argument that deadlock freedom is impossible to satisfy if reads can block writes for all possible algorithms, we do observe this to be the case for all algorithms we have analysed.

For many algorithms, it is possible for both competing threads to become stuck in their entry protocol. Consider, for example, Peterson's algorithm from \cite{Peterson81}, here given as \cref{alg:peterson}.
If $\varname{turn}$ is initially $0$, and thread 1 manages to set $\varidx{flag}{1}$ to $\true$ before thread 0 starts the competition, then on line 3 thread 0 will get stuck in a busy waiting loop. Thread 1 needs to set $\varname{turn}$ to $1$ to let thread 0 pass line 3, but thread 0's repeated reads of $\varname{turn}$ prevent this write from taking place, resulting in both threads being trapped in the entry protocol.
An alternative way to get a deadlock freedom violation is via the exit protocol.
Once a thread has finished its critical section access, it needs to communicate that it no longer requires access to the other thread. In Peterson's, this is done on line 5 by setting the thread's $\varname{flag}$ to $\false$.
However, if the other thread is repeatedly reading this register, such as is done on line 3, then the completion of the exit protocol can be blocked, once again preventing both threads from accessing their critical sections.

We see similar behaviour in all algorithms we analyse.
Frequently, although not always, the problem lies in busy waiting loops.
Given this behaviour, it would be interesting to modify our models to treat busy waiting reads differently from normal reads and allow only normal reads to interfere with writes, or to model \textbf{await}-statements without busy waiting.
This would give us greater insight into whether for some of the algorithms it is truly the busy waiting that is the source of the deadlock freedom violation.
We leave this as future work.

\subsection{Anderson's algorithm}\label{sec:anderson}
Anderson's algorithm \cite{anderson1993fine}, given in \cref{app:anderson}, is specifically designed for two threads.
In accordance with the claims from the paper, it satisfies all three properties with non-atomic as well as atomic registers.
It is worth noting that Anderson specifically states that the \textbf{await}-statements should not be interpreted as busy wait loops, in contrast to how we model them.
Therefore, it is possible that when \textbf{await}-statements are modelled differently, Anderson's may also satisfy liveness properties under $\conc_I$ or $\conc_A$. 
We do not explore that option in this paper.

\subsection{Aravind's BLRU algorithm}\label{sec:aravind}

Aravind's BLRU algorithm \cite{aravind2010yet}, here given as \cref{alg:aravind}, is designed for an arbitrary number of threads $N$.
\begin{table}[t]
\vspace{-1em}
\noindent\begin{minipage}[t]{0.38\textwidth}
\begin{algorithm}[H]
\caption{Peterson's algorithm}\label{alg:peterson}
\begin{algorithmic}[1]
    \State{$\varidx{flag}{i} \writeop \true$}
    \State{$\varname{turn} \writeop i$}
    \State{\textbf{await} $\varidx{flag}{j} = \false \lor \varname{turn} = j$}
    \State{\textbf{critical section}}
    \State{$\varidx{flag}{i} \writeop \false$}
\end{algorithmic}
\end{algorithm}
\end{minipage}
\hfill
\begin{minipage}[t]{0.58\textwidth}
\begin{algorithm}[H]
\caption{Aravind's BLRU algorithm}\label{alg:aravind}
\begin{algorithmic}[1]
    \State{$\varidx{flag}{i} \writeop \true$}
    \Repeat
        \State{$\varidx{stage}{i} \writeop \false$}
        \State{\textbf{await} $\forall_{j \neq i}: \varidx{flag}{j} = \false \lor \varidx{date}{i} < \varidx{date}{j}$}
        \State{$\varidx{stage}{i} \writeop \true$}
    \Until{$\forall_{j \neq i}:$ $\varidx{stage}{j} = \false$}
    \State{\textbf{critical section}}
    \State{$\varidx{date}{i} \writeop \max(\varidx{date}{0}, ..., \varidx{date}{N-1}) + 1$}
    \If{$\varidx{date}{i} \geq 2N-1$}
        \State{$\forall_{j \in [0...N-1]}: \varidx{date}{j} \writeop j$}
    \EndIf
    \State{$\varidx{stage}{i} \writeop \false$}
    \State{$\varidx{flag}{i} \writeop \false$}
\end{algorithmic}
\end{algorithm}
\end{minipage}
\end{table}
Every thread has three registers: $\varname{flag}$ and $\varname{stage}$, Booleans that are initialised at $\false$, and a natural number $\varname{date}$, initialised at the thread's own id.
We observe that this algorithm satisfies all three properties with safe and regular registers, as claimed in \cite{aravind2010yet}.
However, with atomic registers, deadlock freedom is violated under $\justact{\conc_S}{\block}$.
The following execution for two threads demonstrates this violation:
\begin{itemize}
    \item Thread 1 moves through lines 1 through 5, setting $\varidx{flag}{1}$ and $\varidx{stage}{1}$ to $\true$. Note that thread 1 can go through line 4 because $\varidx{flag}{0} = \false$.
    \item Thread 0 can similarly move through lines 1 through 5; while $\varidx{flag}{1} = \true$, we do have that $\varidx{date}{0} = 0 < 1 = \varidx{date}{1}$, so it can pass through line 4. However, on line 6, thread 0 observes $\varidx{stage}{1} = \true$, so it has to return to line 2. 
\end{itemize}
At this point, thread 0 can repeat lines 2 through 5 endlessly, as long as thread 1 does not set $\varidx{stage}{1}$ to $\false$.
Note that the resulting infinite execution satisfies $\justact{\conc_S}{\block}$: thread 1's read of $\varidx{stage}{0}$, which it has to perform on line 6, is repeatedly blocked by thread 0's writes to $\varidx{stage}{0}$ on lines 3 and 5.

This violation can easily be fixed by preventing a thread from endlessly repeating the loop in the entry protocol while the other thread's $\varname{stage}$ is $\false$.
This can be done by altering line 4 to instead say $\textbf{await } \forall_{j \neq i}:\varidx{flag}{j} = \false \lor (\varidx{date}{i} < \varidx{date}{j} \land \varidx{stage}{j} = \false)$.
While this makes it more difficult to progress through line 4, it is impossible for all threads to get stuck there: if all threads are on line 4, then $\varname{stage}$ is $\false$ for all of them, and so the one with the lowest $\varname{date}$ can go to line 5.
Indeed, as is shown in \cref{tab:results}, with this modification Aravind's algorithm now satisfies starvation freedom under $\justact{\conc_S}{\block}$.

\subsection{Attiya-Welch's algorithm}\label{sec:attiya-welch}
What we call the Attiya-Welch algorithm is presented in both \cite{AttiyaWelch04} (original presentation, \cref{alg:attiya-welch}) and \cite{Shao11} (variant presentation, \cref{alg:attiya-welch-var}) as a variant of Peterson's algorithm.
The two presentations of the algorithm have different behaviour, so we present both.
This algorithm is only defined for $N = 2$.
Each thread $i$ has a Boolean $\varidx{flag}{i}$, initialised to $\false$.
There is also a global variable $\varname{turn}$ over $\TID$, initialised to $0$.

\begin{table}[ht!]
\vspace{-1em}
\noindent\begin{minipage}[t]{0.47\textwidth}
\begin{algorithm}[H]
\caption{Attiya-Welch algorithm, orig.}\label{alg:attiya-welch}
\begin{algorithmic}[1]
        \State{$\varidx{flag}{i} \writeop \false$}\label{attiya-welch-entry2}
        \State{\textbf{await} $\varidx{flag}{j} = \false \lor \varname{turn} = j$}
        \State{$\varidx{flag}{i} \writeop \true$}
        \If{$\varname{turn} = i$}\label[line]{attiya-welch-if}
            \If{$\varidx{flag}{j} = \true$}
                \State{\textbf{goto} \cref{attiya-welch-entry2}}\label{attiya-welch-goto}
            \EndIf
        \Else
            \State{\textbf{await} $\varidx{flag}{j} = \false$}\label{attiya-welch-wait}
        \EndIf
        \State{\textbf{critical section}}
        \State{$\varname{turn} \writeop i$}
        \State{$\varidx{flag}{i} \writeop \false$}
\end{algorithmic}
\end{algorithm}
\end{minipage}
\hfill
\begin{minipage}[t]{0.47\textwidth}
\begin{algorithm}[H]
\caption{Attiya-Welch algorithm, var.}\label{alg:attiya-welch-var}
\begin{algorithmic}[1]
    \Repeat
    \State{$\varidx{flag}{i} \writeop \false$}
    \State{\textbf{await} $\varidx{flag}{j} = \false \lor \varname{turn} = j$}
    \State{$\varidx{flag}{i} \writeop \true$}
    \Until{$\varname{turn} = j \lor \varidx{flag}{j} = \false$}\label{attiya-welch-var-until}
    \If{$\varname{turn} = j$}\label{attiya-welch-var-check}
    \State{\textbf{await} $\varidx{flag}{j} = \false$}
    \EndIf
    \State{\textbf{critical section}}\label{attiya-welch-var-crit}
    \State{$\varname{turn} \writeop i$}
    \State{$\varidx{flag}{i} \writeop \false$}
\end{algorithmic}
\end{algorithm}
\end{minipage}
\end{table}

In the original presentation of the Attiya-Welch algorithm, no claims are made about its correctness with non-atomic registers.
It is therefore perhaps surprising to note that it does satisfy many properties with non-atomic registers.
In \cite{spronck2023process}, we showed that the original presentation of the Attiya-Welch algorithm satisfies reachability of the critical section, a property weaker than deadlock freedom, with both safe and regular registers.
Here, we show that while it satisfies starvation freedom with regular registers, it only satisfies deadlock freedom with safe registers.
An execution violating starvation freedom with safe registers is the following:
\begin{itemize}
    \item Thread 0 runs through lines 1 through 9 without competition; since $\varidx{flag}{1} = \false$ it can reach line 9 without problem. On line 10, it starts writing $0$ to $\varname{turn}$, which is already $0$.
    \item Thread 1 executes line 1, setting $\varidx{flag}{1}$ to $\false$, which it already was. On line 2, it reads $\varidx{flag}{0} = \true$ and $\varname{turn} = 1$. Note that the value $1$ has never been written to $\varname{turn}$, but due to the read overlapping with thread 0's write, the value can still be read. Thread 1 can therefore not proceed through line 2.
    \item Thread 0 finishes the exit protocol, setting $\varname{turn}$ to $0$ and $\varidx{flag}{0}$ to $\false$.
\end{itemize}
At this point, we  have $\varidx{flag}{0} = \varidx{flag}{1} = \false$ and $\varname{turn}=0$, the same values the variables had at the start. Hence, thread 0 can reach line 10 again, without interference by thread 1. By having thread 1 always read $\varidx{flag}{0}$ and $\varname{turn}$ at exactly the wrong time, it will remain forever in line 2, without ever reaching the critical section.

Note that this execution relies on reading a value that has never been written. We can adjust the algorithm so that $\varname{turn}$ is read before it is written, and only updated if it does not already have the intended value. Starvation freedom is then satisfied with safe registers.

We find that the Attiya-Welch algorithm does not satisfy starvation freedom with atomic registers under $\justact{\conc_S}{\block}$. 
This violation is rather trivial: as long as $\varidx{flag}{1} = \false$, thread 0 can infinitely often execute the algorithm to get to its critical section. During this execution, it will write to $\varidx{flag}{0}$ repeatedly, preventing thread 1 from reading $\varidx{flag}{0}$ on line 2. Hence, thread 1 can be prevented from ever reaching the critical section if a write can block a read. 
Note that on line 2 of the algorithm, a thread already wants access to its critical section but has in no way communicated this to the other thread.
This is not something we can fix by merely slightly altering a condition or reading a variable before writing; it would require more significant alterations to the algorithm so that a thread has to communicate its intention before attempting to read a register another thread can infinitely often write to.
We do not explore such alterations here.

Instead, we turn our attention to the variant presentation from \cite{Shao11}.
In this version, the goto statement has been eliminated.

In \cite{Shao11}, it is claimed that the algorithm satisfies all three properties for all four interpretations of MWMR regular registers proposed in that paper.
A comparison between their definitions and our regular register model is given in \cite{spronck2023process}; here it suffices to observe that their weakest definition is weaker than our regular register model.
Hence, we would expect all three properties to be satisfied by our regular model.
As can be observed in \cref{tab:results}, this is not the case.
While the two presentations of the algorithm are seemingly equivalent, the altered pseudocode suggests the $\varname{turn}$ variable needs to be read twice in a row, where it is read only once on line 4 of the original presentation.
This allows new-old inversion to occur with the non-atomic register models, and causes the variant presentation to no longer satisfy deadlock freedom with non-atomic registers.
An execution demonstrating the violation is given in \cite{spronck2023process}.
If we alter the model so that the value of $\varname{turn}$ is only read once for the two conditions, we see the expected behaviour, given the results in \cite{Shao11}.
If we then also make the adjustment that $\varname{turn}$ is only written to when its value would actually change, we see the same behaviour as our altered version of the original presentation.
The issue of starvation freedom not being satisfied for atomic registers with $\justact{\conc_S}{\block}$ is also present in this variant.

\subsection{Bar-David's algorithm}\label{sec:dftosf}
In \cite[Section 2.2.2]{Raynal13}, an algorithm is presented to turn any mutual exclusion algorithm that satisfies mutual exclusion and deadlock freedom into one that satisfies starvation freedom as well. It is due to Yoah Bar-David (1998) and first appears in \cite{taubenfeld2006synchronization}.
We take the pseudocode from \cite{GG}, where it is proven that this algorithm works for safe, regular and atomic registers. This algorithm works for arbitrary $N$.
In addition to the registers used by the deadlock-free algorithm, this algorithm uses a Boolean $\varname{flag}$ for every thread, initialised to $\false$, and a shared register over $\TID$ called $\varname{turn}$, initialised to an arbitrary thread id.
Naturally, these registers must be distinct from those used in the deadlock-free algorithm.

To confirm the correctness of the algorithm experimentally, we applied it
to Lamport's 1-bit algorithm \cite{Lamport86Mutex2}, which (by design) does not satisfy starvation freedom at all, and to the RW-safe version of Dekker's algorithm \cite{buhr2016dekker}, where starvation freedom under $\justact{\conc_S}{\block}$ fails due to writes interfering with reads.
The results are given in the table with ``+ Bar-David''. 
We indeed find that starvation freedom is now satisfied where previously only deadlock freedom was.

\begin{table}[t!]
\vspace{-1em}
\noindent\begin{minipage}[t]{0.55\textwidth}
\begin{algorithm}[H]
\caption{Bar-David's algorithm}\label{alg:GG}
\begin{algorithmic}[1]
    \State{$\varidx{flag}{i} \writeop \true$}
    \Repeat
        \State{$\varname{tmp} \writeop \varname{turn}$}
    \Until{$\varname{tmp} = i \lor \varidx{flag}{tmp} = \false$}
    \State{\textbf{entry protocol of deadlock-free algorithm}}
    \State{\textbf{critical section}}
    \State{$\varidx{flag}{i} \writeop \false$}
    \State{$\varname{tmp} \writeop \varname{turn}$}
    \If{$\varidx{flag}{tmp} = \false$}
        \State{$\varname{turn} \writeop (\varname{tmp} + 1) \mod N$}
    \EndIf
    \State{\textbf{exit protocol of deadlock-free algorithm}}
\end{algorithmic}
\end{algorithm}
\end{minipage}
\hfill
\noindent\begin{minipage}[t]{0.4\textwidth}
\begin{algorithm}[H]
\caption{Dekker's algorithm}\label{alg:dekker}
\begin{algorithmic}[1]
    \State{$\varidx{flag}{i} \writeop \true$}
    \While{$\varidx{flag}{j} = \true$}
        \If{$\varname{turn} = j$}
            \State{$\varidx{flag}{i} \writeop \false$}
            \State{\textbf{await} $\varname{turn} = i$}
            \State{$\varidx{flag}{i} \writeop \true$}
        \EndIf
    \EndWhile
    \State{\textbf{critical section}}
    \State{$\varname{turn} \writeop j$}
    \State{$\varidx{flag}{i} \writeop \false$}
\end{algorithmic}
\end{algorithm}
\end{minipage}
\end{table}

\subsection{Burns-Lynch's algorithm and Lamport's 1-bit algorithm}\label{sec:burns-lynch}
The Burns-Lynch algorithm \cite{burns1993bounds} and Lamport's 1-bit algorithm \cite{Lamport86Mutex2} are two very similar algorithms, albeit with slightly different presentations. Both algorithms are designed for an arbitrary $N$, and both use only a single shared Boolean per thread.
In \cite{Lamport86Mutex2}, Lamport makes the explicit claim that this algorithm satisfies mutual exclusion and deadlock freedom with safe registers;
neither algorithm was designed to satisfy starvation freedom.
In \cite{buhr2015high}, it is claimed that both work with non-atomic registers, albeit without satisfying starvation freedom.
With the Attiya-Welch algorithm, we saw that minor differences in presentation can impact the correctness of an algorithm.
This does not appear to be the case with these two algorithms; they show the same, and the expected, behaviour.

\subsection{Dekker's algorithm}\label{sec:dekker}

Dekker's algorithm originally appears in \cite{dijkstra1962over}. There is no clear pseudocode given there, so we use the pseudocode from \cite{alagarsamy2003some}, here given as \cref{alg:dekker}.
The algorithm uses a Boolean $\varname{flag}$ per thread, initially $\false$, and a multi-writer register $\varname{turn}$, initially an arbitrary thread id.
An execution showing that Dekker's algorithm does not satisfy starvation freedom with safe registers is reported in~\cite{buhr2016dekker}.
This same execution can be found by mCRL2. Let $\varname{turn}$ be initially $0$:
\begin{itemize}
    \item Thread 0 goes through the algorithm without competition, and starts setting $\varidx{flag}{0}$ to $\false$ in the exit protocol, when it is currently $\true$.
    \item Thread 1 starts the competition and reads $\varidx{flag}{0} = \false$, the new value, on line 2. It can therefore go to the critical section, and set $\varname{turn}$ to $0$ in the exit protocol.
    \item Thread 1 then starts the competition again, now reading $\varidx{flag}{0} = \true$, the old value, on line 2. Since $\varname{turn} = 0$, it goes to line 5 and starts waiting for $\varname{turn}$ to be $1$.
    \item Thread 0 finishes the exit protocol and never re-attempts to enter the critical section.
\end{itemize}
Since thread 0 will never set $\varname{turn}$ to $1$, thread 1 can never escape line 5.
This execution violates deadlock freedom as well as starvation freedom, and is also applicable to regular registers.
The phenomenon where two reads concurrent with the same write return first the new and then the old value is called new-old inversion, and is explicitly allowed by Lamport in his definitions of safe and regular registers \cite{Lamport86IPCalg}.
An interesting quality of this execution is that it relies on thread 0 only finitely often executing the algorithm. If we did not define the actions $\noncrit[\tid]$ for all $\tid \in \TID$ to be blockable, this execution would be missed.

In \cite{buhr2016dekker} the following improvements are suggested to make the algorithm ``RW-safe'', i.e., correct with safe registers: on line 5, \textbf{await} $\varname{turn} = i \lor \varidx{flag}{j} = \false$, and on line 8, only write to $\varname{turn}$ if its value would be changed.
Our model checking confirms that with these alterations, starvation freedom is satisfied with both safe and regular registers.

In \cite{nigro2024modeling}, it is claimed that Dekker's algorithm without alterations is correct with non-atomic registers.
Instead of dealing with the spurious violations of liveness properties via completeness criteria, they use the model checking tool UPPAAL \cite{Behrmann2004} to compute the maximum number of times a thread may be overtaken by another thread. They determine that bound to be finite and conclude starvation freedom is satisfied.
However, the deadlock freedom violation observed here and in \cite{buhr2016dekker} shows a thread never gaining access to the critical section while only being overtaken once.
Hence, finding a finite upper bound to the number of overtakes is insufficient to establish deadlock freedom.

As can be observed in \cref{tab:results}, both the version of Dekker's algorithm presented in \cite{alagarsamy2003some} and the RW-safe version from \cite{buhr2016dekker} are starvation-free with atomic registers under $\justact{\conc_T}{\block}$, but only deadlock-free under $\justact{\conc_S}{\block}$. In both variants, this is because one thread, say $0$, can remain stuck on line 5 trying to perform a read, while the other thread, in this case $1$, repeatedly executes the full algorithm without having to wait on thread $0$, since $\varidx{flag}{0} = \false$. In the process, thread $1$ writes to the variable that thread $0$ is trying to read, meaning this execution is just under $\justact{\conc_S}{\block}$. The observation that Dekker's algorithm satisfies starvation freedom only when reads are non-blocking and writes that do not change a value are non-blocking is made in \cite{CDV09}.

That it is those writes specifically that lead to the starvation freedom violation can be observed by altering the algorithm so that the relevant write, namely to $\varname{turn}$ on line 8, is only performed when it would change the value; we can then observe that starvation freedom is satisfied.
This change is part of the changes suggested in \cite{buhr2016dekker}, yet that version of the algorithm is not starvation-free under $\justact{\conc_S}{\block}$.
This is due to the other change: on line 5, thread $0$ now also has to read $\varidx{flag}{1}$, the value of which thread $1$ does change every time it executes the algorithm. 
This violation therefore cannot be fixed so easily.
We discussed in \cref{sec:dftosf} that this algorithm does satisfy starvation freedom when combined with the Bar-David algorithm, but this is a significantly more involved process.

\subsection{Dijkstra's algorithm}\label{sec:dijkstra}
Dijkstra's algorithm \cite{dijkstra65} is given as \cref{alg:dijkstra}.
It works for arbitrary $N$.
Every thread $i$ has two Booleans: $\varidx{b}{i}$ and $\varidx{c}{i}$, both initialised to $\true$.
There is also a global register $\varname{k}$ over $\TID$, which is initially an arbitrary thread id.

\begin{table}[ht!]
\vspace{-1em}
\noindent\begin{minipage}[t]{0.54\textwidth}
\begin{algorithm}[H]
\caption{Dijkstra's algorithm}\label{alg:dijkstra}
\begin{algorithmic}[1]
    \State{$\varidx{b}{i} \writeop \false$}\label{dijkstra-Li0}
    \If{$\varname{k} \neq i$}\label[line]{dijkstra-Li1}
        \State{$\varidx{c}{i} \writeop \true$}\label{dijkstra-Li2}
        \If{$\varidx{b}{k} = \true$}\label{dijkstra-Li3}
            \State{$\varname{k} \writeop i$}
        \EndIf
        \State{\textbf{goto} \cref{dijkstra-Li1}}
    \Else
        \State{$\varidx{c}{i} \writeop \false$}\label{dijkstra-Li4}
        \For{$j$ \textbf{from} $0$ \textbf{to} $N - 1$}
            \If{$j \neq i \land \varidx{c}{j} = \false$}
                \State{\textbf{goto} \cref{dijkstra-Li1}}
            \EndIf
        \EndFor
    \EndIf
    \State{\textbf{critical section}}
    \State{$\varidx{c}{i} \writeop \true$}
    \State{$\varidx{b}{i} \writeop \true$}
\end{algorithmic}
\end{algorithm}
\vspace{-1em}
\begin{algorithm}[H]
\caption{Kessels's algorithm}\label{alg:kessels}
\begin{algorithmic}[1]
    \State{$\varidx{q}{j} \writeop \true$}
    \State{$\varidx{r}{j} \writeop (\varidx{r}{i} + j) \mod 2$}
    \State{\textbf{await} $\varidx{q}{i} = \false \lor (\varidx{r}{j} \neq ((\varidx{r}{i} + j) \mod 2))$}
    \State{\textbf{critical section}}
    \State{$\varidx{q}{j} \writeop \false$}    
\end{algorithmic}
\end{algorithm}
\end{minipage}
\hfill
\noindent\begin{minipage}[t]{0.43\textwidth}
\begin{algorithm}[H]
  \caption{Knuth's algorithm}\label{alg:knuth}
  \begin{algorithmic}[1]
    \State{$\varidx{control}{i} \writeop 1$}\label{knuth-L0}
    \For{$j$ \textbf{from} $\varname{k}$ \textbf{downto} $0$}\label{knuth-L1}
        \If{$j = i$}
            \State{\textbf{goto} \cref{knuth-L2}}
        \EndIf
        \If{$\varidx{control}{j} \neq 0$}
            \State{\textbf{goto} \cref{knuth-L1}}
        \EndIf
    \EndFor
    \For{$j$ \textbf{from} $N{-}1$ \textbf{downto} $0$}
        \If{$j = i$}
            \State{\textbf{goto} \cref{knuth-L2}}
        \EndIf
        \If{$\varidx{control}{j} \neq 0$}
            \State{\textbf{goto} \cref{knuth-L1}}
        \EndIf
    \EndFor    
    \State{$\varidx{control}{i} \writeop 2$}\label{knuth-L2}
    \For{$j$ \textbf{from} $N{-}1$ \textbf{downto} $0$}
        \If{$j \neq i \land \varidx{control}{j} = 2$}
            \State{\textbf{goto} \cref{knuth-L0}}
        \EndIf
    \EndFor
    \State{$\varname{k} \writeop i$}\label{knuth-L3}
    \State{\textbf{critical section}}
    \If{$i = 0$}
        \State{$\varname{k} \writeop N{-}1$}
    \Else  
        \State{$\varname{k} \writeop i {-} 1$}
    \EndIf
    \State{$\varidx{control}{i} \writeop 0$}\label{knuth-L4}
    \end{algorithmic}
\end{algorithm}
\end{minipage}
\end{table}

That Dijkstra's algorithm does not satisfy starvation freedom for any memory model is unsurprising, since Dijkstra did not consider this property in his original presentation of the mutual exclusion problem.
He also did not consider non-atomic registers or the impact of blocking, so our results do not contradict any of his claims.
For completeness, we do present executions witnessing the deadlock freedom violations we found. 

The violation with safe registers is rather straightforward: consider the case that $k$ is initially $2$, but thread 2 never leaves its non-critical section. Instead, thread 0 and thread 1 both loop through lines 2 through 6 endlessly: they read $k = 2$ on line 2 and $\varidx{b}{2} = \true$ on line 4 simultaneously, and they simultaneously attempt to set $k$ to their own id. Due to the overlapping writes, the result is that $k$ is set to $2$ again, thus allowing the loop to continue.
This execution is dependent on writing an unintended value, hence why it does not occur with regular registers.

Additionally, there is the following execution with two threads and atomic registers under $\conc_S$, where $k$ is initially $0$:
\begin{itemize}
    \item Thread 0 starts the algorithm; since $k = 0$ it can immediately go to line 8 after executing line 1. Here it sets $\varidx{c}{0}$ to $\false$. On line 10, it needs to read $\varidx{c}{1}$, but this read is blocked by thread 1.
    \item Thread 1 starts the algorithm. It sees $k =0$ so it goes to line 3, where it writes to $\varidx{c}{1}$. Since $\varidx{b}{0} = \false$, it then immediately goes back to line 2, at which point it loops while repeatedly setting $\varidx{c}{1}$ to $\true$.
\end{itemize}
This counterexample can easily be fixed by only performing the writes on lines 3 and 8 when they would change the value.

\subsection{Kessels' algorithm and Peterson's algorithm}\label{sec:kessels}
Kessels's algorithm (\cref{alg:kessels}) is a variant of Peterson's (\cref{alg:peterson}), presented in \cite{kessels1982arbitration}.
It only uses single-writer registers, where Peterson's also uses a multi-writer register.
Kessels' algorithm's basic form is designed for $N = 2$, and each thread $i$ has two variables: a Boolean $\varidx{q}{i}$, initialised at $\false$, and $\varidx{r}{i}$ over $\TID$, which has an arbitrary initial value.

Since Kessels' algorithm is a variant of Peterson's, merely with only single-writer registers, it is unsurprising that the two algorithms give the same results to our verification.

In \cite{spronck2023process}, we already gave a violating execution showing that Peterson's algorithm does not satisfy mutual exclusion with non-atomic registers.
In the current paper, we extend our verification by letting the initial value of $\varname{turn}$ be an arbitrary thread id.
We do give an execution that witnesses that Kessels' algorithm violates mutual exclusion with regular (and thus also with safe) registers. Like many of the executions discussed in this section, this execution relies on new-old inversion. Let the initial values of $\varidx{r}{i}$ be $0$ for all $i \in \TID$.
\begin{itemize}
    \item Thread 0 sets $\varidx{q}{1}$ to $\true$, and then reads $\varidx{r}{0} = 0$. It therefore starts writing $1$ to $\varidx{r}{1}$, which was $0$.
    \item Thread 1 sets $\varidx{q}{0}$ to $\true$, and then reads $\varidx{r}{1}$ with overlap, reading the new value $1$. It therefore writes $1$ to $\varidx{r}{0}$. It then reads $\varidx{r}{0} = 1$ and $\varidx{r}{1} = 0$, the latter being the old value. Since $1 \neq (0 + 0) \mod 2$, it can reach its critical section.
    \item Thread 0 finishes its write to $\varidx{r}{1}$, and on line 3 reads $\varidx{r}{1} = 1$ and $\varidx{r}{0} = 1$; since $1 \neq (1 + 1)\mod 2$, it can reach its critical section.
\end{itemize}

These results do not contradict any claims made by Kessels or Peterson, since they did not claim their algorithms were correct with non-atomic registers.
It is interesting to note, however, that merely replacing multi-writer registers with single-writer Booleans does not automatically mean that an algorithm is robust for non-atomic registers.

\subsection{Knuth's algorithm}\label{sec:knuth}
Knuth's algorithm (\cref{alg:knuth}) comes from \cite{knuth1966additional}.
It works for arbitrary $N$.
Every thread $i$ has a register $\varidx{control}{i}$, which is over $\{0,1,2\}$, initialised to $0$.
Additionally, there is a global register $\varname{k}$ over $\TID$, also initialised to $0$.\footnote{In Knuth's original presentation \cite{knuth1966additional}, $\varname{k}$ was initialised to $-1$, which there was called $0$, as the threads ranged from $1$ to $N$. This difference is immaterial; when thread 1 goes once through  the algorithm without contention, $k$ is set to $0$, so this could just as well be the initial value.}

The goal of Knuth's algorithm is to improve upon Dijkstra's result by having starvation freedom in addition to deadlock freedom. This goal is indeed accomplished with atomic registers, as we see in \cref{tab:results}. Knuth makes no claims on this algorithm's behaviour with non-atomic registers. However, it is interesting that there is a difference between the behaviour with safe and regular registers.
An execution violating deadlock freedom for two threads and safe registers is as follows:
\begin{itemize}
    \item Thread 1 starts setting $\varidx{control}{1}$ to $1$
    \item Thread 0 starts the competition; since $\varname{k} = 0$, it goes to line 12 and sets $\varidx{control}{0}$ to $2$. When it reads $\varidx{control}{1}$ on line 14, it reads a $2$, which has never been written. Hence, it has to return to line 1, and hence start writing to $\varidx{control}{0}$.
    \item Thread 1 finishes line 1 and goes to line 2, where it has to read $\varidx{control}{0}$. It reads $\varidx{control}{0} = 0$, not the new or old value, and hence goes to line 7. There, since $N - 1 = 1$, it goes to line 12 and sets $\varidx{control}{1}$ to $2$. It then goes to line 13, where it finds that $\varidx{control}{0} = 2$, due to the overlap with thread 0's write. Hence, it has to return to line 1, where it starts writing to $\varidx{control}{1}$ again.
\end{itemize}
At this point, we are back where we started, with both threads at the beginning of the algorithm. We see there is a loop where both threads continually make progress, but always get sent back on line 15.
Since the value of $\varname{control}$ changes when written, we cannot fix this by merely ensuring that a register is only written to when its value would change.

\subsection{Lamport's 3-bit algorithm}\label{sec:lamport-3bit}
Lamport's 3-bit algorithm \cite{Lamport86Mutex2} is designed to improve upon the 1-bit algorithm by satisfying starvation freedom as well.
As we can see in \cref{tab:results}, this goal is accomplished.
We discussed Lamport's 3-bit algorithm in detail in \cite{spronck2023process}. The most important observation is that, similar to the variant of the Attiya-Welch algorithm, it is very important that certain registers are only read once when their value is referenced in multiple comparisons, since otherwise new-old inversion can occur.

\begin{table}[ht!]
\vspace{-1.5em}
\noindent\begin{minipage}[t]{0.48\textwidth}
\begin{algorithm}[H]
\caption{Peterson's new solution}\label{alg:peterson-new-int}
\begin{algorithmic}[1]
    \State{$\varidx{c}{i} \writeop \mathit{left}_i$}
    \State{$\mathit{tick}_i$}
    \State{$\mathit{tick}_i$}
    \State{$s \writeop \varidx{c}{i}$}
    \State{$\varidx{c}{i} \writeop s$}\label{new-reset}
    \For{$j$ \textbf{from} $i + 1$ \textbf{to} $N - 1$}
        \State{\textbf{await} $\varidx{c}{j} \neq 3$}
    \EndFor
    \State{$\varidx{c}{i} \writeop 3$}
    \For{$j$ \textbf{from} $i + 1$ \textbf{to} $N - 1$}
        \If{$\varidx{c}{j} = 3$}
            \State{\textbf{goto} \cref{new-reset}}
        \EndIf
    \EndFor
    \For{$j$ \textbf{from} $0$ \textbf{to} $i - 1$}
        \State{\textbf{await} $\varidx{c}{j} \neq 3$}
    \EndFor
    \State{\textbf{critical section}}
    \State{$\varidx{c}{i} \writeop 0$}
\end{algorithmic}
\end{algorithm}
\vspace{-1em}
\begin{algorithm}[H]
\caption{The $\mathit{tick}_i$ algorithm}\label{alg:tick}
\begin{algorithmic}[1]
    \State{\textbf{await} $\varidx{c}{i} \neq \mathit{left}_i$}\label{tick-wait}
    \For{$j$ \textbf{from} $0$ \textbf{to} $i - 1$}
        \State{$t \writeop \varidx{c}{j}$}
        \If{$t = \varidx{c}{i} \lor t = 3$}
            \State{\textbf{goto} \cref{tick-wait}}
        \EndIf 
    \EndFor 
    \State{$\varidx{c}{i} \writeop 3 - \varidx{c}{i}$}
\end{algorithmic}
\end{algorithm}
\vspace{-1em}
\begin{algorithm}[H]
\caption{The $\mathit{left}_i$ algorithm}\label{alg:left}
\begin{algorithmic}[1]
    \For{$j$ \textbf{from} $i - 1$ \textbf{down to} $0$}
        \State{$\varname{t} \writeop \varidx{c}{j}$}
        \If{$\varname{t} = 1 \lor \varname{t} = 3$}
            \State{\textbf{return} $1$}
        \EndIf
        \If{$\varname{t} = 2$}
            \State{\textbf{return} $2$}
        \EndIf
    \EndFor
    \For{$j$ \textbf{from} $N - 1$ \textbf{down to} $i$}
        \State{$t \writeop \varidx{c}{j}$}
        \If{$t = 1 \lor t = 3$}
            \State{\textbf{return} $2$}
        \EndIf
        \If{$t = 2$}
            \State{\textbf{return} $1$}
        \EndIf 
    \EndFor
    \State{\textbf{return} $1$}
\end{algorithmic}
\end{algorithm}
\end{minipage}
\hfill
\noindent\begin{minipage}[t]{0.47\textwidth}
\begin{algorithm}[H]
\caption{The $\mathit{read\_c}(i)$ algorithm}\label{alg:read-bits}
\begin{algorithmic}[1]
    \For{$k$ \textbf{from} $1$ \textbf{to} $2$}
        \If{$\varidx{c_1}{i}$}
            \If{$\varidx{c_2}{i}$}
                \State{\textbf{return} $3$}
            \Else
                \State{\textbf{return} $1$}
            \EndIf
        \EndIf
        \If{$\varidx{c_2}{i}$}
            \State{\textbf{return} $2$}
        \EndIf
    \EndFor
    \State{\textbf{return} $0$}
\end{algorithmic}
\end{algorithm}
\vspace{-1em}
\begin{algorithm}[H]
\caption{The $\mathit{write\_c}(i, \varname{old}, \varname{new})$\\ algorithm}\label{alg:write-bits}
\begin{algorithmic}[1]
    \If{$\varname{old} = 0$}
        \If{$\varname{new} = 1$}
            \State{$\varidx{c_1}{i} \writeop \true$}
        \ElsIf{$\varname{new} = 2$}
            \State{$\varidx{c_2}{i} \writeop \true$}
        \EndIf
    \ElsIf{$\varname{old} = 1$}
        \If{$\varname{new} = 2$}
            \State{$\varidx{c_2}{i} \writeop \true$}
            \State{$\varidx{c_1}{i} \writeop \false$}
        \ElsIf{$\varname{new} = 3$}
            \State{$\varidx{c_2}{i} \writeop \true$}
        \EndIf
    \ElsIf{$\varname{old} = 2$}
        \If{$\varname{new} = 1$}
            \State{$\varidx{c_1}{i} \writeop \true$}
            \State{$\varidx{c_2}{i} \writeop \false$}
        \ElsIf{$\varname{new} = 3$}
            \State{$\varidx{c_1}{i} \writeop \true$}
        \EndIf
    \ElsIf{$\varname{old} = 3$}
        \If{$\varname{new} = 0$}
            \State{$\varidx{c_1}{i} \writeop \false$}
            \State{$\varidx{c_2}{i} \writeop \false$}
        \ElsIf{$\varname{new} = 1$}
            \State{$\varidx{c_2}{i} \writeop \false$}
        \ElsIf{$\varname{new} = 2$}
            \State{$\varidx{c_1}{i} \writeop \false$}
        \EndIf
    \EndIf
\end{algorithmic}
\end{algorithm}
\end{minipage}
\end{table}

\subsection{Peterson's new solution}\label{sec:peterson-new}
In \cite{peterson1983new}, Peterson provides a new solution to the mutual exclusion problem, one that also considers Lamport's extension of the problem (i.e.\ considering non-atomic registers). He gives two variants: an integer version and a Boolean version, both for arbitrary $N$.
He claims that only the Boolean version is immune to non-atomic registers.

We first discuss the integer version (\cref{alg:peterson-new-int}).
Every thread has a single shared variable $\varidx{c}{i}$ which ranges over the domain $\{0, 1, 2, 3\}$, initially $0$.
The solution uses two sub-algorithms, which get passed the thread id as well: $\mathit{left}_i$ (\cref{alg:left}), which returns either $1$ or $2$, and $\mathit{tick}_i$ (\cref{alg:tick}), which has no return value.

Interestingly enough, the integer version of the algorithm already satisfies all three properties with regular registers, and with safe registers it only violates starvation freedom.
We can see this starvation freedom violation with two threads with the following execution:
\begin{itemize}
    \item Thread 0 executes the first 8 lines of the main algorithm without thread 1 doing anything, which results in $\varidx{c}{0} = 3$, with $s$ for thread 0 being $1$. It then starts reading $\varidx{c}{1}$ on line 10. 
    \item Thread 1 starts writing $1$ to $\varidx{c}{1}$ on line 1.
    \item Thread 0, due to overlap, can now read $3$ in $\varidx{c}{1}$, meaning it has to return to line 5 to reset $\varidx{c}{0}$ to $1$. It then gets through line 7 by reading $\varidx{c}{1} \neq 3$, and starts writing $3$ to $\varidx{c}{0}$ again on line 8. 
    \item Thread 1 finishes its write on line 1. By reading the right values for $\varidx{c}{0}$ at the right times, possible due to overlap, it can get through both calls to $\mathit{tick}_1$, ending with $\varidx{c}{1} = 1$. It executes lines 4 and 5 and can skip lines 6 and 7 because it has the highest thread id, and then begins setting $\varidx{c}{1}$ to $3$ on line 8.
    \item Thread 0 can complete its write on line 8 now. On line 10, it sees $\varidx{c}{1} = 3$, so it returns to line 5. It then proceeds back to line 8, since due to overlap it can read $\varidx{c}{1} \neq 3$ on line 7. It starts the write to $\varidx{c}{0}$ on line 8.
    \item Thread 1 finishes its write on line 8. It can once again skip lines 9 to 11 because it has the highest id. On lines 12 and 13, it can read $\varidx{c}{0} \neq 3$ due to overlap, allowing it to enter its critical section. In its exit protocol, it sets $\varidx{c}{1}$ back to $0$.
    \item Thread 1 can now once again attempt to enter the critical section. Due to overlapping reads on $\varidx{c}{0}$, it can get through lines 1 to 3, ending with $\varidx{c}{1} = 1$.
\end{itemize}
We have now reached a situation that was also encountered previously in this execution (in step 4): thread 0 is actively writing to $\varidx{c}{0}$ on line 8, and thread 1 is about to execute line 4. 
We have therefore found a loop where thread 1 endlessly enters the critical section while thread 0 stays in its entry protocol. This execution is just because thread 0 is still performing actions, even though it never reaches the critical section.

Note that this scenario is only possible with safe registers. For example, thread 0 reads a $3$ in $\varidx{c}{1}$ while thread 1 is changing the value from $0$ to $1$.
This is not possible with regular registers, hence  this violation does not appear there.
Every write that causes a problem is a write that truly changes the value of the register, so we cannot fix this by preventing unnecessary writes.
Hence, there does not appear to be an easier solution here than the one suggested by Peterson of changing the integers to Booleans.

To make the Boolean version, Peterson replaces $\varidx{c}{i}$ with Booleans $\varidx{c_1}{i}$ and $\varidx{c_2}{i}$, both initialised to $\false$. 
The resulting algorithm is the same, but reads from and writes to $\varidx{c}{i}$ are replaced with function calls $\mathit{read\_c}(i)$ and $\mathit{write\_c}(i, \varname{old}, \varname{new})$ instead. 
To be more precise: whenever a read of a register $\varidx{c}{i}$ is performed in \cref{alg:peterson-new-int,alg:tick,alg:left}, the function $\mathit{read\_c}(i)$ is called instead; this function reads variables $\varidx{c_1}{i}$ and $\varidx{c_2}{i}$ and returns the appropriate value for $\varidx{c}{i}$. When a write of value $x$ to $\varidx{c}{i}$ is performed, the thread first performs a read of $\varidx{c}{i}$ (using $\mathit{read\_c}(i)$) to determine the current value $y$, and then calls $\mathit{write\_c}(i, y, x)$ to correctly update $\varidx{c_1}{i}$ and $\varidx{c_2}{i}$.
The read and write functions are presented in \cref{alg:read-bits} and \cref{alg:write-bits} respectively.
The change from integers to Booleans makes the algorithm satisfy starvation freedom with safe registers.

The integer and Boolean variants both violate starvation freedom with atomic registers under the concurrency relation $\conc_S$. 
This violation already appears with two threads: the first thing thread 1 must do when it begins the algorithm is read $\varidx{c}{0}$ (resp.\ $\varidx{c_1}{0}$) as part of $\mathit{left}_1$. At this point, it has not performed any writes yet and so it cannot prevent thread 0 from running through the whole algorithm. Whenever thread 0 runs through the algorithm, it performs various writes to $\varidx{c}{0}$ (resp.\ $\varidx{c_1}{0}$). Hence, if writes can block reads, then thread 1 can be indefinitely prevented from completing its first operation.

\subsection{Szymanski's flag algorithm}\label{sec:szymanski-flag}
In \cite{Szy88}, Szymanski presents the flag algorithm, which we here give as \cref{alg:szymanski-flag}.
It is designed for arbitrary $N$.
We give the pseudocode with a minor fix of an obvious typo: on \cref{alg:Szy-flag:exit}, we use a $\lor$ instead of a $\land$. 
In its original presentation, the flag algorithm uses a single integer variable $\varname{flag}$ per thread ranging over $\{0, 1, 2, 3,4\}$, initially $0$. 

\begin{table}[t!]
\vspace{-1em}
\noindent\begin{minipage}[t]{0.48\textwidth}
\begin{algorithm}[H]
\caption{Szymanski's flag algorithm}\label{alg:szymanski-flag}
\begin{algorithmic}[1]
  \State{$\varidx{flag}{i} \writeop 1$}
  \State{\textbf{await} $\forall_j:\ \varidx{flag}{j} < 3$}
  \State{$\varidx{flag}{i} \writeop 3$}
  \If{$\exists_j:\ \varidx{flag}{j}=1$}
    \State{$\varidx{flag}{i} \writeop 2$}
    \State{\textbf{await} $\exists_j:\ \varidx{flag}{j}=4$}
  \EndIf
  \State{$\varidx{flag}{i} \writeop 4$}
  \State{\textbf{await} $\forall_{j < i}:\ \varidx{flag}{j} < 2$}
  \State{\textbf{critical section}}
  \State{\textbf{await} $\forall_{j>i}:\ \varidx{flag}{j} < 2 \lor \varidx{flag}{j} > 3$}\label{alg:Szy-flag:exit}
  \State{$\varidx{flag}{i} \writeop 0$}
\end{algorithmic}
\end{algorithm}
\end{minipage}
\hfill
\noindent\begin{minipage}[t]{0.48\textwidth}
\vspace{-0.6em}
\begin{table}[H]
\caption{Translating $\mathit{flag}$ to three Booleans}
\label{tab:szy-flag2bits}
\centering
\begin{tabular}{@{}c|ccc@{}}
\toprule
\textit{flag} & \textit{intent} & \textit{door\_in} & \textit{door\_out} \\ 
\midrule
0             & $\false$              & $\false$                & $\false$                 \\ 
1             & $\true$               & $\false$                & $\false$                 \\ 
2             & $\false$              & $\true$                 & $\false$                  \\ 
3             & $\true$               & $\true$                 & $\false$                  \\ 
4             & $\true$               & $\true$                 & $\true$                  \\ 
\bottomrule
\end{tabular}
\end{table}
\end{minipage}
\end{table}

It is claimed that, while the integer version of the algorithm is not correct with non-atomic registers, by converting the integers to three Booleans each, $\varname{door\_in}$, $\varname{door\_out}$ and $\varname{intent}$, the algorithm can be made correct for non-atomic registers.
According to \cite{Szy88}, this translation should be done according to \cref{tab:szy-flag2bits}.
This results in \cref{alg:Szy-flag-bits}. The three Booleans are all initialised as $\false$.

\begin{table}[ht!]
\vspace{-1em}
\noindent\begin{minipage}[t]{0.8\textwidth}
\begin{algorithm}[H]
\caption{Szymanski's flag algorithm implemented with Booleans}\label{alg:Szy-flag-bits}
\begin{algorithmic}[1]
  \State{$\varidx{intent}{i} \writeop \true$}\label{szy-bits-1}
  \State{\textbf{await} $\forall_j:\ \varidx{intent}{j} = \false \lor \varidx{door\_in}{j} = \false$}\label{szy-bits-2}
  \State{$\varidx{door\_in}{i}\writeop \true$}\label{szy-bits-3}
  \If{$\exists_j:\ \varidx{intent}{j} = \true \land \varidx{door\_in}{j} = \false$}\label{szy-bits-4}
    \State{$\varidx{intent}{i} \writeop \false$}\label{szy-bits-5}
    \State{\textbf{await} $\exists_j:\ \varidx{door\_out}{j} = \true$}\label{szy-bits-6}
  \EndIf
  \If{$\varidx{intent}{i} = \false$} {$\varidx{intent}{i} \writeop \true$}\EndIf
  \State{$\varidx{door\_out}{i} \writeop \true$}\label{szy-bits-9}
  \State{\textbf{await} $\forall_{j < i}:\ \varidx{door\_in}{j} = \false$}\label{szy-bits-10}
  \State{\textbf{critical section}}\label{szy-bits-11}
  \State{\textbf{await} $\forall_{j>i}:\ \varidx{door\_in}{j} = \false \lor \varidx{door\_out}{j} = \true$}\label{szy-bits-12}
  \State{$\varidx{intent}{i} \writeop \false$}\label{szy-bits-13}
  \State{$\varidx{door\_in}{i} \writeop \false$}\label{szy-bits-14}
  \State{$\varidx{door\_out}{i}\writeop \false$}\label{szy-bits-15}
\end{algorithmic}
\label{alg:szymanski-flag-bits}
\end{algorithm}
\end{minipage}
\end{table}

We analysed both variants in \cite{spronck2023process}, and found that neither algorithm is correct with non-atomic registers. Executions demonstrating these violations are provided there.
Additionally, we observed in \cite{spronck2023process} that the Boolean variant of the algorithm violates mutual exclusion even with atomic registers. We observed that the violating execution reported by mCRL2, given in full in \cite{spronck2023process}, relies on the exit protocol resetting the three Booleans in the order $\varname{intent} - \varname{door\_in} - \varname{door\_out}$. This is the order that is suggested by the final pseudocode provided by Szymanski in \cite[Figure 3]{Szy88}. We posited that the violation of mutual exclusion could be fixed by changing the order to be $\varname{door\_out} -\varname{intent} - \varname{door\_in}$ instead. Now that we can verify the liveness properties as well, we can confirm that this is indeed true: if this alternate exit protocol is used, the Boolean version of the flag protocol is equivalent to the integer version with respect to the properties verified in this paper.

\subsection{Szymanski's 3-bit linear wait algorithm}\label{sec:Szymanski}

\begin{table}[ht!]
\vspace{-1em}
\noindent\begin{minipage}[t]{0.8\textwidth}
\begin{algorithm}[H]
    \caption{Szymanski's 3-bit linear wait algorithm}
    \label{alg:szymanski-3bit}
    \begin{algorithmic}[1]
      \State{$\varidx{a}{i} \writeop \true$}\label{szy-3bit-1}
      \lFor{$j$ \textbf{from} $0$\ \textbf{to}\ $N{-}1$} {\textbf{await} $\varidx{s}{j} = \false$}\label{szy-3bit-2}
      \State{$\varidx{w}{i} \writeop \true$}\label{szy-3bit-3}
      \State{$\varidx{a}{i} \writeop \false$}\label{szy-3bit-4}
      \While{$\varidx{s}{i} = \false$}\label{szy-3bit-5}
        \State{$j \writeop 0$}\label{szy-3bit-6}
        \lWhile{$j < N \land \varidx{a}{j} = \false$}{$j\writeop j+1$}\label{szy-3bit-7}
          \If{$j=N$}\label{szy-3bit-8}
            \State{$\varidx{s}{i} \writeop \true$}\label{szy-3bit-9}
            \State{$j\writeop 0$}         \label{szy-3bit-10}
            \lWhile{$j<N\land \varidx{a}{j} = \false$}{$j\writeop j+1$}\label{szy-3bit-11}
            \If{$j< N$} {$\varidx{s}{i} \writeop \false$}\label{szy-3bit-12}
           \Else\label{szy-3bit-13}
             \State{$\varidx{w}{i} \writeop \false$}\label{szy-3bit-14}
             \lFor{$j$ \textbf{from} $0$\ \textbf{to} $N-1$}{\textbf{await} $\varidx{w}{j} = \false$}\label{szy-3bit-15}
           \EndIf
         \EndIf
         \If{$j<N$}\label{szy-3bit-16}
           \State{$j\writeop 0$}\label{szy-3bit-17}
           \lWhile{$j<N \land (\varidx{w}{j} = \true \lor \varidx{s}{j} = \false)$}{$j\writeop j+1$}\label{szy-3bit-18}
         \EndIf
         \If{$j \neq i \land j<N$}\label{szy-3bit-19}
         \State{$\varidx{s}{i} \writeop \true$}\label{szy-3bit-20}
         \State{$\varidx{w}{i} \writeop \false$}\label{szy-3bit-21}
         \EndIf
     \EndWhile
      \lFor{$j$ \textbf{from} $0$ \textbf{to} $i-1$}{\textbf{await} $\varidx{s}{j} = \false$}\label{szy-3bit-22}
      \State{\textbf{critical section}}\label{szy-3bit-23}
      \State{$\varidx{s}{i} \writeop \false$}\label{szy-3bit-24}
     \end{algorithmic}
  \end{algorithm}
\end{minipage}
\end{table}
  
In \cite{Szy90}, Szymanski proposes four mutual exclusion algorithms. Here, we discuss the first: the 3-bit linear wait algorithm, which is claimed to be correct with non-atomic registers.
See \cref{alg:szymanski-3bit}. Each thread has three Booleans, $\varname{a}$, $\varname{w}$, and $\varname{s}$, all initially $\false$.
An execution showing a mutual exclusion violation for safe, regular, and atomic registers with three threads for the 3-bit linear wait algorithm is given in \cite{spronck2023process}.

With two threads, we do still get a mutual exclusion violation with safe and regular registers, given in detail in \cite{spronck2023process}, but not with atomic registers.
The violation specifically relies on reading $\varidx{w}{j}$ before $\varidx{s}{j}$ on line 18, so that there can be new-old inversion on $\varidx{s}{j}$, while still obtaining the value of $\varidx{w}{j}$ from before the other thread started writing to $\varidx{s}{j}$.
We therefore considered the alternative, where $\varidx{s}{j}$ is read before $\varidx{w}{j}$ on line 18.
With this change and only two threads, the algorithm satisfies all expected properties.

The observation that Szymanski's 3-bit linear wait algorithm violates mutual exclusion with three threads is also made in \cite{nigro2024Verifying}. 
They suggest a different way to make the algorithm correct for two threads: instead of swapping the reads on line 18, they require a thread to also read $\varidx{w}{j}$ on line 22.
We did not investigate this suggested change further, as we prefer the solution that does not require an additional read.

\subsection{Szymanski's 4-bit robust algorithm}\label{sec:szymanski-4bit}
    We discussed the first of the four algorithms given in \cite{Szy90} in the previous section. After the 3-bit linear wait algorithm, a variant is given that is claimed to be robust to premature termination, abortion, failure and transient malfunction; then a 4-bit extension of the 3-bit algorithm that is claimed to satisfy the first-come first-serve priority property; and finally a variant of the 4-bit algorithm that is claimed to be robust to all aforementioned failures, including non-atomic registers.
  We elected to not cover all four algorithms, as we already established that their base, \cref{alg:szymanski-3bit}, is not correct with non-atomic registers.
  However, out of interest we do cover the final, 4-bit robust algorithm as well.

  The 4-bit robust algorithm comes from \cite[Figure 5]{Szy90}. We give it as \cref{alg:szymanski-4bitR}.
  Here, every thread has four Booleans $a$, $w$, $p$ and $s$, all initialised at $\false$.
  Additionally each thread has a private variable $c$ as a counter, which is initialised to $N {-} 1$. While this is not stated explicitly in the algorithm as presented by Szymanski, we enforce that the counter resets to $N {-} 1$ before \cref{Szy-4bitR-do}, since otherwise it could become negative on \cref{Szy-4bitR-45}.
  Every thread also has private variables $\varidx{la}{j}$ and $\varidx{lp}{j}$ for all $j \in \TID$.
  We write $r \optwriteop x$ as shorthand notation for reading $r$ and then performing $r \writeop x$ only if $r$ is not already $x$.

  With safe and regular registers, mutual exclusion is violated. As is to be expected, this issue originates from new-old inversion.
Consider the following execution with two threads:
\begin{itemize}
    \item Thread 0 and 1 both perform lines 1 through 4. Crucially, they both set $\varidx{la}{0} = \varidx{la}{1} = \false$.
    \item Thread 1 continues with lines 5 through 9. Here, it sets $\varidx{p}{1}$ and $\varidx{w}{1}$ to $\true$. On lines 10 to 14, it sees $\varidx{s}{0} = \varidx{s}{1} = \false$, so nothing significant happens. On line 15, it sees $\varidx{s}{1} = \false$, so it performs line 16 and 17 and sets $\varidx{a}{1}$ to $\false$. 
    \item Thread 0 performs lines 5 through 16. Here, it sets $\varidx{p}{0}$ and $\varidx{w}{0}$ to $\true$. On line 15, it sees $\varidx{s}{0} = \false$, so it continues into the while-loop. On line 17, it begins setting $\varidx{a}{0}$ to $\false$, but it does not complete this write yet. 
    \item Thread 1 continues from line 18. Due to new-old inversion, it can first read the new value of $\varidx{a}{0}$, which is $\false$. Thus, it can see $\varidx{a}{0} = \varidx{a}{1} = \false$ and so terminate line 19 with $j = 2 = N$. Hence, it sets $\varidx{s}{1}$ to $\true$. On line 23, it again sees $\varidx{a}{0} = \varidx{a}{1} = \false$, so terminates with $j = N$ and thus sets $\varidx{w}{1}$ to $\false$. Lines 28 to 36 are all skipped because $j = N$, so thread 1 returns to line 15 where it sees $\varidx{s}{1} = \true$. Hence, thread 1 continues to line 37. Here, it sees $\varidx{a}{0} = \true$, the old value of thread 0's write, instead. It also sees $\varidx{w}{1} = \false$, so it does not get caught in the while loop and can continue to line 41.
    \item Thread 0 completes its write to $\varidx{a}{0}$ on line 17, setting it to $\false$. At this point, $\varidx{a}{0} = \varidx{a}{1} = \false$, so line 19 terminates with $j = N$. Thread 0 therefore goes to line 21, but does not execute it yet.  
    \item Thread 1 continues its execution from line 41. Note that $\varidx{la}{0} =  \varidx{la}{1} = \false$, so line 43 terminates with $k = N$. Hence, it skips lines 44 to 54 and continues on line 55. Here it sees $\varidx{s}{0} = \false$, so it can enter its critical section. 
    \item Thread 0 continues on line 21, setting $\varidx{s}{0}$ to $\true$. Since $\varidx{a}{0} = \varidx{a}{1} = \false$, it terminates line 23 with $j = N$ and so executes line 27, setting $\varidx{w}{0}$ to $\false$. Since $j = N$, it then immediately returns to line 15 where it sees $\varidx{s}{0} = \true$. So it continues on line 37. Note that $\varidx{w}{0} = \varidx{w}{1} = \false$, so it does not change its variables and immediately goes to line 41. Since $\varidx{la}{0} = \varidx{la}{1} = \false$, it goes right to line 55. Here, since $0$ is the lowest thread id, it can immediately go to line 59 where it enters its critical section.
    \item Now thread 0 and thread 1 can both enter their critical sections.
\end{itemize}
This way, mutual exclusion is violated.

In the table, we give the results for atomic registers with $\conc_T$ and $\conc_S$ as $\satsf^{\star}$; this is because we were unable to confirm these results with three threads, due to the length of the algorithm and subsequent size of the state space. 
With 2 threads, starvation freedom is satisfied in these cases. 
Regardless, in terms of the properties we check the 4-bit robust algorithm is at most as good as the flag algorithm.

   \begin{algorithm}[hb!]
    \caption{Szymanski's 4 bit robust FCFS algorithm}
    \label{alg:szymanski-4bitR}
    \begin{algorithmic}[1]
      \State{$\varidx{a}{i} \writeop \true$}
      \For{$j$ \textbf{from} $0$ \textbf{to} $N - 1$}
        \State{$\varidx{la}{j} \writeop \varidx{w}{j}$}
        \State{$\varidx{lp}{j} \writeop \varidx{p}{j}$}
      \EndFor
      \If{$\varidx{p}{i} = \true$}
        \State{$\varidx{p}{i} \writeop \false$}
      \Else
        \State{$\varidx{p}{i} \writeop \true$}
      \EndIf
      \State{$\varidx{w}{i} \writeop \true$}
      \For{$k$ \textbf{from} $1$ \textbf{to} $N-1$}
        \For{$j$ \textbf{from} $0$ \textbf{to} $N - 1$}
            \While{$\varidx{s}{j} = \true$}
                \State{$\varidx{a}{i} \optwriteop \true$}
                \State{$\varidx{s}{i} \optwriteop \false$}
            \EndWhile
        \EndFor
      \EndFor
      \While{$\varidx{s}{i} = \false$}
        \State{$\varidx{w}{i} \optwriteop \true$}
        \State{$\varidx{a}{i} \optwriteop \false$}
        \State{$j \writeop 0$}
        \lWhile{$j < N \land \varidx{a}{j} = \false$} {$j \writeop j + 1$}
        \If{$j = N$}
            \State{$\varidx{s}{i} \writeop \true$}
            \State{$j \writeop 0$}
            \lWhile{$j < N \land \varidx{a}{j} = \false$} {$j \writeop j + 1$}
            \If{$j < N$}
                \State{$\varidx{s}{i} \writeop \false$}
            \Else
                \State{$\varidx{w}{i} \writeop \false$}
            \EndIf
        \EndIf
        \If{$j < N$}
            \State{$j \writeop 0$}
            \lWhile{$j < N \land (\varidx{w}{j} = \true \lor \varidx{s}{j} = \false)$} {$j \writeop j + 1$}
        \EndIf
        \If{$j \neq i \land j < N$}
            \State{$\varidx{s}{i} \writeop \true$}
            \If{$\varidx{s}{j} = \false$}
      \algstore{szy-4bit-split}
    \end{algorithmic}
\end{algorithm}
\clearpage

\begin{algorithm}[H]
    \begin{algorithmic}[1]
    \algrestore{szy-4bit-split}
                \State{$\varidx{s}{i} \writeop \false$}
            \Else
                \State{$\varidx{w}{i} \writeop \false$}
            \EndIf
        \EndIf
      \EndWhile
    \For{$j$ \textbf{from} $0$ \textbf{to} $N - 1$}
        \While{$\varidx{w}{j} = \true \land \varidx{a}{j} = \false$}
            \State{$\varidx{a}{i} \optwriteop \false$}
            \State{$\varidx{w}{i} \optwriteop \false$}
        \EndWhile
      \EndFor
      \State{$c \writeop N - 1$}
      \State{$k \writeop 0$}\label{Szy-4bitR-do}
      \lWhile{$k < N \land (\varidx{la}{k} = \false \lor \varidx{p}{k} \neq \varidx{lp}{k} \lor \varidx{s}{k} = \false)$} {$k \writeop k+1$}
      \If{$k < N$}
        \State{$c \writeop c - 1$}\label{Szy-4bitR-45}
        \For{$j$ \textbf{from} $0$ \textbf{to} $N - 1$}
            \While{$\varidx{a}{j} = \false \land \varidx{s}{j} = \true$}
                \State{$\varidx{a}{i} \optwriteop (j \geq i)$}
                \State{$\varidx{w}{i} \optwriteop \false$}
            \EndWhile
        \EndFor
        \For{$j$ \textbf{from} $0$ \textbf{to} $N - 1$}
            \While{$\varidx{a}{j} = \true \land \varidx{s}{j} = \true$}
                \State{$\varidx{a}{i} \optwriteop (j < i)$}
                \State{$\varidx{w}{i} \optwriteop \false$}
            \EndWhile
        \EndFor
      \EndIf
      \If{$k < N \land c > 0$} {\textbf{goto} \cref{Szy-4bitR-do}} \EndIf
      \For{$j$ \textbf{from} $0$ \textbf{to} $i - 1$}
        \While{$\varidx{a}{j} = \false \land \varidx{s}{j} = \true$}
            \State{$\varidx{a}{i} \optwriteop \false$}
            \State{$\varidx{w}{i} \optwriteop \false$}
        \EndWhile
      \EndFor
      \State{\textbf{critical section}}
      \State{$\varidx{s}{i} \writeop \false$}
     \end{algorithmic}
  \end{algorithm}  

\section{Related work}\label{sec:relwork}
To the best of our knowledge, we are the first to do automatic verification of mutual exclusion algorithms while incorporating both which operations can block each other and the effects of overlapping write operations. The two elements have been considered separately previously.

We have previously cited \cite{CDV09} as a source of the term ``blocking'' used in the same way as we do.\footnote{Subtly different notions of blocking are widespread; see for example \url{https://en.wikipedia.org/wiki/Blocking_(computing)} (accessed March 2026).} That work specifically considers blocking vs.\ non-blocking reading by introducing a process algebra, PAFAS$_s$, that allows for non-blocking reading, similar to read arcs in Petri nets. It also allows writes to be non-blocking when they do not change the value of a variable. Note that this situation is between our $\conc_T$, where all writes are non-blocking, and $\conc_S$, where all writes are blocking. However, where the distinction between writes that change a value and writes that do not is relevant, we have altered the algorithms to read registers first and only perform a write if that would change the value, thus allowing us to consider the same case as \cite{CDV09} using $\conc_S$. Dekker's algorithm is analysed in \cite{CDV09}, as discussed in \cref{sec:dekker}.

In \cite{Buti2011automated}, PAFAS$_s$ is used in conjunction with the tool FASE to check starvation freedom of Peterson's, Lamport's 1-bit, Dijkstra's and Knuth's algorithms under a fairness assumption that turned out to be equivalent \cite{glabbeek2019progress} to our notion of justness, with the same non-blocking actions as used in \cite{CDV09}.
They show that Peterson's requires non-blocking reading to satisfy liveness properties. 
For Lamport's, they show that with two threads with id's $1$ and $2$, thread $1$ satisfies starvation freedom but thread $2$ does not. 
They check Dijkstra's and Knuth's with 2 threads and find that starvation freedom is violated for both. 
These results all agree with ours.

In \cite{bouwman2020off}, starvation-freedom of Peterson's algorithm with atomic registers is checked using mCRL2 under the justness assumption. Two different concurrency relations are considered, which are similar to our $\conc_S$ and $\conc_{\!A}$.

There have been many formal verifications of mutual exclusion algorithms with atomic registers \cite{benari2008principles,mateescu2013model,groote2021tutorial}.
Non-atomic registers have been covered less frequently, and their verification is often restricted to single-writer registers \cite{lamportHyperook,buhr2016dekker}.
The safety properties of mutual exclusion algorithms with MWMR non-atomic registers were verified using mCRL2 in~\cite{spronck2023process}. In several papers by Nigro, including \cite{nigro2024Verifying,nigro2024modeling}, safety and liveness verification with MWMR non-atomic registers has been done using UPPAAL.

A major drawback of our approach is that we only consider a small number of parallel threads.
There has been much work in the literature on parametrised verification, where the correctness of an algorithm is established for an arbitrary number of threads \cite{bouajjani2000regular,fisman2001beyond,zuck2004model,bloem2016decidability}.
Where these techniques handle liveness, it is under forms of weak or strong fairness, not justness.
To our knowledge, these techniques have not yet been applied in the context of non-atomic registers. 
While several papers on parametrised verification, such as \cite{abdulla2008handling,abdulla2016parameterized}, mention dropping the atomicity assumption, this refers to evaluating existential and universal conditions on all threads in a single atomic operation, rather than atomicity of the registers.

In \cite{AravindH11,Hesselink13a,hesselink2015mutual} and other work by Hesselink, interactive theorem proving with the proof assistant PVS is used to check safety and liveness properties (under fairness) of mutual exclusion algorithms for an arbitrary number of threads with non-atomic registers. To our knowledge, the case of writes overlapping each other has not been covered with this technique.

\section{Conclusion}\label{sec:conclusion}
When it comes to analysing the correctness of algorithms, particularly when considering non-atomic registers, behavioural reasoning is often insufficient.
Mistakes can be subtle, and may depend on edge-cases that are easily overlooked.
Model checking is a solution here; we formally model the threads executing the algorithm, as well as the registers through which they communicate, and the entire state-space is searched for possibly violations of correctness properties.
In this work, we verify a large number of mutual exclusion algorithms using the model checking toolset mCRL2.
We expand on previous work by checking liveness properties -- deadlock freedom and starvation freedom -- in addition to the main safety property.
To circumvent spurious violating paths in our models, we incorporate the completeness criterion justness into our verifications.
We checked algorithms under six different memory models, where a memory model is a combination of a register model (safe, regular, or atomic) and a model of which register access operations can block each other.
The former dimension we capture in the models themselves, by modelling the behaviour of the three types of register.
The latter, we capture in the concurrency relations employed as part of justness.
We found a number of interesting violations of correctness properties, and in some cases could suggest improvements to algorithms to fix these violations.
We find that there are several algorithms that satisfy all three properties for four out of six memory models.
For three threads, this is accomplished by the fixed version of Aravind's BLRU algorithm and Lamport's 3-bit algorithm. If we also consider algorithms for just two threads, then Anderson's algorithm and the fixed version of Szymanski's 3-bit linear wait algorithm also meet this bar. We also considered Bar-David's algorithm for turning deadlock-free algorithms into starvation-free ones, and experimentally confirmed that it indeed works for Dekker's algorithm made RW-safe and Lamport's 1-bit algorithm.

\bibliographystyle{alphaurl}
\bibliography{bib2doi}

\newpage
\appendix

\section{Register models}\label{app:registers}
In this section, we expand on the material of \cref{sec:fullread} and \cref{sec:instantread}, by giving the formal process-algebraic models of both the full-read and instant-read registers.
First, we present the general structure of such models and explain how to translate them to LTSs.
We then give the safe, regular and atomic models, as well as the definitions of all the access and update functions for both model types.
We leave the status objects abstract: it is not necessary to define the data structures themselves, as long as it is clear what information can be retrieved from them.

\subsection{Structure and translation to LTSs}\label{app:structure}
Recall that we use $\TID$ for the thread identifiers, $\RID$ for register identifiers, and that for every $\rid \in \RID$, the set $\Data$ contains all values that $\rid$ may hold.
Additionally, we use a status object as the finite memory of a register. 
The full-read and instant-read models use different status objects, but for the purpose of presenting their shared structure we denote for both the set of possible statuses as $\StatusAll$.

We define the following structure, shared by all register models.
Let $\regtype \mathbin\in \{\safRep, \regRep, \atoRep\}$ and $\modeltype \in \{\frRep, \irRep\}$; then each model looks as follows, for some natural number $n$:
\begin{equation}\label{eq:structure}
    \Reg[\modeltype,\regtype](\rid : \RID, s: \StatusAll) = \sum_{\tid \in \TID}\sum_{\data \in \Data}\sum_{0 \leq j < n} (c_j(s, \tid, \data) \rightarrow a_j(\tid,\data) \cdot \Reg[\modeltype,\regtype](\rid, u_j(s,\tid,\data)))
\end{equation}
This represents a register with id $\rid$ and status $s$.
The process first sums over $\tid \in \TID$ and $\data \in \Data$, allowing interaction by all threads and with all possible data parameters.
Furthermore, it has $n$ summands, each of the form $c_j(s,\tid,\data) \rightarrow a_j(\tid,\data) \cdot \Reg[\modeltype,\regtype](\rid, u_j(s,\tid,\data))$ where  $c_j(s,\tid,\data)$ is a Boolean condition, $a_j(\tid,\data)$ is an action, and $u_j$ is an \emph{update function} that takes $s$, $\tid$ and $\data$ and returns the updated status.

Such a process equation gives rise to an LTS. 
Given a predefined initial state $\initstate \in \StatusAll$, the LTS of $\Reg[\modeltype,\regtype](\rid, \initstate)$ is:
\begin{equation}\label{eq:LTS}
    (\StatusAll, \!\!\bigcup_{0 \leq j < n}\!\!\{a_j(\tid,\data) \mid \tid \mathbin\in \TID, \data \mathbin\in \Data\}, \initstate, \!\!\bigcup_{0 \leq j < n}\!\!\{(s, a_j(\tid,\data), u_j(s,\tid,\data)) \mid \tid \mathbin\in \TID, \data \mathbin\in \Data \land c_j(s,\tid,\data)\})
\end{equation}

\subsection{Full-read models}\label{app:fullread}
We here describe the full-read models. To this end, we first define the status object by describing the initial state and update functions. We then give the process-algebraic definitions of the three types of register.

\paragraph{Full-read status object.}\label{app:fullread-status}

As stated previously, we do not wish to go into the implementation details of the status object.
Instead, we define a collection of access functions (one of which is really a predicate) to retrieve information from the current state. These are sufficient to describe the relevant aspects the initial state and update functions. 
We use the following access functions, which are local to any given register $r$:
\begin{itemize}
    \item The function $\truevalsym: \StatusAll \rightarrow \Data$ returns the value that is currently stored in the register.
    \item The function $\readerssym: \StatusAll \rightarrow 2^{\TID}$ returns the set of thread id's of threads that have invoked a read operation on this register that has not yet had its response.
    \item The function $\writerssym: \StatusAll \rightarrow 2^{\TID}$ returns the set of thread id's of threads that have invoked a write operation on this register that has not yet had its response.
    \item The function $\pendingsym: \StatusAll \rightarrow 2^{\TID}$ returns the set of thread id's of threads that have invoked an operation that has not been ordered yet (only used by the regular and atomic models).
    \item The function $\valssym: \StatusAll \times \TID \rightarrow \Data$ allows us to record a single data value per thread; it is, for instance, used to remember what value was passed with a start write action.
    \item The predicate $\overlapsym$ on $\StatusAll \times \TID$ reflects whether an ongoing read or write operation of a thread has encountered an overlapping write (only used by the safe model
    ).
    \item The function $\posvalsym: \StatusAll \times \TID \rightarrow 2^{\Data}$ stores a set of values per thread, representing the possible return values of an ongoing read by a thread (only used by the regular model).
\end{itemize}

Let $\data_{\initstate}$ be an appropriate initial value of the register, which depends on the modelled algorithm. The initial state $\initstate$ is characterised by specifying, for all $\tid \in \TID$, the results of the access functions:
\[\begin{array}{r@{~}c@{~}lr@{~}c@{~}lr@{~}c@{~}lr@{~}c@{~}l}
    \trueval{\initstate} &=& \data_{\initstate} & \writers{\initstate} &=& \emptyset & \vals{\initstate, \tid} &=& \data_{\initstate}  & \posval{\initstate, \tid} &=& \emptyset\\
    \readers{\initstate} &=& \emptyset & \pending{\initstate} &=& \emptyset & \overlap{\initstate, \tid} &=& \false & &\\[-1ex]
\end{array}\]

We now define the update functions by showing how the return values of the access functions are altered by the update function.
Each update function corresponds to an action and is applied when that action occurs; if the action's name is $a$, the update function is called $\mathit{ua}$.
Not every update function uses the data parameter that is passed to it according to \labelcref{eq:structure}; in these cases we only give the thread id and status parameters.
If an access function is not mentioned, then its return value after the update is the same as before.
Given an arbitrary state $s \in \StatusAll$, thread id $\tid \in \TID$ and data value $\data \in \Data$:\\
If $s' = \usr{s, \tid}$, then:\vspace{-1.5ex}
\begin{align*}
    \readers{s'} &= \readers{s} \cup \{\tid\}\\
    \pending{s'} &= \pending{s} \cup \{\tid\}\\
    \overlap{s', \tid} &= (\writers{s} > 0)\\
    \posval{s', \tid} &= \{\trueval{s}\} \cup \{\data' \mid \exists_{\tidtwo \in \writers{s}}.\vals{s, \tidtwo} = \data'\}
\end{align*}
If $s' = \ufr{s, \tid}$, then:\vspace{-1.5ex}
\begin{align*}
    \readers{s'} &= \readers{s} \setminus \{\tid\}
\end{align*}
If $s' = \usw{s, \tid, \data}$, then for all $\tidtwo \neq \tid$:\vspace{-1.5ex}
\begin{align*}
    \writers{s'} &= \writers{s} \cup \{\tid\}\\
    \pending{s'} &= \pending{s} \cup \{\tid\}\\
    \vals{s', \tid} &= \data\\
    \overlap{s', \tid} &= (\writers{s} > 0)\\
    \overlap{s', \tidtwo} &= \true\\
    \posval{s', \tidtwo} &= \posval{s, \tidtwo} \cup \{\data\}
\end{align*}
If $s' = \ufw{s, \tid, \data}$, then:\vspace{-1.5ex}
\begin{align*}
    \trueval{s'} &= \data\\
    \writers{s'} &= \writers{s} \setminus \{\tid\}
\end{align*}
If $s' = \uor{s, \tid}$, then:\vspace{-1.5ex}
\begin{align*}
    \pending{s'} &= \pending{s} \setminus \{\tid\}\\
    \vals{s', \tid} &= \trueval{s}
\end{align*}
If $s' = \uow{s, \tid, \data}$, then:\vspace{-1.5ex}
\begin{align*}
    \trueval{s'} &= \data\\
    \pending{s'} &= \pending{s} \setminus \{\tid\}
\end{align*}

\noindent
These formal definitions correspond to the intuitive descriptions given above. Of note is that $\overlap{s', \tidtwo}$ is set to $\true$ whenever a thread $\tid \neq \tidtwo$ starts a write, even if $\tidtwo$ is not actively reading or writing. This is done for simplicity of the definition: when $\tidtwo$ starts reading or writing, it will reset its own $\overlapsym$ to the correct value, depending on whether there is an overlapping write active at that point.
Something similar is done with $\posvalsym$: when a write is started, its value gets added to the $\posvalsym$ sets of every other thread, even if they are not actively reading. When a thread starts reading, it sets its own $\posvalsym$ correctly.

\paragraph{Full-read safe MWMR registers.}\label{app:fullread-safe}

See \cref{fig:fullread-procsafe} for the process equation representing our full-read safe register model.
\begin{figure}[ht]
    \begin{multline*}
        \Reg[\frRep, \safRep](\rid: \RID, s:\StatusAll) = \\
        \sum_{\tid\in\TID}\sum_{\data \in \Data}
        \left(\begin{array}{ll}
            & (\tid\notin(\readers{s} \cup \writers{s})) \then
                 \startread[\tid,\rid]\co\Reg[\frRep, \safRep](\rid,\usr{s,\tid}) \\
          + & (\tid\in\readers{s}\wedge\neg\overlap{s,\tid}) \then
                 \finishread[\tid,\rid]{\trueval{s}}\co\Reg[\frRep, \safRep](\rid,\ufr{s,\tid}) \\
          + & (\tid\in\readers{s}\wedge\overlap{s,\tid}) \then
                 \finishread[\tid,\rid]{\data}\co\Reg[\frRep, \safRep](\rid,\ufr{s,\tid})\\
          + & (\tid\notin(\readers{s} \cup \writers{s})) \then
                 \startwrite[\tid,\rid]{\data}\co\Reg[\frRep, \safRep](\rid,\usw{s,\tid,\data})\\
          + & (\tid\in\writers{s}\wedge\neg\overlap{s,\tid}) \then
                 \finishwrite[\tid,\rid]\co\Reg[\frRep, \safRep](\rid,\ufw{s,\tid,\vals{s, \tid}})\\
          + & (\tid\in\writers{s}\wedge\overlap{s,\tid}) \then
                 \finishwrite[\tid,\rid]\co\Reg[\frRep, \safRep](\rid,\ufw{s,\tid,\data})
        \end{array}
        \right)
    \end{multline*} \caption{Full-read safe register process}\label{fig:fullread-procsafe}
\end{figure}

The correspondence between the process and the four rules given for MWMR safe registers in \cref{sec:safe} is rather direct: the first and fourth summands allow a thread that is not currently reading or writing to begin a read or a write, and the remaining four each represent one of the four rules, in order. 
Note that, in the case of a finish write without overlap, we use $\valssym$ to retrieve which value this thread intended to write so that the register state can be appropriately updated.

\paragraph{Full-read regular MWMR registers.}\label{app:fullread-regular}
See \cref{fig:fullread-procregular} for the process equation representing our full-read regular register model.
Recall that regular registers use the order write action to generate a global ordering on all write operations on a register on the fly.

\begin{figure}[htb]
    \begin{multline*}
        \Reg[\frRep,\regRep](\rid: \RID, s:\StatusAll) = \\
        \sum_{\tid\in\TID}\sum_{\data \in \Data}
        \left(\begin{array}{ll}
            & (\tid\notin(\readers{s} \cup \writers{s})) \then
                 \startread[\tid,\rid]\co\Reg[\frRep,\regRep](\rid,\usr{s,\tid}) \\
          + & (\tid\in\readers{s} \land \data \in \posval{s,\tid}) \then
                 \finishread[\tid,\rid]{\data}\co\Reg[\frRep,\regRep](\rid,\ufr{s,\tid})\\
          + & (\tid\notin(\readers{s} \cup \writers{s})) \then
                 \startwrite[\tid,\rid]{\data}\co\Reg[\frRep,\regRep](\rid,\usw{s,\tid,\data})\\
          + & (\tid \in \writers{s} \land \tid\in\pending{s}) \then \orderwrite[\tid,\rid]\co\Reg[\frRep,\regRep](\rid,\uow{s,\tid,\vals{s,\tid}})\\
          + & (\tid\in\writers{s}\wedge \tid\notin\pending{s}) \then
                 \finishwrite[\tid,\rid]\co\Reg[\frRep,\regRep](\rid,\ufw{s,\tid,\trueval{s}})
        \end{array}
        \right)
    \end{multline*}
    \caption{Full-read regular register process}\label{fig:fullread-procregular}
\end{figure}

Similar to the full-read safe register process, the first and third summands are merely allowing an idle thread to begin a read or write operation. 
The second summand corresponds to finishing a read by returning a value that is in the set of possible values to be returned for this read.
Recall that, by the definition of $\posvalsym$, this set is constructed as follows: when a read starts, the set is initialised to the current stored value of the register and the intended write value of every active write. 
Subsequently, whenever a write occurs, its intended value is added to the set.
This way, at the finish read, the set will contain exactly those values that the read could return.
The fourth summand allows the occurrence of the order write action; at this time the intended value of the write, which was temporary stored in the access function $\valssym$, is logged as the stored value of the register. 
The final summand describes the ending of a write operation. The $\ufwsym$ update function will set the stored value to whatever data value is passed. In this case, the stored value should not change at the finish write, since it was already changed at the order write. Hence, we simply pass $\trueval{s}$.
Note that we use $\pendingsym$ to determine if the order write action still has to occur.

\paragraph{Full-read atomic MWMR registers.}\label{app:fullread-atomic}
See \cref{fig:fullread-procatomic} for our model of full-read MWMR atomic registers.
Recall that, in addition to the order write action, the atomic register model also uses the order read action. This way, it generates an ordering on all operations on a register.

\begin{figure}[ht]
    \begin{multline*}
        \Reg[\frRep,\atoRep](\rid:\RID, s:\StatusAll) = \\
        \sum_{\tid\in\TID}\sum_{\data \in \Data}
        \left(\begin{array}{ll}
            & (\tid\notin(\readers{s}\cup\writers{s})) \then
                 \startread[\tid,\rid]\co\Reg[\frRep,\atoRep](\rid,\usr{s,\tid}) \\
          + & (\tid\in\readers{s} \wedge \tid \in \pending{s}) \then\orderread[\tid,\rid]\co\Reg[\frRep,\atoRep](\rid,\uor{s,\tid}) \\
          + & (\tid\in\readers{s}\wedge\tid \notin \pending{s})\then \finishread[\tid,\rid]{\vals{s,\tid}}\co\Reg[\frRep,\atoRep](\rid,\ufr{s,\tid}) \\
          + & (\tid\notin(\readers{s}\cup\writers{s})) \then
                 \startwrite[\tid,\rid]{\data}\co\Reg[\frRep,\atoRep](\rid,\usw{s,\tid,\data})\\
          + & (\tid\in\writers{s}\wedge\tid \in \pending{s}) \then \orderwrite[\tid,\rid]\co\Reg[\frRep,\atoRep](\rid,\uow{s,\tid,\vals{s,\tid}}) \\
          + & (\tid\in\writers{s}\wedge\tid \notin \pending{s})\then\finishwrite[\tid,\rid]\co\Reg[\frRep,\atoRep](\rid,\ufw{s,\tid,\trueval{s}})
          \end{array}
        \right)
    \end{multline*}
    \caption{Full-read atomic register process}\label{fig:fullread-procatomic}
\end{figure}

The first three summands are the invocation, ordering and response of a read operation respectively.
Similarly, the last three summands are the invocation, ordering and response of a write operation.
We use $\pendingsym$ to determine whether an operation's order action still has to occur. When a read is ordered, we save the current stored value of the register in $\valssym$, so that this value can be returned when the read ends.
For writes, we already update $\truevalsym$ when the write is ordered, meaning that when the write ends we do not want to change $\truevalsym$ further and just keep the current value.

\subsection{Instant-read models}\label{app:instantread}
We now give the status object and process-algebraic models for the instant-read models. We focus on their differences with the full-read models.

\paragraph{Instant-read status object.}\label{app:instantread-status}
For the instant-read models, we use the same access functions as for the full-read models, with the following minor differences:

\begin{itemize}
    \item The function $\readerssym$ has been removed, since we no longer have a time between the invocation and response of a read.
    \item The function $\pendingsym$ is now only used for write operations, since read operations are never pending.
    \item The function $\valssym$ is still used to record a single data value per thread, but we no longer use it to remember which value a read should return in the regular register model.
    \item The predicate $\overlapsym$ is now only used for write operations, since read operations are never ``ongoing''.
    \item The function $\posvalsym$ has been removed, since we can now determine all possible values a read can return at the moment it occurs, rather than building up a set over time.
\end{itemize}

The initial state $\initstate$ is defined similar to how it is done for the full-read models, with the two removed access functions no longer needing to be initialised. I.e., for all $\tid \in \TID$ and a pre-defined initial value $\data_{\initstate}$:
\[\begin{array}{r@{~}c@{~}lr@{~}c@{~}lr@{~}c@{~}lr@{~}c@{~}l}
    \trueval{\initstate} &=& \data_{\initstate} & \writers{\initstate} &=& \emptyset & \vals{\initstate, \tid} &=& \data_{\initstate}  & &&\\
    &&& \pending{\initstate} &=& \emptyset & \overlap{\initstate, \tid} &=& \false & &\\[-1ex]
\end{array}\]

Since the instant-read models use different actions from the full-read models, the update functions $\usrsym, \uorsym$ and $\ufrsym$ do not exist; we instead have an update function $\ursym$. 
However, the $\ursym$ update function does not change the result of any access function; that a read has taken place has no impact on any other operations, and therefore needs not be remembered in any way.
The update functions $\uswsym$, $\ufwsym$ and $\uowsym$ have the same effects as in the full-read models except $\uswsym$ no longer updates $\posvalsym$; we refer to \cref{app:fullread-status} for their definitions.

\paragraph{Instant-read safe MWMR registers.}\label{app:instantread-safe}
See \cref{fig:instantread-procsafe} for the process equation representing our instant-read safe register model.

\begin{figure}[ht]
    \begin{multline*}
        \Reg[\irRep,\safRep](\rid: \RID, s:\StatusAll) = \\
        \sum_{\tid\in\TID}\sum_{\data \in \Data}
        \left(\begin{array}{ll}
            & (\tid\notin\writers{s}\wedge\writers{s} = \emptyset) \then
                 \readop[\tid,\rid]{\trueval{s}}\co\Reg[\irRep,\safRep](\rid,\ur{s,\tid}) \\
          + & (\tid\notin\writers{s}\wedge\writers{s} \neq \emptyset) \then
                 \readop[\tid,\rid]{\data}\co\Reg[\irRep,\safRep](\rid,\ur{s,\tid}) \\
          +  & (\tid\notin \writers{s}) \then
                 \startwrite[\tid,\rid]{\data}\co\Reg[\irRep,\safRep](\rid,\usw{s,\tid,\data})\\
          + & (\tid\in\writers{s}\wedge\neg\overlap{s,\tid}) \then
                 \finishwrite[\tid,\rid]\co\Reg[\irRep,\safRep](\rid,\ufw{s,\tid,\vals{s, \tid}})\\
          + & (\tid\in\writers{s}\wedge\overlap{s,\tid}) \then
                 \finishwrite[\tid,\rid]\co\Reg[\irRep,\safRep](\rid,\ufw{s,\tid,\data})
        \end{array}
        \right)
    \end{multline*} \caption{Instant-read safe register process}\label{fig:instantread-procsafe}
\end{figure}

The main difference with \cref{fig:fullread-procsafe} is that we do not use $\overlapsym$ to determine whether the read operation returns the stored value of the register or an arbitrary value. Instead, we check whether any writes are active at the moment the read takes place.

\paragraph{Instant-read regular MWMR registers.}\label{app:instantread-regular}
See \cref{fig:instantread-procregular} for the process equation representing our instant-read regular register model.

\begin{figure}[ht]
    \begin{multline*}
        \Reg[\irRep,\regRep](\rid: \RID, s:\StatusAll) = \\
        \sum_{\tid\in\TID}\sum_{\data \in \Data}
        \left(\begin{array}{ll}
          & (\tid\notin\writers{s} \land \data \in (\{\trueval{s}\} \cup \{d' \mid \exists_{t' \in \writers{s}}.\vals{s,t'} = d'\})) \then\\
          & \qquad \readop[\tid,\rid]{\data}\co\Reg[\irRep,\regRep](\rid,\ur{s,\tid})\\
          +  & (\tid\notin\writers{s}) \then
                 \startwrite[\tid,\rid]{\data}\co\Reg[\irRep,\regRep](\rid,\usw{s,\tid,\data})\\
          + & (\tid \in \writers{s} \land \tid\in\pending{s}) \then \orderwrite[\tid,\rid]\co\Reg[\irRep,\regRep](\rid,\uow{s,\tid,\vals{s,\tid}})\\
          + & (\tid\in\writers{s}\wedge \tid\notin\pending{s}) \then
                 \finishwrite[\tid,\rid]\co\Reg[\irRep,\regRep](\rid,\ufw{s,\tid,\trueval{s}})
        \end{array}
        \right)
    \end{multline*}
    \caption{Instant-read regular register process}\label{fig:instantread-procregular}
\end{figure}

When the read occurs, we require that the returned value is either the stored value of the register or a value that is being written by a currently active writer. This differs from \cref{fig:fullread-procregular}, where this set is how $\posvalsym$ is initialised when a thread $\tid$ starts a read, but the set is then subsequently grown when new writes start. Since the read in \cref{fig:instantread-procregular} happens instantaneously, we can immediately use the initial set of possible values.

\paragraph{Instant-read atomic MWMR registers.}\label{app:instantread-atomic}
See \cref{fig:instantread-procatomic} for our instant-read model of MWMR atomic registers.

\begin{figure}[hb]
    \begin{multline*}
        \Reg[\irRep, \atoRep](\rid:\RID, s:\StatusAll) = \\
        \sum_{\tid\in\TID}\sum_{\data \in \Data}
        \left(\begin{array}{ll}
           & (\tid\notin\writers{s})\then \readop[\tid,\rid]{\trueval{s}}\co\Reg[\irRep, \atoRep](\rid,\ur{s,t}) \\
          +  & (\tid\notin\writers{s}) \then
                 \startwrite[\tid,\rid]{\data}\co\Reg[\irRep, \atoRep](\rid,\usw{s,\tid,\data})\\
          + & (\tid\in\writers{s}\wedge\tid \in \pending{s}) \then \orderwrite[\tid,\rid]\co\Reg[\irRep, \atoRep](\rid,\uow{s,\tid,\vals{s,\tid}}) \\
          + & (\tid\in\writers{s}\wedge\tid \notin \pending{s})\then\finishwrite[\tid,\rid]\co\Reg[\irRep, \atoRep](\rid,\ufw{s,\tid,\trueval{s}})
          \end{array}
        \right)
    \end{multline*}
    \caption{Instant-read atomic register process}\label{fig:instantread-procatomic}
\end{figure}

Effectively, the read action here takes the place of the order read action in \cref{fig:fullread-procatomic}, but instead of remembering the currently stored value of the register for later, it is immediately returned.

\subsection{Additional results on thread-register models}\label{app:additionalresults}
In future proofs, the following results on thread-register models, which depend on the full definitions of the registers as provided in this appendix, will be used.

First, we provide an alternate perspective on the definitions of the access functions.
When it comes to applying these definitions, we often find it easier to reason from the perspective of transitions: given a transition $s \xrightarrow{a} s'$ in an LTS that is derived from a register model as described in \cref{app:fullread,app:instantread}, the relation of the output of an access function on $s'$ to the output on $s$, depending on the action $a$.
This perspective can be derived directly from the definitions of the update functions, and the tight coupling between update functions and actions.
We make this observation for both the full-read and instant-read models.
\begin{obs}\rm\label{obs:access-fullread}
Consider the LTS $(\states, \actionset, \initstate, \transrel)$ associated with a full-read register $\rid \in \RID$. Let $s, s' \in \states$ and $a \in \actionset$ such that $s \xrightarrow{a} s'$, and let $\tid \in \TID$. Then the following equations hold:
{\allowdisplaybreaks
\begin{align*}
    \trueval{s'} &= \begin{cases}
        \data & \text{if $a =\finishwrite[\tid,\rid]$ and $ s' = \ufw{s,\tid,\data}$ for $\tid \in \TID, \data \in \Data$}\\
        \data & \text{if $a =\orderwrite[\tid,\rid]$ and $s' = \uow{s,\tid,\data}$ for $\tid \in \TID, \data \in \Data$}\\
        \trueval{s} & \text{otherwise}
    \end{cases}\\
\hspace{-8pt}
    \readers{s'} &= \begin{cases}
        \readers{s} \cup \{\tid\} & \text{if $a =\startread[\tid,\rid]$ for $\tid \in \TID$}\\
        \readers{s} \setminus \{\tid\} & \text{if $a =\finishread[\tid,\rid]{\data}$ for $\tid \in \TID, \data \in \Data$}\\
        \readers{s} & \text{otherwise}
    \end{cases}\\
\hspace{-8pt}
    \writers{s'} &= \begin{cases}
        \writers{s} \cup \{\tid\} & \text{if $a = \startwrite[\tid,\rid]{\data}$ for $\tid \in \TID, \data \in \Data$}\\
        \writers{s} \setminus \{\tid\} & \text{if $a = \finishwrite[\tid,\rid]$ for $\tid \in \TID$}\\
        \writers{s} & \text{otherwise}
    \end{cases}\\
\hspace{-8pt}
    \pending{s'} &= \begin{cases}
        \pending{s} \cup \{\tid\} & \text{if $a = \startread[\tid,\rid]$ or $a = \startwrite[\tid,\rid]{\data}$ for $\tid \in \TID, \data \in \Data$}\\
        \pending{s} \setminus \{\tid\} & \text{if $a = \orderread[\tid,\rid]$ or $a = \orderwrite[\tid,\rid]$ for $\tid \in \TID$}\\
        \pending{s} & \text{otherwise}
    \end{cases}\\
\hspace{-8pt}
    \vals{s', \tid} &= \begin{cases}
        \data & \text{if $a = \startwrite[\tid,\rid]{\data}$ for $\data \in \Data$}\\
        \trueval{s} & \text{if $a = \orderread[\tid,\rid]$}\\
        \vals{s, \tid} & \text{otherwise}        
    \end{cases}\\
\hspace{-8pt}
    \overlap{s', \tid} &= \begin{cases}
        (\writers{s} > 0) & \text{if $a = \startread[\tid,\rid]$ or $a = \startwrite[\tid,\rid]{\data}$ for $\data \in \Data$}\\
        \true & \text{if $a = \startwrite[\tidtwo,\rid]{\data}$ for $\tidtwo\in\TID, \tidtwo \neq \tid$ and $\data \in \Data$}\\
        \overlap{s, \tid} & \text{otherwise}
    \end{cases}\\
\hspace{-8pt}
    \posval{s', \tid} &= \begin{cases}
        \{\trueval{s}\} \cup \{\data \mid \exists_{\tidtwo \in \writers{s}}:\vals{s, \tidtwo} = d\} & \text{if $a = \startread[\tid,\rid]$}\\
        \posval{s, \tid} \cup \{\data\} & \hspace{-35pt}\text{if $a = \startwrite[\tidtwo,\rid]{\data}$ for $\tidtwo\in\TID, \tid\neq\tidtwo, \data \in \Data$}\\
        \posval{s, \tid} & \hspace{-35pt}\text{otherwise}
    \end{cases}
\end{align*}
}
\end{obs}

\begin{obs}\rm\label{obs:access-instantread}
Consider the LTS $(\states, \actionset, \initstate, \transrel)$ associated with an instant-read register $\rid \in \RID$. Let $s, s' \in \states$ and $a \in \actionset$ such that $s \xrightarrow{a} s'$, and let $\tid \in \TID$. Then the following equations hold:
{\allowdisplaybreaks
\begin{align*}
    \trueval{s'} &= \begin{cases}
        \data & \text{if $a =\finishwrite[\tid,\rid]$ and $ s' = \ufw{s,\tid,\data}$ for $\tid \in \TID, \data \in \Data$}\\
        \data & \text{if $a =\orderwrite[\tid,\rid]$ and $s' = \uow{s,\tid,\data}$ for $\tid \in \TID, \data \in \Data$}\\
        \trueval{s} & \text{otherwise}
    \end{cases}\\
\hspace{-8pt}
    \writers{s'} &= \begin{cases}
        \writers{s} \cup \{\tid\} & \text{if $a = \startwrite[\tid,\rid]{\data}$ for $\tid \in \TID, \data \in \Data$}\\
        \writers{s} \setminus \{\tid\} & \text{if $a = \finishwrite[\tid,\rid]$ for $\tid \in \TID$}\\
        \writers{s} & \text{otherwise}
    \end{cases}\\
\hspace{-8pt}
    \pending{s'} &= \begin{cases}
        \pending{s} \cup \{\tid\} & \text{if $a = \startwrite[\tid,\rid]{\data}$ for $\tid \in \TID, \data \in \Data$}\\
        \pending{s} \setminus \{\tid\} & \text{if $a = \orderwrite[\tid,\rid]$ for $\tid \in \TID$}\\
        \pending{s} & \text{otherwise}
    \end{cases}\\
\hspace{-8pt}
    \vals{s', \tid} &= \begin{cases}
        \data & \text{if $a = \startwrite[\tid,\rid]{\data}$ for $\data \in \Data$}\\
        \vals{s, \tid} & \text{otherwise}        
    \end{cases}\\
\hspace{-8pt}
    \overlap{s', \tid} &= \begin{cases}
        (\writers{s} > 0) & \text{if $a = \startwrite[\tid,\rid]{\data}$ for $\data \in \Data$}\\
        \true & \text{if $a = \startwrite[\tidtwo,\rid]{\data}$ for $\tidtwo\in\TID, \tidtwo \neq \tid$ and $\data \in \Data$}\\
        \overlap{s, \tid} & \text{otherwise}
    \end{cases}
\end{align*}
}
\end{obs}

Next, we consider when two states of a register model are effectively identical. We prove for the atomic register model that in many cases, two actions that follow each other can be swapped around, while ending up in an equivalent state.

\begin{defi}
Consider the LTS $(\states, \actionset, \initstate, \transrel)$ associated with a full-read register $\rid \in \RID$. For two states $s,s' \in \states$ write $s \cong s'$ if the access functions that are relevant for $\rid$ cannot tell the difference between them, that is: if $\trueval{s} = \trueval{s'}$, $\readers{s} = \readers{s'}$, $\writers{s} = \writers{s'}$, for all $t\in\TID$ one has $\vals{s,\tid} = \vals{s',\tid}$, and:
\begin{itemize}
    \item If $\rid$ is a safe register, $\overlap{s,\tid} \Leftrightarrow \overlap{s',\tid}$ for all $t\in\TID$.
    \item If $\rid$ is a regular register, then $\pending{s}=\pending{s'}$, and  $\posval{s,\tid} = \posval{s',\tid}$ for all $t\in\TID$.
    \item If $\rid$ is an atomic register, then $\pending{s}=\pending{s'}$.
\end{itemize}
\noindent
Now consider a full-read thread-register model $(\states, \actionset, \initstate, \transrel, \thrsym, \regsym)$.
For two states $s=(s_1,\dots,s_k)$ and $s'=(s'_1,\dots,s'_k)$ write $s \cong s'$ if $s_{\idx{\tid}}=s'_{\idx{\tid}}$
for each thread $t \in\TID$ and \plat{$s_{\idx{\rid}}\cong s'_{\idx{\rid}}$} for each register $r\in\RID$.
\end{defi}
\noindent
In case $s \cong s'$, these states are for our purposes indistinguishable. 
\begin{obs}\rm\label{obs:trivial-cong}
    Trivially, $s_1 \cong s'_1$ and $s_1\xrightarrow{a} s_2$ implies that $s'_1\xrightarrow{a} s'_2$ for some $s'_2$ with $s_2 \cong s'_2$.
\end{obs}

\begin{lem}\rm\label{lem:swapping}
  Consider the LTS $(\states, \actionset, \initstate, \transrel)$\vspace{-2pt} associated with an atomic full-read register $\rid \mathbin\in \RID$. Let $s_0, s_1, s_2 \in \states$ and $a,b \in \actionset$ such that $s_0 \xrightarrow{a} s_1 \xrightarrow{b} s_2$ and $\thrmap{a} \neq \thrmap{b}$.\linebreak[3] Discard the options
  (i) $a \mathbin= \orderread[\tid,\rid] \wedge b \mathbin= \orderwrite[\tid',\rid]$,
  (ii) $a \mathbin= \orderwrite[\tid,\rid] \wedge b \mathbin= \orderread[\tid',\rid]$ and
  (iii) $a \mathbin= \orderwrite[\tid,\rid] \wedge b \mathbin= \orderwrite[\tid',\rid]$\vspace{-2pt} for some $t,t'\in \TID$.
Then there exists $s_3,s_4 \in \states$ such that $s_0 \xrightarrow{b} s_3 \xrightarrow{a} s_4$ and $s_2 \cong s_4$.
\end{lem}

\begin{proof}
Let $t=\thrmap{a}$ and $t'=\thrmap{b}$.
By inspection of \cref{fig:fullread-procatomic} one sees that whether the action $b$ is enabled in a state $s \in \states$ depends on a Boolean combination of the conditions $t' \in\readers{s}$, $t' \in\writers{s}$ and $t' \in\pending{s}$---except in the case that $b$ is of the form $\finishread[\tid',\rid]{d}$, in which case it also depends on the value of $\vals{s,\tid'}$. Since $\thrmap{a}\neq t'$ we have $\vals{s_0,\tid'}=\vals{s_1,\tid'}$ as well as
$t' \in\readers{s_0} \Leftrightarrow t' \in\readers{s_1}$---see \cref{obs:access-fullread}---and likewise for $\writerssym$ and $\pendingsym$.
Since $b$ is enabled in $s_1$, it is thus also enabled in $s_0$, and hence there exists a transition \plat{$s_0 \xrightarrow{b} s_3$}. By \cref{lem:thr-consist} (thread-consistency) $a$ must be enabled in $s_3$, so there exists a transition \plat{$s_3 \xrightarrow{a} s_4$}.

Since $\thrmap{a} \neq \thrmap{b}$, performing the update functions associated with $a$ and $b$ in either order turns out to make no difference to the relevant access functions: it follows directly from \cref{obs:access-fullread} that $\readers{s_2} = \readers{s_4}$, $\writers{s_2} = \writers{s_4}$ and $\pending{s_2}=\pending{s_4}$.
A careful case distinction on $a$ and $b$ combined with \cref{obs:access-fullread} and \cref{fig:fullread-procatomic} shows that
$\trueval{s_2}=\trueval{s_4}$, and that
$\vals{s_2,\tid''}=\vals{s_4,\tid''}$.
\end{proof}
\noindent
In this proof the statement $\trueval{s_2}=\trueval{s_4}$ would fail if option (iii) were allowed,
the statement $\vals{s_2,\tid}=\vals{s_4,\tid}$ if option (i) were allowed, and
the statement $\vals{s_2,\tid'}=\vals{s_4,\tid'}$ if option (ii) were allowed.

\begin{prop}\rm\label{prop:swapping}
  Consider a full-read thread-register model $(\states, \actionset, \initstate, \transrel, \thrsym, \regsym)$\vspace{-2pt} with only atomic registers. Let $s_0, s_1, s_2 \in \states$ and $a,b \in \actionset$ such that $s_0 \xrightarrow{a} s_1 \xrightarrow{b} s_2$ and $\thrmap{a} \neq \thrmap{b}$. Discard the options
  (i) $a = \orderread[\tid,\rid] \wedge b = \orderwrite[\tid',\rid]$,
  (ii) $a = \orderwrite[\tid,\rid] \wedge b = \orderread[\tid',\rid]$ and
  (iii) $a = \orderwrite[\tid,\rid] \wedge b = \orderwrite[\tid',\rid]$ for some $t,t'\in \TID$.\vspace{-2pt}
Then there exists $s_3,s_4 \in \states$ such that $s_0 \xrightarrow{b} s_3 \xrightarrow{a} s_4$ and $s_2 \cong s_4$.
\end{prop}
\begin{proof}
  In case $\regmap{a} = \regmap{b} = \bot$ or $\regmap{a} \neq \regmap{b}$, the statement follows trivially.
  In case $\regmap{a} = \regmap{b} \in\RID$ the statement arises as a corollary of \cref{lem:swapping}.
\end{proof}

\section{Transforming full-read thread LTSs into instant-read ones}\label{app:thread-transform}

As stated in \cref{sec:threads} and \cref{sec:thrregmodels}, we leave the exact modelling of threads as LTSs mostly open, so as to allow a wide variety of algorithms to be modelled. We only impose a number of restrictions to ensure they behave reasonably.

It is important, however, that the full-read and instant-read variants of a thread LTS are identical outside of the different approach to modelling read operations.
In this section, we describe how a full-read thread LTS can be transformed into its instant-read version.

Let $T_{\frRep} = (\states_{\frRep}, \actionset_{\frRep}, \initstate_{\frRep},\transrel_{\frRep})$ be an LTS modelling a thread $\tid\in\TID$ in the full-read approach. We create the instant-read variant $T_{\irRep}= (\states_{\irRep}, \actionset_{\irRep}, \initstate_{\irRep},\transrel_{\irRep})$. Here, $\actionset_{\irRep} = \actionset_{\frRep} \setminus \{\startread[\tid,\rid],\finishread[\tid,\rid]{\data} \mid \rid\in\RID,\,\data\in\Data[\rid]\} \cup \{\readop[\tid,\rid]{\data}\mid\rid\in\RID,\,\data\in\Data[\rid]\}$ and $\initstate_{\irRep} = \initstate_{\frRep}$.
To create $\states_{\irRep}$ and $\transrel{\irRep}$, start with $\states_{\frRep}$ and $\transrel_{\frRep}$ and execute the following steps for each transition $u\in\transrel_{\frRep}$ with label $\startread[\tid,\rid]$ for some $\rid\in\RID$:
\begin{itemize}
    \item Let $s$ be the source of $u$ and $s'$ the target. By \cref{thread-prop-2} of \cref{sec:threads}, $s'$ enables only transitions labelled with $\finishread[\tid,\rid]{\data}$ and does so for all $\data\in\Data[\rid]$.
    \item Drop each transition $(s', \finishread[\tid,\rid]{\data}, s'')$ and add transitions $(s, \readop[\tid,\rid]{\data},s'')$.
    \item Drop $u$ and $s'$ from $\transrel_{\irRep}$ and $\states_{\irRep}$ respectively.
\end{itemize}
Effectively, the target of transition $u$ is merged into its source, and then the start read transition is dropped while the finish read transitions are relabelled to read transitions.

This transformation can be reversed to obtain the full-read variant of a thread LTS back from its instant-read version.
Let $T_{\irRep} = (\states_{\irRep}, \actionset_{\irRep}, \initstate_{\irRep},\transrel_{\irRep})$ be an LTS modelling a thread $\tid\in\TID$ in the instant-read approach. We create full-read variant $T_{\frRep}= (\states_{\frRep}, \actionset_{\frRep}, \initstate_{\frRep},\transrel_{\frRep})$. Here, $\actionset_{\frRep} = \actionset_{\irRep} \setminus \{\readop[\tid,\rid]{\data}\mid\rid\in\RID,\,\data\in\Data[\rid]\} \cup\linebreak[3] \{\startread[\tid,\rid],\finishread[\tid,\rid]{\data} \mid
\rid\in\RID,\,\data\in\Data[\rid]\}$ and $\initstate_{\frRep} = \initstate_{\irRep}$.
To create $\states_{\frRep}$ and $\transrel_{\frRep}$, start with $\states_{\irRep}$ and $\transrel_{\irRep}$ and execute the following steps for each state $s\in\states_{\irRep}$ and each register $\rid\in\RID$ such that $s$ enables an action $\readop[\tid,\rid]{\data}$ for some (and thus all) $\data\in\Data[\rid]$:
\begin{itemize}
    \item Add a fresh state $s_r$ and a transition $(s,\startread[\tid,\rid],s_r)$.
    \item Replace each transition $(s, \readop[\tid,\rid]{\data},s'')$ by
      $(s_r, \finishread[\tid,\rid]{\data}, s'')$.
\end{itemize}

From these transformations, we can derive a mapping from states in the model of one type into sets of states in the model of another type; we name these mappings $\itofstate$ and $\ftoistate$.
\begin{itemize}
    \item To go from a full-read to an instant-read LTS, we only drop states. Thus, every state that is in $\states_{\irRep}$ also exists in $\states_{\frRep}$. Given a state $s\in\states_{\irRep}$, let $\itofstate(s)$ be the singleton set containing only that same state $s\in\states_{\frRep}$.
    \item To go from an instant-read to a full-read LTS, we only add new states. Given a state $s\in\states_{\frRep}$ that also exists in $\states_{\irRep}$, let $\ftoistate(s)$ be the singleton set containing only that same state. Given a state $s\in\states_{\frRep}$ such that $s\notin\states_{\irRep}$, there is only a single transition $(s_b,\startread[\tid,\rid],s)$, for some $\rid\in\RID$, in $\transrel_{\frRep}$ that has $s$ as its target. Additionally, $s$ enables exclusively finish read actions, let $S_a$ be the set of all targets of all transitions enabled in $s$. Note that $\{s_b\} \cup S_a \subseteq  \states_{\irRep}$. Let $\ftoistate(s)$ be $\{s_b\} \cup S_a$.
\end{itemize}

When working with these transformations, it is useful to establish that certain properties are preserved, specifically with regards to which actions are enabled.

\begin{lem}\rm\label{lem:itof-preserve}
    Let $s$ be a state in $T_{\irRep}$ and let $s'$ be a state in $T_{\frRep}$ such that $s' \in \itofstate(s)$.
    If $s$ enables an action $a \notin \{\readop[\tid,\rid]{\data}\mid \tid\in\TID,\rid\in\RID,\data\in\Data[\rid]\}$ then $s'$ also enables $a$.
    Additionally, if $s'$ enables an action $a \notin\{\startread[\tid,\rid],\finishread[\tid,\rid]{\data}\mid\tid\in\TID,\rid\in\RID,\data\in\Data[\rid]\}$ then $s$ also enables $a$.
\end{lem}
\begin{proof}
    Note that the $\itofstate(s)$ is a singleton set, thus we know exactly which state $s'$ is: the state that becomes $s$ when we transform $T_{\frRep}$ into $T_{\irRep}$ using the transformation described above.
    In this transformation, the only transitions that are dropped are ones labelled with start read and finish read actions, and the only transitions that are added are those with read actions.
    Thus, for all other actions, the transitions are preserved.
    Hence when an action $a \notin \{\startread[\tid,\rid],\finishread[\tid,\rid]{\data},\readop[\tid,\rid]{\data}\mid\tid\in\TID,\rid\in\RID,\data\in\Data[\rid]\}$ is enabled in state $s$ (resp. $s'$), it is also enabled in $s'$ (resp. $s$).
\end{proof}

\begin{lem}\rm\label{lem:itof-read}
    Let $s$ be a state in $T_{\irRep}$ and let $s'$ be a state in $T_{\frRep}$ such that $s' \in \itofstate(s)$.
    If $s$ enables an action $a = \readop[\tid,\rid]{\data}$ for some $\tid\in\tid,\rid\in\RID,\data\in\Data[\rid]$ then $s'$ enables a transition $(s', \startread[\tid,\rid],s'')$ and $s''$ enables $\finishread[\tid,\rid]{\data}$.
    Additionally, if $s'$ enables an action $a = \startread[\tid,\rid]$ for some $\tid\in\TID,\rid\in\RID$, then $s$ enables $\readop[\tid,\rid]{\data}$ for all $\data\in\Data[\rid]$.
\end{lem}
\begin{proof}
    Note that the $\itofstate(s)$ is a singleton set, thus we know exactly which state $s'$ is: the state that becomes $s$ when we transform $T_{\frRep}$ into $T_{\irRep}$ using the transformation described above.

    By this transformation, if $s'$ enables $\startread[\tid,\rid]$ then it is preserved, but the $\startread[\tid,\rid]$ transition is replaced with transitions labelled with $\readop[\tid,\rid]{\data}$ for all $\data\in\Data[\rid]$. Thus, if $s'$ enables $\startread[\tid,\rid]$ then $s$ enables $\readop[\tid,\rid]{\data}$ for all $\data\in\Data[\rid]$.

    Note that the state originally reached by $\startread[\tid,\rid]$ enabled $\finishread[\tid,\rid]{\data}$ for all $\data\in\Data[\rid]$; thus also proving that when $s$ enables $\readop[\tid,\rid]{\data}$, $s'$ enables $\startread[\tid,\rid]$ and the state after enables $\finishread[\tid,\rid]{\data}$.
\end{proof}

\begin{lem}\rm\label{lem:ftoi-preserve}
    Let $s$ be a state in $T_{\frRep}$ and let $s'$ be a state in $T_{\irRep}$ such that $s' \in \ftoistate(s)$.
    If $s$ enables an action $a \notin\{\startread[\tid,\rid],\finishread[\tid,\rid]{\data}\mid\tid\in\TID,\rid\in\RID,\data\in\Data[\rid]\}$ then $s'$ also enables $a$.
    Additionally, if $s'$ enables an action $a \notin \{\readop[\tid,\rid]{\data}\mid \tid\in\TID,\rid\in\RID,\data\in\Data[\rid]\}$ then $s$ also enables $a$ or it enables some action $\finishread[\tid,\rid]{\data}$.
\end{lem}
\begin{proof}
    The mapping $\ftoistate(s)$ is a singleton set unless $s$ does not also exist in $\states_{\irRep}$. The states that exist in $\states_{\frRep}$ but not $\states_{\irRep}$ are exactly those states that are reached by a start read action and enable finish read actions.

    If $s$ enables an action $a \notin\{\startread[\tid,\rid],\finishread[\tid,\rid]{\data}\mid\tid\in\TID,\rid\in\RID,\data\in\Data[\rid]\}$ then $s \in \states_{\irRep}$. Thus, if we transform $T_{\irRep}$ into $T_{\frRep}$ using the transformation above, $s'$ becomes $s$. Since this transformation only drops states labelled with read actions, all actions that are enabled in $s$ that are not start or finish read actions are also enabled in $s'$.

    If $s'$ enables an action $a \notin \{\readop[\tid,\rid]{\data}\mid \tid\in\TID,\rid\in\RID,\data\in\Data[\rid]\}$, then both $s\in\states_{\irRep}$ and $s\notin\states_{\irRep}$ are possible.
    In the former case, $\ftoistate(s)$ is a singleton set and a similar argument as above gives us that $a$ is enabled in $s$.
    In the latter case, $s$ is a state that enables exclusively finish read actions, so it enables some action $\finishread[\tid,\rid]{\data}$.
\end{proof}

\begin{lem}\rm\label{lem:ftoi-read}
    Let $s$ be a state in $T_{\frRep}$ and let $s'$ be a state in $T_{\irRep}$ such that $s' \in \ftoistate(s)$.
     If $s$ enables an action $a=\startread[\tid,\rid]$ for some $\tid\in\TID,\rid\in\RID$, then $s'$ enables $\readop[\tid,\rid]{\data}$ for all $\data\in\Data[\rid]$.
     If $s'$ enables an action $a=\readop[\tid,\rid]{\data}$ for some $\tid\in\TID,\rid\in\RID,\data\in\Data[\rid]$, then $s$ enables $\startread[\tid,\rid]$ or $\finishread[\tid,\rid]{\data}$.
\end{lem}
\begin{proof}
    The mapping $\ftoistate(s)$ is a singleton set unless $s$ does not also exist in $\states_{\irRep}$. The states that exist in $\states_{\frRep}$ but not $\states_{\irRep}$ are exactly those states that are reached by a start read action and enable finish read actions.

    If $s$ enables an action $\startread[\tid,\rid]$, then it does not enable finish read actions because a thread cannot perform two operations at the same time. Thus, $s\in\states_{\irRep}$ and $s'$ is the same state in $T_{\irRep}$. In the transformation from $T_{\irRep}$ into $T_{\frRep}$, start read transitions are added in place of read transitions, thus $s'$ enables $\readop[\tid,\rid]{\data}$ for all $\data\in\Data[\rid]$.

    If $s'$ enables an action $\readop[\tid,\rid]{\data}$, then $s\in\states_{\irRep}$ and $s\notin\states_{\irRep}$ are both possible.
    In the former case, it is the state where $\startread[\tid,\rid]$ is enabled.
    In the latter case, it enables all $\finishread[\tid,\rid]{\data'}$ for all $\data'\in\Data[\rid]$, including $\data$.
\end{proof}

\section{Equivalence on actions for instant-reads}\label{app:equivalence}
In \cref{sec:justness}, we briefly argued that the equivalence relation on actions we use to define justness must equate all actions $\readop[\tid,\rid]{\data}$ and $\readop[\tid,\rid]{\data'}$ for $\tid\in\TID,\rid\in\RID,\data,\data'\in\Data[\rid]$.
Here, we expand on the argument that using $\mathit{id}$ for $\actequivsym$ would result in the instant-read models not being thread-consistent.

Consider a state $s$ in which the action $a= \readop[\tid,\rid]{\data}$ is enabled and a transition $(s, b, s')$ such that $\thrmap{b} \neq \thrmap{a} = \tid$. In all three instant-read register models, it is possible that $a$ is no longer enabled in $s'$:
\begin{itemize}
    \item 
    For the safe register model, consider $b = \finishwrite[\tid',\rid]$ where $\tid' \neq \tid$ and $\tid'$ is the only active writer to $\rid$ in $s$. That means that in state $s$ the value $\data$ that is read by $a = \readop[\tid,\rid]{\data}$ may be an arbitrary value in the domain of the register due to there being an overlapping write. However, in state $s'$ there are no more active writes to $\rid$, and hence there is only one value that $\tid$ can read from $\rid$ in $s'$, which need not be $\data$. Thus, it is possible for $a$ to be disabled in $s'$.
    \item 
    For the regular register model, consider $b = \orderwrite[\tid', \rid]$ where $\tid' \neq \tid$, $\tid'$ is the only active writer to $\rid$ in $s$, the value $\data$ is the stored value of the register in $s$, and the write by $\tid'$ is writing a different value $\data'$. Then $\tid$ can read both $\data$ and $\data'$ from $\rid$ in $s$, but only $\data'$ in $s'$. Hence, $a$ is disabled in $s'$.
    \item 
    For the atomic register model, consider the same scenario as described for the regular register model above. Then $\tid$ can read only $\data$ from $\rid$ in $s$, and only $\data'$ in $s'$. Hence, $a$ is disabled in $s'$.
\end{itemize}
\noindent
We prove in \cref{app:thrconsist} that by defining $\actequivsym$ as we do in \cref{sec:justness}, the instant-read models are thread-consistent.
From this it follows (as argued in \cref{app:concTproof} and \cref{obs:subset}) that our concurrency relations are valid for the instant-read models.

It may seem strange to equate reads that return different values, since the returned value is important to how a thread behaves afterwards.
However, note that we only use this equivalence relation for the purpose of determining which paths are just, and \cref{def:concT,def:concS,def:concI,def:concA} treat these equivalent actions identically regardless: if $a = \readop[\tid,\rid]{\data}$ and $b=\readop[\tid,\rid]{\data'}$ with $\data \neq \data'$, then we have $\thrmap{a} = \thrmap{b} = \tid$, $\isread{a} = \isread{b} = \true$, $\iswrite{a} = \iswrite{b} = \false$ and $\regmap{a}=\regmap{b} = \rid$.
It therefore holds for $C \in \{T,S,I,A\}$ that $\actequiv{a}{b} \land b \nconc_C c \land \actequiv{c}{d} \Rightarrow a \nconc_C d$ for all actions $a,b,c,d$.
Thus, as far as justness with our four concurrency relation is concerned, there is no difference between using $\actequivsym$ or $\mathit{id}$, beyond making the concurrency relations valid for our models.

\section{Thread-consistency proof}\label{app:thrconsist}

In this appendix, we prove \cref{lem:thr-consist} for both full-read and instant-read thread-register models.
We use the following lemma in our proof of \cref{lem:thr-consist} for full-read thread-register models.
\begin{lem}\rm\label{lem:overlap}
    Let $s$ be the state of a full-read safe register that is reachable from its initial state, and let $\tid \in \TID$ be a thread identifier. If $\tid \in \readers{s}$ and $\writers{s} \neq \emptyset$, then $\overlap{s,\tid}$.
\end{lem}
\begin{proof}
  Let $\rid$ be a full-read safe register process with LTS $(\states, \actionset, \initstate, \transrel)$ and let $s \in \states$ be a state reachable from $\initstate$. Let $\tid$ be an element of $\TID$ and assume that $\tid \in \readers{s}$ and $\writers{s} \neq \emptyset$. Then there exists $\tidtwo \in \TID$ such that $\tidtwo \in \writers{s}$. Since a thread cannot read and write simultaneously, we know that $\tid \neq \tidtwo$. We prove $\overlap{s,\tid}$. Recall that $\readers{\initstate} = \writers{\initstate} = \emptyset$. Since $\tid \in \readers{s}$ and $\tidtwo \in \writers{s}$, it follows from the definitions of $\readerssym$ and $\writerssym$ that $s \neq \initstate$ and that every path from $\initstate$ to $s$ must contain occurrences of $\startread[\tid,\rid]$ and $\startwrite[\tidtwo,\rid]{\data}$ for some $\data \in \Data$. Let $\pi$ be an arbitrary path from $\initstate$ to $s$, and consider the last occurrence $oc$ of one of these actions on $\pi$. Let $s_b \startread[\tid,\rid] s_a$ or $s_b \startwrite[\tidtwo,\rid]{\data} s_a$ be the fragment $\pi$ containing $oc$, and let $\pi'$ be suffix of $\pi$ starting right after $oc$; so $\pi'$ goes from $s_a$ to $s$ and has no occurrences of $\startread[\tid,\rid]$ and $\startwrite[\tidtwo,\rid]{\data'}$ for some $\data' \in \Data$.
Thus, in case $oc$ is an occurrence of $\startread[\tid,\rid]$, we have $\tidtwo \in \writers{s_b}$ by \cref{obs:access-fullread}.
If $\pi'$ would contain an occurrence of $\finishread[\tid,\rid]{\data'}$ for some $\data' \in \Data$, then $\tid \notin \readers{s}$, so there is no such occurrence.
By Property (\ref{thread-prop-2}) in \cref{sec:threads},
$\pi$ cannot have occurrences of $\startwrite[\tid,\rid]{\data'}$ for some $\data' \in \Data$ either.
So \cref{obs:access-fullread} implies that $\overlap{s_a,\tid} \Rightarrow \overlap{s,\tid}$. Thus it suffices to show $\overlap{s_a,\tid}$.
For both cases of $oc$ this follows immediately from the definition of $\overlapsym$ as given in \cref{obs:access-fullread}. \qedhere
\end{proof}

We prove \cref{lem:thr-consist} separately for full-read and instant-read models. From that, we can conclude that \cref{lem:thr-consist} holds in both cases.
The two proofs are highly similar, since the two modelling approaches are identical save for how they model reads. 
To reduce the repetitiveness of the proofs, we highlight the full-read specific parts in the full-read proof in \highlight{\colourname}. For the instant-read proof, we merely state how those parts differ.

\begin{lem}\rm\label{lem:thr-consist-fullread}
    Let $M = (\states,\actionset,\initstate,\transrel,\thrsym,\regsym)$ be a full-read thread-register model. Then the LTS $(\states, \actionset, \initstate,\transrel)$ is thread-consistent with respect to the mapping $\thrsym$, and $\actequivsym$ as defined in \cref{sec:justness}.
\end{lem}

\begin{proof}
    \hypertarget{thrconsist}{Let $M = (\states, \actionset, \initstate, \transrel, \thrsym, \regsym)$ be a \highlight{full-read} thread-register model constructed as the parallel composition $P_1 \parcomp \ldots \parcomp P_k$ with $k = |\TID| + |\RID|$, where $P_1$ to $P_{|\TID|}$ are thread LTSs and $P_{|\TID| + 1}$ to $P_k$ are register LTSs. Let $P_x = (\states_x, \actionset_x, \initstate_x, \transrel_x)$ for all $1 \leq x \leq k$. Let $\#$ be a mapping from thread and register id's to natural numbers, which yields the index of that thread or register in the parallel composition.}

    Let $s = (s_1, \ldots, s_k)$ be a state in $\states$ and $a$ an action in $\actionset$ such that $a$ is enabled in $s$ and there exists an action $b \in \actionset$ and a state $s' = (s_1', \ldots, s_k') \in \states$ such that $\thrmap{a} \neq \thrmap{b}$ and $(s, b, s') \in \transrel$. We prove that there exists an action $c$ such that $\actequiv{a}{c}$ and $c$ is enabled in $s'$.
    
    Let $\thrmap{a} = \tid$ and $\regmap{a} = \rid$, with $\tid \in \TID$ and $\rid \in \RID \cup \{\undefsymb\}$.
    From the definition of parallel composition, specifically the construction of $\transrel$, it follows that the target state of $(s, b, s')$ can only differ from $s$ for the parallel components that are involved in $b$: $P_{\idx{\thrmap{b}}}$ and, if $\regmap{b} \neq \undefsymb$, $P_{\idx{\regmap{b}}}$. Since $\thrmap{b} \neq \tid$, $s_{\idx{\tid}} = s'_{\idx{\tid}}$.
    We do a case distinction on whether $\regmap{a} = \undefsymb$. 
    \begin{itemize}
        \item If $\regmap{a} = \undefsymb$, then $a$ is in $\thrlocacts_{\tid}$. Since such an action is enabled in a state $s''$ iff it is enabled in $s''_{\idx{\tid}}$, and $s_{\idx{\tid}} \mathbin= s'_{\idx{\tid}}$, it follows that $a$ is enabled in $s'$. Trivially, $\actequiv{a}{a}$.

        \item If $\regmap{a} \neq \undefsymb$, then $a$ is not a thread-local action, and must be a register interface or register-local action. In both cases, it is an action in $\actionset_{\idx{\rid}}$. If it is a register-local action, it solely remains to prove that an action $c$ with $\actequiv{a}{c}$ is enabled in $s'_{\idx{\rid}}$; if it is a register interface action, we must also prove that $c$ is enabled in $s'_{\idx{\tid}}$. \highlight{In the full-read model, the only action equivalent to $a$ is $a$ itself, so we must show that $a$ is enabled in $s'_{\idx{\rid}}$ and, if $a$ is an interface action, $a$ is enabled in $s'_{\idx{\tid}}$.} \highlight{The latter follows immediately from the observation that $s_{\idx{\tid}} = s'_{\idx{\tid}}$ and $a$ being enabled in $s_{\idx{\tid}}$. In both cases, it therefore suffices to prove that \highlight{$a$} is enabled in $s'_{\idx{\rid}}$}.
        From the construction of $\transrel$ it follows that $a$ is enabled in $s_{\idx{\rid}}$. Note that if $\regmap{b} \neq \rid$, then $s_{\idx{\rid}} = s'_{\idx{\rid}}$ and so $a$ is trivially enabled in $s'_{\idx{\rid}}$. Therefore, we continue under the assumption that $\regmap{b} = \rid$. Since $b$ is the action label of a transition from $s$ to $s'$, and $\regmap{b} = \rid \neq \undefsymb$, we know that $(s_{\idx{\rid}}, b, s'_{\idx{\rid}}) \in \transrel_{\idx{\rid}}$. 
        We do a further case distinction on what action $a$ is:
        \begin{itemize}
            \item If \highlight{$a = \startread[\tid,\rid]$ or} $a=\startwrite[\tid,\rid]{\data}$ for some $\data \in \Data$, then for all three register models \highlight{(\cref{fig:fullread-procsafe,fig:fullread-procregular,fig:fullread-procatomic})}, it follows that \highlight{$\tid \notin (\readers{s_{\idx{\rid}}} \cup \writers{s_{\idx{\rid}}})$}. We need to show that \highlight{$\tid \notin (\readers{s'_{\idx{\rid}}} \cup \writers{s'_{\idx{\rid}}})$} as well. From $\thrmap{b} \neq \tid$ it follows that \highlight{$b \neq \startread[\tid,\rid]$ and} $b \neq \startwrite[\tid,\rid]{\data'}$ for any $\data' \in \Data$. We know by \highlight{\cref{obs:access-fullread}} that if \highlight{$\tid \notin \readers{s_{\idx{\rid}}}$ then $\tid \notin \readers{s'_{\idx{\rid}}}$ and that if} $\tid \notin \writers{s_{\idx{\rid}}}$ then $\tid \notin \writers{s'_{\idx{\rid}}}$. Therefore, \highlight{$\tid \notin (\readers{s'_{\idx{\rid}}} \cup \writers{s'_{\idx{rid}}})$}. 
            Hence, it follows for all three register models that $a$ is enabled in $s'_{\idx{\rid}}$ and thus enabled in $s'$. 

            \item \highlight{If $a = \finishread[\tid,\rid]{\data}$ for some $\data \in \Data$, then to prove that $a$ is enabled in $s'_{\idx{\rid}}$, we do a case distinction on which type of register $\rid$ is:}
            \begin{itemize}
                \item \highlight{If $\rid$ is a safe register, then it follows from $a$ being enabled in $s_{\idx{\rid}}$ that $\tid \in \readers{s_{\idx{\rid}}}$ and either $\neg \overlap{s_{\idx{\rid}}, \tid}$ and $\data =\trueval{s_{\idx{\rid}}}$, or $\overlap{s_{\idx{\rid}}, \tid}$. To show that $a$ is enabled in $s'_{\idx{\rid}}$, we need to show $\tid \in \readers{s'_{\idx{\rid}}}$, and $\data = \trueval{s'_{\idx{\rid}}}$ or $\overlap{s'_{\idx{\rid}}, \tid}$.\linebreak[3] From $\thrmap{b} \neq \tid$ it follows that $b \neq \finishread[\tid,\rid]{\data'}$ for some $\data' \in \Data$. Hence, it follows from \cref{obs:access-fullread} that if $\tid \in \readers{s_{\idx{\rid}}}$ then $\tid \in \readers{s'_{\idx{\rid}}}$. Additionally, if $\overlap{s_{\idx{\rid}}, \tid}$ then $\neg\overlap{s'_{\idx{\rid}}, \tid}$ would only be possible if $b = \startread[\tid,\rid]$ or $b = \startwrite[\tid,\rid]{\data'}$ for some $\data' \in \Data$. Since $\thrmap{b} \neq \tid$, we can conclude that this is not the case and therefore that if $\overlap{s_{\idx{\rid}}, \tid}$, then $\overlap{s'_{\idx{\rid}}, \tid}$. Hence, if $\overlap{s_{\idx{\rid}}, \tid}$ then $a$ is enabled in $s'_{\idx{\rid}}$ and thus $a$ is enabled in $s'$. We continue under the assumption that $\neg\overlap{s_{\idx{\rid}}, \tid}$, and thus that $\data= \trueval{s_{\idx{\rid}}}$. Assume towards a contradiction that $\neg\overlap{s'_{\idx{\rid}}, \tid}$ and $\data \neq \trueval{s'_{\idx{\rid}}}$. If $\trueval{s_{\idx{\rid}}} \neq \trueval{s'_{\idx{\rid}}}$, then $b$ must be $\finishwrite[\tidtwo,\rid]$ or $\orderwrite[\tidtwo,\rid]$ for some $\tidtwo \neq \tid$. Since there are no order write actions in the safe register model, it must be the case that $b$ is a finish write action. As $b$ is enabled in $s_{\idx{\rid}}$, this means that $\tidtwo \in \writers{s_{\idx{\rid}}} \neq \emptyset$. Since $\tid \in \readers{s_{\idx{\rid}}}$, it follows from \cref{lem:overlap} that $\overlap{ s_{\idx{\rid}}, \tid}$. This contradicts our assumption that $\neg\overlap{s_{\idx{\rid}}, \tid}$; thus we can conclude that we must have $\overlap{s'_{\idx{\rid}}, \tid}$ or $\data = \trueval{s'_{\idx{\rid}}}$, and consequently $a$ is enabled in $s'_{\idx{\rid}}$. Hence, $a$ is enabled in $s'$.}

                \item \highlight{If $\rid$ is a regular register, then it follows from $a$ being enabled in $s_{\idx{\rid}}$ that $\tid \in \readers{s_{\idx{\rid}}}$ and  $\data \in \posval{s_{\idx{\rid}}, \tid}$. To show $a$ is enabled in $s'_{\idx{\rid}}$, it suffices to show that $\tid \in \readers{s'_{\idx{\rid}}}$ and $\data \in \posval{s'_{\idx{\rid}}. \tid}$. From \cref{obs:access-fullread} and $\thrmap{b} \neq \tid$, it follows that if $\tid \in \readers{s_{\idx{\rid}}}$, then $\tid \in \readers{s'_{\idx{\rid}}}$. Additionally, as can be seen in \cref{obs:access-fullread}, the set $\posvalsym$ for a specific thread id only grows over time until it is reset when that thread starts a read. Therefore, it is only possible that $\data \in \posval{s_{\idx{\rid}}, \tid}$ and $\data \notin \posval{s'_{\idx{\rid}}, \tid}$ are both true if $b = \startread[\tid,\rid]$. Since $\thrmap{b} \neq \tid$, we can conclude that $\data \in \posval{s'_{\idx{\rid}}, \tid}$. Thus, $a$ is enabled in $s'_{\idx{\rid}}$ and therefore $a$ is enabled in $s'$.}

                \item \highlight{If $r$ is an atomic register, then it follows from $a$ being enabled in $s_{\idx{\rid}}$ that $\tid \in \readers{s_{\idx{\rid}}}$, $\tid \notin \pending{s_{\idx{\rid}}}$ and $\data = \vals{s_{\idx{\rid}}, \tid}$. From \cref{obs:access-fullread} it follows that when $\thrmap{b} \neq \tid$, then $\tid \in \readers{s_{\idx{\rid}}}$ implies $\tid \in \readers{s'_{\idx{\rid}}}$ and $\tid \notin \pending{s_{\idx{\rid}}}$ implies $\tid \notin \pending{s'_{\idx{\rid}}}$. Similarly, since $\thrmap{b} \neq \tid$, $\vals{s_{\idx{\rid}}, \tid} = \vals{s'_{\idx{\rid}}, \tid}$. Thus, $a$ is enabled in $s'_{\idx{\rid}}$ and therefore also in $s'$.}
            \end{itemize}
            \highlight{In all three cases, $a$ is enabled in $s'$. }

            \item If $a = \finishwrite[\tid,\rid]$, then what we need to show about $s'_{\idx{\rid}}$ to establish that $a$ is enabled differs depending on which type of register $\rid$ is. If $\rid$ is safe, then it suffices to show $\tid \in \writers{s'_{\idx{\rid}}}$. If $\rid$ is regular or atomic, we need to show $\tid \in \writers{s'_{\idx{\rid}}}$ and $\tid \notin \pending{s'_{\idx{\rid}}}$. From the definitions of these access function as given in \highlight{\cref{obs:access-fullread}}, and $\thrmap{b} \neq \tid$, it follows that $\tid \in \writers{s_{\idx{\rid}}}$ implies $\tid \in \writers{s'_{\idx{\rid}}}$ and $\tid \notin \pending{s_{\idx{\rid}}}$ implies $\tid \notin \pending{s'_{\idx{\rid}}}$. Consequently, if $a$ is enabled in $s_{\idx{\rid}}$ then it is also enabled in $s'_{\idx{\rid}}$. Thus, we can conclude that $a$ is enabled in $s'_{\idx{\rid}}$ and therefore also in $s'$. 

            \item If $a = \orderwrite[\tid,\rid]$, then $\rid$ could be a regular or atomic register. From the two models, we know that since $a$ is enabled in $s_{\idx{\rid}}$, we have $\tid \in \writers{s_{\idx{\rid}}}$ and $\tid \in \pending{s_{\idx{\rid}}}$. For both models, to show $a$ is enabled in $s'_{\idx{\rid}}$, it suffices to show $\tid \in \writers{s'_{\idx{\rid}}}$ and $\tid \in \pending{s'_{\idx{\rid}}}$. From \highlight{\cref{obs:access-fullread}} and $\thrmap{b} \neq \tid$, it follows that if $\tid \in \writers{s_{\idx{\rid}}}$ then $\tid \in \writers{s'_{\idx{\rid}}}$ and if $\tid \in \pending{s_{\idx{\rid}}}$ then  $\tid \in \pending{s'_{\idx{\rid}}}$. Thus, $a$ is enabled in $s'_{\idx{\rid}}$ and therefore $a$ is enabled in $s'$. 

            \item \highlight{If $a = \orderread[\tid,\rid]$ then $\rid$ must be an atomic register.  From the atomic register model and $a$ being enabled in $s_{\idx{\rid}}$, we know that $\tid \in \readers{s_{\idx{\rid}}}$ and $\tid \in \pending{s_{\idx{\rid}}}$. From the definitions of the access functions as given in \cref{obs:access-fullread} and $\thrmap{b}\neq \tid$, it follows that $\tid \in \readers{s_{\idx{\rid}}}$ implies $\tid \in \readers{s'_{\idx{\rid}}}$ and $\tid \in \pending{s_{\idx{\rid}}}$ implies $\tid \in \pending{s'_{\idx{\rid}}}$. Thus, when $a$ is enabled in $s_{\idx{\rid}}$ it is also enabled in $s'_{\idx{\rid}}$. We conclude that $a$ is enabled in $s'$.}
        \end{itemize}
    \end{itemize}
    We have shown in all cases that \highlight{$a$} is enabled in $s'$. Hence, our \highlight{full-read} thread-register models are thread-consistent with respect to the mapping $\thrsym$ and equivalence relation $\actequivsym$.
\end{proof}

\begin{lem}\rm\label{lem:thr-consist-instantread}
Let $M = (\states,\actionset,\initstate,\transrel,\thrsym,\regsym)$ be a instant-read thread-register model. Then the LTS $(\states, \actionset, \initstate,\transrel)$ is thread-consistent with respect to the mapping $\thrsym$, and $\actequivsym$ as defined in \cref{sec:justness}.
\end{lem}
\begin{proof}
    As announced, this proof is in many places identical to the proof of \cref{lem:thr-consist-fullread}. We highlight the differences.
    Firstly, any reference to full-read models should be replaced with \highlight{instant-read}. Additionally, references to \cref{obs:access-fullread} must be replaced with \highlight{\cref{obs:access-instantread}}.

    We can no longer conclude in the $\regmap{a} \neq \bot$ case that we must necessarily show that $a$ itself is enabled in $s'_{\idx{\rid}}$. Our final conclusion after covering the different cases for $a$ is therefore no longer that we have shown that $a$ is enabled in $s'$ in all cases, but rather that \highlight{$a$ or an action equivalent to $a$ is enabled in $s'$}.

We cannot hope to show that $a$ is enabled in $s'$, as we showed in \cref{app:equivalence} that this is not necessarily the case with the instant-read model. However, \highlight{if $a$ is not a read action, then we prove that $a$ is enabled in $s'$. It follows from $s_{\idx{\tid}} = s'_{\idx{\tid}}$ and $a$ being enabled in $s_{\idx{\tid}}$ that $a$ is enabled in $s'_{\idx{\tid}}$, so it is sufficient to prove $a$ is enabled in $s'_{\idx{\rid}}$}. 

    The cases that $a$ is a start read, order read or finish read action are not applicable to the instant-read case; instead, we must prove the case that $a$ is a read action.
    \highlight{If $a = \readop[\tid,\rid]{\data}$ for some $\data \in \Data[\rid]$ then we will show that there exists a $\data' \in \Data[\rid]$ such that action $c=\readop[\tid,\rid]{\data'}$ is enabled in $s'_{\idx{\tid}}$ and $s'_{\idx{\rid}}$. From \cref{thread-prop-2} of \cref{sec:threads}, $a$ being enabled in $s_{\idx{\tid}}$ and $\thrmap{b}\neq\tid$, it follows that $\readop[\tid,\rid]{\data'}$ is enabled in $s'_{\idx{\tid}} =s_{\idx{\tid}}$ for all $\data' \in \Data[\rid]$. Thus, it remains to prove that $c=\readop[\tid,\rid]{\data'}$ is enabled in $s'_{\idx{\rid}}$ for some $\data'\in \Data[\rid]$. 
    What we need to show to prove this depends on the type of register $\rid$: if $\rid$ is safe, we need to show that $\writers{s'_{\idx{\rid}}} \neq \emptyset$ or there exists a $\data'\in\Data[\rid]$ such that $\trueval{s'_{\idx{\rid}}} = \data'$ (\cref{fig:instantread-procsafe}); if $\rid$ is regular then we need to show there exists a $\data'\in\Data[\rid]$ such that $\data' \in \{\trueval{s'_{\idx{\rid}}}\} \cup \{\data'' \mid \exists_{\tid''\in\writers{s'_{\idx{\rid}}}}.\vals{s'_{\idx{\rid}},\tid''} = \data''\}$ (\cref{fig:instantread-procregular}); and if $\rid$ is atomic we need to show there exists $\data'\in\Data[\rid]$ such that $\trueval{s'_{\idx{\rid}}} = \data'$ (\cref{fig:instantread-procatomic}). Additionally, in all three cases we need to show $\tid \notin \writers{s'_{\idx{\rid}}}$.
    Note that the first condition, for safe, regular, and atomic registers, is satisfied if there exists a $\data' \in \Data[\rid]$ such that $\trueval{s'_{\idx{\rid}}} = \data'$. Since $\truevalsym$ is never undefined, this is trivially true.
    It remains to show that $\tid \notin \writers{s'_{\idx{\rid}}}$.     
    In all three register models, we know from $\readop[\tid,\rid]{\data}$ being enabled in $s_{\idx{\rid}}$ that $\tid \notin \writers{s_{\idx{\rid}}}$. Since $\thrmap{b} \neq \tid$, by \cref{obs:access-instantread}, $\tid \notin \writers{s_{\idx{\rid}}}$ implies $\tid \notin \writers{s'_{\idx{\rid}}}$. Thus, there exists some $\data' \in \Data[\rid]$ such that $\readop[\tid,\rid]{\data'}$ is enabled in $s'_{\idx{\rid}}$. 
    Hence, there exists an action $c$ that is equivalent to $a$ and is enabled in $s'$.}

    Finally, note that the proof that $a$ is enabled in $s'$ in the start write case references $\readerssym$, since this is part of the condition of start write actions being enabled in the full-read models. For the instant-read models, the condition on $\readerssym$ is not present. The argument for $\writerssym$ is identical. The reference to \cref{fig:fullread-procsafe,fig:fullread-procregular,fig:fullread-procatomic} must be replaced with \highlight{\cref{fig:instantread-procsafe,fig:instantread-procregular,fig:instantread-procatomic}}.
\end{proof}

\section{The thread interference relation is a concurrency relation}\label{app:concTproof}

\concTval*
\begin{proof}
To prove that $\conc_T$ is a concurrency relation, we need to prove the two properties of concurrency relations (\cref{def:conc}).
    \begin{enumerate}
        \item We must establish that $\actequivsym$ and $\conc_T$ are disjoint. Recall that $\actequivsym$ is defined as $\mathit{id} \cup \{(\readop[\tid,\rid]{\data},\readop[\tid,\rid]{\data'}) \mid  \tid\in\TID,\, \rid \in \RID,\, \data,\data' \in \Data[\rid]\}$. It is trivial that $\conc_T$ is irreflexive, since $\thrmap{a} = \thrmap{a}$ for all $a \in \actionset$. Additionally, since $\thrmap{\readop[\tid,\rid]{\data}} = \thrmap{\readop[\tid,\rid]{\data'}}$ for all $\tid\in\TID,\rid\in\RID,\data,\data'\in\Data[\rid]$, $\conc_T$ is indeed disjoint from $\actequivsym$.
        \item Let $a$ be an arbitrary action in $\actionset$ and let $s$ be an arbitrary state in $\states$. Assume that $a$ is enabled in $s$ and that there exists a path $\pi$ from $s$ to some $s'$ such that $a \conc_T b$ for all $b$ occurring on $\pi$. We prove that there exists an action $c$ such that $\actequiv{a}{c}$ and $c$ is enabled in $s'$. 
        We do induction on the length of $\pi$.
        \begin{itemize}
            \item For the base case, $|\pi| = 0$, we observe that if $\pi$ has length $0$, then $s = s'$ and hence trivially $a$ is enabled in $s'$. Trivially, $\actequiv{a}{a}$.
            \item Let $i \geq 0$ and assume that the claim holds for all paths from $s$ of length $i$; we prove that it also holds for all paths from $s$ of length $i + 1$. Let $\pi$ be a path from $s$ to some state $s' \in \states$ with $|\pi| = i + 1$ such that $a \conc_T b$ for all $b$ occurring on $\pi$. Then there exists a path $\pi'$ from $s$ to some state $s''$ of length $i$ such that $a \conc_T b$ for all $b$ occurring on $\pi'$, and $\pi = \pi'ds'$ for some $d \in \actionset$. Note that by assumption on $\pi$, $a \conc_T d$. By the induction hypothesis, some action $e$ is enabled in $s''$ such that $\actequiv{a}{e}$. Since $a \conc_T d$, it follows by definition of $\conc_T$ that  $\thrmap{a} \neq \thrmap{d}$. By our choice of $\actequivsym$, we have $\thrmap{e} = \thrmap{a}$ and thus $\thrmap{e} \neq \thrmap{d}$. Since $M$ is thread-consistent by \cref{lem:thr-consist}, it follows that there exists an action $c$ that is enabled in $s'$ such that $\actequiv{e}{c}$ and thus also $\actequiv{a}{c}$.
        \end{itemize}
        We conclude that the property holds for all paths $\pi$.\qedhere
    \end{enumerate}
\end{proof}

\section{An explicit characterisation of complete paths}\label{app:characterise}

In this appendix, we provide a characterisation of which paths in a thread-register model are complete with respect to justness and one of the four concurrency relations we consider. This characterisation is largely the same for full-read and instant-read models.

For a given mutual exclusion algorithm, let $M$ be the LTS of its full-read or instant-read thread-register model from \cref{sec:thrregmodels}, using one of the register models employed in this paper---cf.\ \cref{app:registers}. Recalling that $M$ is a parallel composition of thread and register processes, each state in $M$ is a tuple $s=(s_1,\dots,s_k)$, where each of the indices $i \in \{1,\dots,k\}$ corresponds to a thread $t\in\TID$ or a register $r\in\RID$; as in \hyperlink{thrconsist}{the proof of} \cref{lem:thr-consist-fullread} in
\cref{app:thrconsist} we denote the corresponding component $s_i$ as $s_{\idx{\tid}}$ or $s_{\idx{\rid}}$; it is a state in the LTS $T_t$ or $R_r$.

\begin{lem}\rm\label{lem:1}
If $s$ is a reachable state of $M$ such that $s_{\idx{r}}$ enables an action $\orderwrite[\tid,\rid]$ or $\orderread[\tid,\rid]$ in $R_r$, then $s_{\idx{t}}$ enables an action $\finishwrite[\tid,\rid]$ or $\finishread[\tid,\rid]{d}$, respectively, in $T_t$.
Moreover, if  $s_{\idx{r}}$ enables $c = \finishwrite[\tid,\rid]$ or $c = \finishread[\tid,\rid]{d}$, then also  $s_{\idx{t}}$ enables $c$, and thus $s$ enables $c$ in $M$.
\end{lem}
\begin{proof}
We will provide the proof for the cases that $s_{\idx{r}}$ enables an action $\orderwrite[\tid,\rid]$ or $\finishwrite[\tid,\rid]$. This case applies to both full-read and instant-read models. The case for read actions, which only applies to full-read models, proceeds along the exact same lines.

Given a register $r\in \RID$ and a thread $t\in \TID$, in any execution path for any of the register models from \cref{app:registers}, the actions $c$ with $\thrmap{c}=t$ and $\regmap{c}=r$ occur strictly in the order
$\startwrite[\tid,\rid]{d}$ -- $\orderwrite[\tid,\rid]$ -- $\finishwrite[\tid,\rid]$, with the occurrence of $\orderwrite[\tid,\rid]$ only being present in the regular and atomic models. This is a direct consequence of the conditions $\tid \in \writers{s}$ and $\tid \in \pending{s}$ in \cref{fig:fullread-procsafe,fig:fullread-procregular,fig:fullread-procatomic,fig:instantread-procsafe,fig:instantread-procregular,fig:instantread-procatomic}. Moreover, the actions $\startwrite[\tid,\rid]{d}$ for $d \in \Data$ and $\finishwrite[\tid,\rid]$ can occur only in synchronisation between the register $r$ and the thread $t$. Thus, if $s$ is a reachable state of $M$ such that $s_{\idx{r}}$ enables an action $\orderwrite[\tid,\rid]$ or $\finishwrite[\tid,\rid]$, then the last synchronisation between $r$ and $t$ must have been an action $\startwrite[\tid,\rid]{d}$.

In \cref{sec:threads}, we required that after a start write action the thread must enable only the appropriate finish write action (\cref{thread-prop-3}). Since the last synchronisation between $r$ and $t$ was on an action $\startwrite[\tid,\rid]{\data}$, this $\finishwrite[\tid,\rid]$ cannot have occurred yet. Hence $s_{\idx{t}}$ enables the action $\finishwrite[\tid,\rid]$.

In case $s_{\idx{r}}$ enables $\finishwrite[\tid,\rid]$, then (by the above) both $s_{\idx{r}}$ and $s_{\idx{t}}$ enable this action, and therefore it is also enabled by $s$.
\end{proof}

\begin{lem}\rm\label{lem:2}
If $s$ is a reachable state of $M$ such that $s_{\idx{t}}$ enables an action $\finishwrite[\tid,\rid]$ or $\finishread[\tid,\rid]{d}$, then either $\orderwrite[\tid,\rid]$ or $\orderread[\tid,\rid]$ or 
$\finishwrite[\tid,\rid]$ or $\finishread[\tid,\rid]{d'}$ for some $d'\in\Data$ is enabled by $s_{\idx{r}}$. 
\end{lem}
\begin{proof}
We give the argument only for the read actions, which is applicable only to the full-read model. The argument for the write actions, which is applicable to both the full-read and instant-read models, proceeds along the same lines.

In \cref{sec:threads}, we required that a thread can only enable a finish read immediately after starting that read (\cref{thread-prop-4}). As in $M$ the  $\startread[\tid,\rid]$-transition must have been a synchronisation between thread $t$ and register $r$, and the register cannot execute an action $\finishread[\tid,\rid]{\data}$ without synchronising with $t$,
in each of our register models we have $\tid \in \readers{s_{\idx{\rid}}}$. Consequently, by the register models in \cref{app:registers}, either $\orderread[\tid,\rid]$ or $\finishread[\tid,\rid]{d'}$ for some $d'\in\Data$ is enabled by $s_{\idx{r}}$.
\end{proof}

\begin{defi}\label{def:thread-enabled}\rm
Given any path $\pi$ starting in the initial state of $M$, and given any thread $t\in \TID$, if $\pi$ has a suffix $\pi'$ on which no actions $b$ with $\thrmap{b} = t$ occur, starting from a state $s$ of $\pi$, then $s'_{\idx{t}} = s_{\idx{t}}$ for all states $s'$ in $\pi'$ and we define $\finish_t (\pi) = s_{\idx{t}}$.

Call action $a\in Act$ \emph{thread-enabled} by $\pi$ if, for $t= \thrmap{a}$,
$\pi$ contains only finitely many actions $b$ with $\thrmap{b} = t$ and 
$a$ is enabled in the state $\finish_t(\pi)$ of the LTS $T_t$.
\end{defi}
\noindent
Recall from \cref{def:concT} that $a \nconc_T b$ iff $\thrmap{a}=\thrmap{b}$. Hence a path $\pi$ is \hyperlink{just}{$\block$-$\conc_T$-unjust} if, and only if, it has a suffix $\pi'$ such that an action $a \in\nonblock$ is enabled in the initial state of $\pi'$, but $\pi$ does not contain any action $b$ with $\thrmap{a}=\thrmap{b}$. Using this, we obtain the following characterisation of the $\block$-$\conc_T$-just paths of $M$.

\begin{prop}\rm\label{pr:T-just paths}
  A path $\pi$ starting in the initial state of $M$, is \hyperlink{just}{$\block$-$\conc_T$-just} if, and only if, $\pi$ thread-enables no actions $a \in \nonblock$.
\end{prop}
\begin{proof}
Suppose $\pi$ thread-enables an action $a \in\nonblock$. We have to show that $\pi$ is \hyperlink{just}{not just}. Let $t= \thrmap{a}$ and $r= \regmap{a}$.
Let $\pi'$ be a suffix of $\pi$ on which no actions $b$ with $\thrmap{b}=t$ occur. Let $s$ be the initial state of $\pi'$. So $s_{\idx{t}} = \finish_t(\pi)$ and $s_{\idx{t}}$ enables $a$.

In case $r \mathbin=\undefsymb$, as \plat{$s_{\idx{t}}$} enables $a$ and $a$ does not require synchronisation with any register, also $s$ enables $a$. As $\pi'$ does not contain actions $b$ with $a \nconc_T b$, the path $\pi$ is \hyperlink{just}{not just}.\footnote{When rereading this proof as part of the proof of \cref{pr:A-just paths} we also use that $\forall r\in\RID.~ a \notin \textit{start}(r)$.}

So assume that $r \in \RID$.
Based on the process-algebraic definitions from \cref{app:registers}, any state $s'_r$ of $r$ enables either
(i) $\startwrite[\tid,\rid]{d}$ for all $d\in\Data$ and either (ia) $\startread[\tid,\rid]$ or (ib) $\readop[\tid,\rid]{d'}$ for at least one $d' \in \Data$,
(ii) $\orderread[\tid,\rid]$ or $\orderwrite[\tid,\rid]$, or
(iii) $\finishwrite[\tid,\rid]$ or $\finishread[\tid,\rid]{d}$ for some $d \in \Data$. 

If state $s_{\idx{r}}$ in $R_r$ enables $c=\orderread[\tid,\rid]$ or $c=\orderwrite[\tid,\rid]$, also $s$ enables $c$, since $c$ does not require synchronisation with thread $t$. As $\pi'$ contains no actions $b$ with $c \nconc_T b$, using that $\thrmap{c}=t$, the path $\pi$ is \hyperlink{just}{not just}.\footnote{When rereading this proof as part of the proof of \cref{pr:A-just paths} we also use that $c \notin \textit{start}(r)$.}

In case state $s_{\idx{r}}$ in $R_r$ enables an action $c=\finishread[\tid,\rid]{d}$ or $c=\finishwrite[\tid,\rid]$, by \cref{lem:1} above also state $s$ enables $c$. As $\pi'$ contains no actions $b$ with $c \nconc_T b$, the path $\pi$ is \hyperlink{just}{not just}.$^{\rm \thefootnote}$

We may now restrict attention to case (i) above, that $s_{\idx{r}}$ enables $\startwrite[\tid,\rid]{d}$ for all $d \in \Data$ and either $\startread[\tid,\rid]$ or $\readop[\tid,\rid]{d'}$ for some $d'\in\Data$. 
In this setting we proceed with a case distinction on the action $a$, which must be of the form $\startread[\tid,\rid]$, $\startwrite[\tid,\rid]{d}$, $\finishread[\tid,\rid]{d}$, $\finishwrite[\tid,\rid]$ or $\readop[\tid,\rid]{d}$, since it is an action in LTS $T_{\tid}$ and $r \neq \undefsymb$. By \cref{lem:2}, the case that $a = \finishwrite[\tid,\rid]$ or $a= \finishread[\tid,\rid]{d}$ for some $d \in \Data$ is already subsumed by the cases considered above. Thus we may restrict attention to the case that $a = \startread[\tid,\rid]$, $a = \startwrite[\tid,\rid]{d}$ or $a = \readop[\tid,\rid]{d}$ for some $d \in \Data$. In the first two cases (using (i) above) also $s$ enables $a$. In the case that $a = \readop[\tid,\rid]{d}$, $s$ may not enable $a$ itself but it will enable $\readop[\tid,\rid]{d'}$ for some $d' \in \Data$, as follows by (i) and \cref{thread-prop-2} of \cref{sec:threads}; let this action be called $c$ and note that $\thrmap{a} = \thrmap{c}$. As $\pi'$ does not contain actions $b$ with $a \nconc_T b$ and thus also no actions $b$ with $c \nconc_T b$, the path $\pi$ is \hyperlink{just}{not just}.\footnote{When reusing this in the proof of \cref{pr:A-just paths}, we also use that $\pi'$ contains no actions $b \in \textit{start}(r)$.}

For the other direction, suppose that $\pi$ is not $\block$-$\conc_T$-just. We have to establish that $\pi$ thread-enables an action $a \in \nonblock$. Let $\pi'$ be a suffix of $\pi$ such that an action $a \in \nonblock$ is enabled in the initial state $s$ of $\pi'$, but $\pi'$ contains no action $b$ such that $a \nconc_T b$. Let $t=\thrmap{a}$. Then $\pi'$ contains no action $b$ with $\thrmap{b}=t$. In case $a=\orderwrite[\tid,\rid]$ or $a=\orderread[\tid,\rid]$, then by \cref{lem:1} above \plat{$\finish_t(\pi)=s_{\idx{t}}$} enables an action $c = \finishwrite[\tid,\rid]$ or $c = \finishread[\tid,\rid]{d}$, and thus $\pi$ thread-enables the action $c \in \nonblock$. Otherwise, $\finish_t(\pi)=s_{\idx{t}}$ enables $a$ and thus $\pi$ thread-enables $a \in \nonblock$.
\end{proof}
Interestingly, this explicit characterisation of the \hyperlink{just}{$\block$-$\conc_T$-just} paths in our thread-register models is independent on whether we employ safe, regular or atomic registers.

Next, we provide similar characterisations of the $\block$-$\conc_C$-\hyperlink{just}{just} paths, for $C \in \{A, I, S\}$.
For these definitions, we use the following shorthand notation: given a register $r\in\RID$, $\textit{start}(r) = \{\startread[\tid,\rid],\startwrite[\tid,\rid]{d},\readop[\tid,\rid]{d}\mid\tid\in\TID,d\in\Data\}$.
Recall from \cref{def:concT,def:concS,def:concI,def:concA}, as adapted in \cref{sec:justness}, that $a \nconc_{\!\!A} b$ iff either $\thrmap{a}=\thrmap{b}$ or $a,b\in\textit{start}(r)$ for some $r\in\RID$. Hence a path $\pi$ is $\block$-$\conc_{\!\!A}$-unjust if, and only if, it has a suffix $\pi'$ such that an action $a \in\nonblock$ is enabled in the initial state of $\pi'$, but $\pi'$ does not contain any action $b$ such that $\thrmap{a}=\thrmap{b}$ or $a,b\in\textit{start}(r)$ for some $r\in\RID$.

\begin{prop}\rm\label{pr:A-just paths}
  A path $\pi$ starting in the initial state of $M$ is $\block$-$\conc_{\!\!A}$-just if, and only if,
  \begin{enumerate}[a)]
  \item $\pi$ thread-enables no actions $a \in \nonblock$ other than 
actions from $\textit{start}(r)$ for some $r\in\RID$, and 
  \item if, for some $r\in\RID$, an action $a \in \textit{start}(r)$ is thread-enabled by $\pi$, then $\pi$ contains infinitely many occurrences of actions $b\in\textit{start}(r)$.
  \end{enumerate}
\end{prop}
\begin{proof}
Suppose $\pi$ thread-enables an action $a \in \nonblock$, say with $t= \thrmap{a}$ and $r= \regmap{a}$, such that if $a \in \textit{start}(r)$ then $\pi$ contains only finitely many actions $b\in\textit{start}(r)$. We have to show that $\pi$ is \hyperlink{just}{not $\block$-$\conc_{\!\!A}$-just}. Let $\pi'$ be a suffix of $\pi$ in which no actions $b$ with $\thrmap{b}=t$ occur; in case $a \in \textit{start}(r)$ we moreover choose $\pi'$ such that it contains no actions $b\in\textit{start}(r)$. From here on, the proof proceeds exactly as the one of \cref{pr:T-just paths}, but reading $\nconc_{\!\!A}$ for $\nconc_T$.

For the other direction, suppose that $\pi$ is \hyperlink{just}{not $\block$-$\conc_{\!\!A}$-}just. We have to establish that\linebreak[3] 
(a) $\pi$ thread-enables an action $a \in \nonblock$, and (b) in case $a \in \textit{start}(r)$ then $\pi$ contains only finitely many actions $b\in\textit{start}(r)$. Let $\pi'$ be a suffix of $\pi$ such that an action $a \in \nonblock$ is enabled in the initial state $s$ of $\pi'$, but $\pi'$ contains no action $b$ such that $a \nconc_{\!\!A} b$, that is, $\pi'$ contains no action $b$ with $\thrmap{a}=\thrmap{b}$ or $a,b\in\textit{start}(r)$ for some $r\in\RID$.
The proof of (a) above proceeds as in the proof of \cref{pr:T-just paths}, and (b) is now a trivial corollary.
\end{proof}

Recall from \cref{def:concT,def:concS,def:concI} that $a \nconc_{\!I} b$ iff either $\thrmap{a}=\thrmap{b}$ or $a,b\in\textit{start}(r)$ for some $r\in\RID$, with $\iswrite{s} \vee \iswrite{b}$\footnote{Recall from \cref{sec:justness} that for our models, $\iswrite{a}$ is true iff $a$ is a start write action}. Hence a path $\pi$ is $\block$-$\conc_{\!I}$-unjust if, and only if, it has a suffix $\pi'$ such that an action $a \in\nonblock$ is enabled in the initial state of $\pi'$, but $\pi$ does not contain any action $b$ such that $\thrmap{a}=\thrmap{b}$ or $a,b\in\textit{start}(r)$ for some $r\in\RID$, with $\iswrite{a} \vee \iswrite{b}$.

Similarly, $a \nconc_{S} b$ iff either $\thrmap{a}=\thrmap{b}$ or $a,b\in\textit{start}(r)$ for some $r\in\RID$, with $\iswrite{b}$. Hence a path $\pi$ is $\block$-$\conc_{S}$-unjust if, and only if, it has a suffix $\pi'$ such that an action $a \in\nonblock$ is enabled in the initial state of $\pi'$, but $\pi$ does not contain any action $b$ such that $\thrmap{a}=\thrmap{b}$ or $a,b\in\textit{start}(r)$ for some $r\in\RID$, with $\iswrite{b}$.

\begin{prop}\rm\label{pr:I-just paths}
  A path $\pi$ starting in the initial state of $M$, is $\block$-$\conc_{\!I}$-just if, and only if,
  \begin{enumerate}[a)]
  \item $\pi$ thread-enables no actions $a \in \nonblock$ other than actions from $\textit{start}(r)$ for some $r\in\RID$,
  \item if an action $\startwrite[\tid,\rid]{d}$ is thread-enabled by $\pi$, then $\pi$ contains infinitely many occurrences of actions $b\in\textit{start}(r)$, and
  \item if an action $\startread[\tid,\rid]$ or $\readop[\tid,\rid]{\data}$ is thread-enabled by $\pi$, then $\pi$ contains infinitely many occurrences of actions $b$ of the form $\startwrite[\tid',\rid]{d}$ for some $\tid'\in\TID$ and $d'\in\Data$.
  \end{enumerate}
\end{prop}
\noindent
The proof of this proposition, and the next one, proceeds just like the one of \cref{pr:A-just paths}.

\begin{prop}\rm\label{pr:S-just paths}
  A path $\pi$ starting in the initial state of $M$, is $\block$-$\conc_{S}$-just if, and only if,
  \begin{enumerate}[a)]
  \item $\pi$ thread-enables no actions $a \in \nonblock$ other than actions from $\textit{start}(r)$ for some $r\in\RID$,
  \item if an action $a \in \textit{start}(r)$ is thread-enabled by $\pi$, then $\pi$ contains infinitely many occurrences of actions $b$ of the form $\startwrite[\tid',\rid]{d}$ for some $\tid'\in\TID$ and $d'\in\Data$.
  \end{enumerate}
\end{prop}

\noindent
An interesting consequence of these propositions is that whether or not a given path $\pi$ is
$\block$-$\conc_{C}$-just, for some $C \in \{T,S,I,A\}$, is completely determined by the set of
actions that are thread-enabled by $\pi$, and by the function $\textit{occ}_\pi:Act \rightarrow \mathbbm{N} \cup \{\infty\}$ that tells for each action how often it occurs in $\pi$.

\section{Complete paths in full-read and instant-read models}\label{app:complete-same}
In this appendix, we establish the correspondence between the instant-read and full-read models.
Ideally, we would be able to establish a standard equivalence between our models, such as bisimilarity or trace equivalence.
However, since the two approaches differ in which actions they use, and in the points at which the decision is made what value a read returns, such equivalences do not hold.
Instead, we prove that the verification results we obtain for instant-read models when checking the properties we cover in this paper are also valid for the full-read models. We also show the reverse; if we had done our verifications on the full-read models, those results would also be valid for the instant-read models.
Thus, with respect to the properties we check, the models are equivalent.

All our correctness properties for mutual exclusion protocols are linear-time properties: they hold for a process, modelled as an LTS, iff they hold for all complete paths starting in the initial state of that LTS. Moreover, whether they hold for a particular path depends solely on the labels of the transitions in that path, and in fact only on those transitions labelled $\crit[\tid]$ or $\noncrit[\tid]$ (with $\tid\in\TID$). This is witnessed by the modal $\mu$-calculus formulae for these properties given in \cref{app:mucalc}.
For the (sole) purpose of deciding whether mutual exclusion, deadlock freedom or starvation freedom hold for a given LTS, once it is determined  which paths are complete, all other actions may be considered internal, or hidden.

\begin{defi}\label{def:wct}\rm
  Given a path $\pi=s_0 a_1 s_1 a_2 s_2 \dots$, let $\ell(\pi)=a_1 a_2 \dots$ be the sequence of actions occurring on that path, and let $\ell^-(\pi)$ be the result of omitting from $\ell(\pi)$ all actions other than $\crit[\tid]$ or $\noncrit[\tid]$ (with $\tid\in\TID$).

  Given a completeness criterion $C$, the set of \emph{weak $C$-complete traces} $\WCT_C(P)$ of an LTS $P$ consists of those finite and infinite strings $\ell^-(\pi)$ for $\pi$ a $C$-complete path starting in the initial state of $P$.
  Two LTSs $P$ and $Q$ are \emph{weak completed trace equivalent} w.r.t.\ $C$, notation $P =^C_\WCT Q$ if, and only if, $\WCT_C(P) = \WCT_C(Q)$.
\end{defi}

It now follows that, given a completeness criterion $C$, two weak completed trace equivalent LTSs satisfy the same correctness properties for mutual exclusion protocols. Thus, we proceed by proving that for all six memory models we consider, the full-read thread-register models and instant-read thread register models are weak completed trace equivalent.
We use $x$ as shorthand notation for $\conc_x$ for $x \in \{T,S,I,A\}$.

Let $M_{\modeltype} = (\states_{\modeltype}, \actionset_{\modeltype}, \initstate_{\modeltype}, \transrel_{\modeltype}, \thrsym_{\modeltype}, \regsym_{\modeltype})$ with
$\modeltype \in \{\frRep, \irRep\}$ be the LTS of the thread-register model for a given mutual exclusion algorithm, using the appropriate
(with respect to
$\modeltype$ and the types of the registers) register model from \cref{fig:fullread-procsafe,fig:fullread-procregular,fig:fullread-procatomic,fig:instantread-procsafe,fig:instantread-procregular,fig:instantread-procatomic}. The two models $M_{\irRep}$ and $M_{\frRep}$ capture the same types of registers, with the same identifiers and domains, and the same threads, except that they use different action(s) for performing a read operation.

A state $s$ in $\states_{\modeltype}$ must have the form $(s_1,\dots,s_k)$, where each of the indices $i \in \{1,\dots,k\}$ corresponds to a thread $t\in\TID$ or a register $r\in\RID$; as in \hyperlink{thrconsist}{the proof of} \cref{lem:thr-consist-fullread} in \cref{app:thrconsist} we denote the corresponding component $s_i$ as $s_{\idx{\tid}}$ or $s_{\idx{\rid}}$; it is a state in the LTS $T^{\modeltype}_t$ or $R^{\regtype,\modeltype}_r$ for $\regtype\in\{\safRep,\regRep,\atoRep\}$.
Let $R^{\regtype,\modeltype}_{\rid} = (\states^{\regtype,\modeltype}_{\rid}, \actionset^{\regtype,\modeltype}_{\rid}, \initstate^{\regtype,\modeltype}_{\rid}, \transrel^{\regtype,\modeltype}_{\rid})$ be the LTS associated with each $\rid \in \RID$, and likewise for $T^{\modeltype}_t$.
In our arguments, we use $\thrsym$ to denote both $\thrsym_{\frRep}$ and $\thrsym_{\irRep}$; we do the same for $\regsym$.
Given a path $\pi$ in $M_{\modeltype}$, let $\proj_{\tid}(\pi)$ be the sequence of all thread-local and interface action labels $a$ in $\pi$ such that $\thrmap{a} =\tid$ for all $\tid\in\TID$. 
Similarly, let $\proj_{\rid}(\pi)$ be the sequence of all the register-local and interface action labels $a$ in $\pi$ such that $\regmap{a}=\rid$ for all $\rid\in\RID$.

We first prove that whenever a path $\pi$ exists in a full-read model, a path $\pi'$ exists in the corresponding instant-read model such that $\ell^-(\pi) = \ell^-(\pi')$, and vice versa. We then prove that $\pi'$ is $C$-complete iff $\pi$ is.

\subsection{Paths: instant-read to full-read}\label{app:path-itof}
We introduce the transformation $\itof$ on sequences of actions. Note that this $\itof$ is different from the $\itofstate$ defined on states in \cref{app:thread-transform}.
Let $\pi$ be a path from the initial state of $M_{\irRep}$, then $\itof(\ell(\pi))$ is defined as $\ell(\pi)$ but with every $\readop[\tid,\rid]{\data}$ replaced with $\startread[\tid,\rid]\finishread[\tid,\rid]{\data}$ if $\rid$ is safe or regular, or $\startread[\tid,\rid]\orderread[\tid,\rid]\finishread[\tid,\rid]{\data}$ if $\rid$ is atomic, for all $\tid,\rid$ and $\data$.

Our main result of this subsection is that whenever a path $\pi$ exists in $M_{\irRep}$, a path $\pi'$ with $\ell(\pi') = \itof(\ell(\pi))$ exists in $M_{\frRep}$. We first prove some supporting lemmas on the transformation $\itof$.

\begin{lem}\rm\label{lem:irtofr-registers}
     Let $\pi$ be a finite path from the initial state of $M_{\irRep}$, and let $\pi'$ be a finite path from the initial state of $M_{\frRep}$ such that $\ell(\pi') = \itof(\ell(\pi))$. Let $s = (s_{\irRep,1},\ldots,s_{\irRep,k})$ be the last state of $\pi$ and  let $s' = (s'_{\frRep,1},\ldots,s'_{\frRep,k})$ be the last state of $\pi'$. Then it holds for all $\rid\in\RID$ that $\writers{s_{\irRep,\idx{\rid}}} = \writers{s'_{\frRep,\idx{\rid}}}$ and $\pending{s_{\irRep,\idx{\rid}}} = \pending{s'_{\frRep,\idx{\rid}}}\cap\writers{s'_{\frRep,\idx{\rid}}}$.
\end{lem}
\begin{proof}
    From the definition of $\itof$, it follows that $\proj_{\rid}(\pi)$ and $\proj_{\rid}(\pi')$ are identical, save for those actions that are tied to read operations.
    As can be seen by the definitions of the update functions in \cref{app:fullread}, $\writers{s'_{\frRep,\idx{\rid}}}$ and $\pending{s'_{\frRep,\idx{\rid}}} \cap \writers{s'_{\frRep,\idx{\rid}}}$ are fully determined by the calls to $\uswsym$, $\uowsym$ and $\ufwsym$ and thus the occurrences of start, order and finish write actions in $\proj_{\rid}(\pi')$. The same holds for $\writers{s_{\irRep,\idx{\rid}}}$ and $\pending{s_{\irRep,\idx{\rid}}}$ and $\proj_{\rid}(\pi)$, as can be seen in \cref{app:instantread}.
    Additionally, these update functions are defined identically between the two model types with respect to $\writerssym$ and $\pendingsym$.
    Thus, since the occurrences of actions tied to write operations are the same between $\proj_{\rid}(\pi)$ and $\proj_{\rid}(\pi')$, the claim holds.
\end{proof}

\begin{prop}\rm\label{lem:ir2fr}
    Given a path $\pi$ from the initial state of $M_{\irRep}$, there exists a path $\pi'$ from the initial state of $M_{\frRep}$ such that $\ell(\pi') = \itof(\ell(\pi))$, and for all prefixes $\pi_i$ consisting of the first $i$ transitions of $\pi$, ending in a state $s=(s_1,\ldots,s_k)$, there exists a prefix $\pi'_i$ of $\pi'$, ending in a state $s'=(s'_1,\ldots,s'_k)$, such that $\ell(\pi'_i) = \itof(\ell(\pi_i))$ and $s'_{\idx{\tid}} \in \itofstate(s_{\idx{\tid}})$ for all $\tid\in\TID$. 
\end{prop}
\begin{proof}
    Let $\pi$ be a path from the initial state of $M_{\irRep}$.    
    We do induction on the first $n$ actions in $\ell(\pi)$. To this end, let $\pi_n$ be the prefix of $\pi$ containing the first $n$ actions and $n+1$ states.
    If $n = 0$, then trivially the path consisting of only the initial state of $M_{\frRep}$ has the same actions. Additionally, as witnessed by the transformation from \cref{app:thread-transform}, the initial states of all the threads are equal.
    Let $i \geq 0$ and let $s = (s_1,\ldots,s_k)$ be the final state of $\pi_i$. Assume we have a path $\pi'_i$ with final state $s' = (s'_1, \ldots, s'_k)$ from the initial state of $M_{\frRep}$ such that $\ell(\pi'_i) = \itof(\ell(\pi_i))$, for all $\tid\in\TID$ $s'_{\idx{\tid}} \in \itofstate(s_{\idx{\tid}})$ and all the appropriate prefixes as described in the proposition exist.
    Consider the $(i+1)^{\text{th}}$ action, $a$, of $\ell(\pi)$.
    Then $a$ is naturally enabled in $s$.
    We must show that $a$ is enabled in $s'$.
    We do a case distinction on what action $a$ is:
    \begin{itemize}
        \item If $a$ is a thread-local action with $\thrmap{a}=\tid$ for some $\tid\in\TID$, then it is sufficient to establish that $a$ is enabled in $s'_{\idx{\tid}}$. This follows from $s'_{\idx{\tid}}$ being in $\itofstate(s_{\idx{\tid}})$, $a$ being enabled in $s_{\idx{\tid}}$ and \cref{lem:itof-preserve}.
        
        \item If $a$ is a register-local action, then it must be $\orderwrite[\tid,\rid]$ for some $\tid\in\TID, \rid\in\RID$, since these are the only register-local actions in the instant-read models. It is sufficient to establish that $a$ is enabled in $s'_{\idx{\rid}}$. In all relevant models (\cref{fig:fullread-procregular,fig:fullread-procatomic,fig:instantread-procregular,fig:instantread-procatomic}), $\orderwrite[\tid,\rid]$ is enabled exactly in those states $s$ where $\tid \in \writers{s} \land \tid \in \pending{s}$\. By \cref{lem:irtofr-registers}, $\tid \in \writers{s_{\idx{\rid}}} \land \tid \in \pending{s_{\idx{\rid}}}$ implies $\tid \in \writers{s'_{\idx{\rid}}} \land \tid \in \pending{s'_{\idx{\rid}}}$. Therefore, $a$ being enabled in $s$ implies it is enabled in $s'$. 
        
        \item If $a$ is $\startwrite[\tid,\rid]{\data}$ or $\finishwrite[\tid,\rid]$ for some $\tid\in\TID,\rid\in\RID,\data\in\Data[\rid]$, then we must establish that it is enabled in both $s'_{\idx{\tid}}$ and $s'_{\idx{\rid}}$. The reasoning is similar to the previous two cases: for $\tid$, this again follows from $s'_{\idx{\tid}}$ being in $\itofstate(s_{\idx{\tid}})$ and \cref{lem:itof-preserve}.
        For $\rid$, we observe from \cref{fig:fullread-procsafe,fig:fullread-procregular,fig:fullread-procatomic,fig:instantread-procsafe,fig:instantread-procregular,fig:instantread-procatomic} that $\startwrite[\tid,\rid]{\data}$ is enabled in a state $s$ exactly when $\tid \notin (\readers{s} \cup \writers{s})$ for the full-read models and $\tid \notin \writers{s}$ for the instant-read models. For both model types, $\finishwrite[\tid,\rid]$ is enabled exactly in those states $s$ where $\tid\in \writers{s}$ for the safe models and $\tid\notin \writers{s} \land \tid \notin \pending{s}$ for the regular and atomic models.
        By \cref{lem:irtofr-registers}, $\tid \in \writers{s_{\idx{\rid}}}$ implies $\tid\in\writers{s'_{\idx{\rid}}}$ and vice versa. Similarly, $\tid \in \pending{s_{\idx{\rid}}}$ implies $\tid \in \pending{s'_{\idx{\rid}}} \cap \writers{s'_{\idx{\rid}}}$ and vice versa.
        Hence, for those conditions, $a$ being enabled in $s_{\idx{\rid}}$ implies it is enabled in $s'_{\idx{\rid}}$. 
        It remains to establish that  $\tid\notin\readers{s'_{\idx{\rid}}}$ when $a$ is $\startwrite[\tid,\rid]{\data}$.
        This follows from \cref{thread-prop-2} from \cref{sec:threads} and that we have already established that $a$ is enabled in $s'_{\idx{\tid}}$: a thread cannot attempt to start writing when it is already actively reading, thus it will not be actively reading when it has $\startwrite[\tid,\rid]{\data}$ enabled.
        Thus, $a$ is enabled by $s'_{\idx{\rid}}$.
        
        \item If $a$ is $\readop[\tid,\rid]{\data}$ for some $\tid\in\TID,\rid\in\RID,\data\in\Data[\rid]$ then it is sufficient to show that from $s'$ the actions $\startread[\tid,\rid]$, optionally $\orderread[\tid,\rid]$ and finally $\finishread[\tid,\rid]{\data}$ can be performed directly in-sequence. Since $\readop[\tid,\rid]{\data}$ is enabled in $s_{\idx{\tid}}$ and $s'_{\idx{\tid}} \in \itofstate(s_{\idx{\tid}})$, it must be that in $s'_{\idx{\tid}}$ $\startread[\tid,\rid]$ is enabled and immediately afterward $\finishread[\tid,\rid]{\data}$ is enabled (by \cref{lem:itof-read}). For $\rid$, we can see from \cref{fig:fullread-procsafe,fig:fullread-procregular,fig:fullread-procatomic} that the only condition for $\startread[\tid,\rid]$ being enabled is that $\tid$ is not already reading or writing to $\rid$. By \cref{thread-prop-2} of \cref{sec:threads}, this is certainly not the case since we have already argued that $\startread[\tid,\rid]$ is enabled in $s'_{\idx{\tid}}$ and a thread never enables a start read when it is already involved in an operation. Thus, $\startread[\tid,\rid]$ is possible. In the atomic model, $\orderread[\tid,\rid]$ is enabled immediately after $\startread[\tid,\rid]$, since $\usrsym$ sets the relevant values correctly (see \cref{app:fullread-status}). Similarly as a consequence of the update functions, for all three models it is straightforward that $\finishread[\tid,\rid]{\data'}$ is enabled next for some $\data' \in \Data[\rid]$ after $\startread[\tid,\rid]$ or $\orderread[\tid,\rid]$. It remains to establish that the specific value $\data$ can be read. Consider that there are no actions, besides possibly $\orderread[\tid,\rid]$, between $\startread[\tid,\rid]$ and $\finishread[\tid,\rid]{\data'}$, thus the read operation still overlaps with exactly the same write operations in both the full-read and instant-read variants. From the definitions of the register models, it follows that the value $\data$ must then be possible to read. Thus $\finishread[\tid,\rid]{\data}$ is also possible.
    \end{itemize}
    Thus, $a$ is enabled in $s$ and there therefore exists a path $\pi'_{i+1}$ such that $\ell(\pi'_{i+1}) = \itof(\ell(\pi_{i+1}))$.
    Additionally, it follows from the transformation described in \cref{app:thread-transform} and $s'_{\idx{\tid}} \in \itof(s_{\idx{\tid}})$ that we can take transition(s) from $s'$ in such a way that the resulting full-read thread states are still in the mapping from instant-read thread states. 
    Thus, by induction we conclude that the claim holds.
\end{proof}
\noindent
It is straightforward to observe that the transformation $\itof$ does not alter the presence or order of actions $\crit[\tid]$ or $\noncrit[\tid]$ for all $\tid\in\TID$. 
\begin{cor}\rm\label{cor:ir2fr}
    Given a path $\pi$ from the initial state of $M_{\irRep}$, there exists a path $\pi'$ from the initial state of $M_{\frRep}$ such that $\ell^-(\pi') = \ell^-(\pi)$.
\end{cor}

\subsection{Paths: full-read to instant-read}\label{app:path-ftoi}
Similar to the transformation $\itof$ in \cref{app:path-itof}, we here introduce the transformations on sequences of actions $\pftoi$ and $\ftoi$. 
Let $\pi$ be a path from the initial state of $M_{\frRep}$, then $\pftoi(\ell(\pi))$ is defined as $\ell(\pi)$ but with the following transformation applied to every read operation in $\ell(\pi)$: 
    We describe the transformation for one read operation. To this end, let $s$ and $f$ be indices in $\ell(\pi)=a_1 a_2 \dots$ such that $a_s=\startread[\tid,\rid]$ and $a_f=\finishread[\tid,\rid]{\data}$ (with $\tid\in\TID,\,\rid\in\RID,\,\data\in\Data[\rid]$) are part of one read operation. If $\rid$ is an atomic register, also let $a_o=\orderread[\tid,\rid]$ be the matching order read action. If $a_f$ does not exist (because a read operation is invoked but never completed in this particular path), nothing further is done for this read operation.
    If these indices do exist, the continuation of the transformation depends on the type of $\rid$:
    \begin{itemize}
        \item If $\rid$ is safe, we further consider whether there is an overlapping write with this read operation. If there is, the action $\readop[\tid,\rid]{\data}$ is either directly after the invocation of the first overlapping write or directly after $a_s$, whichever comes later.\footnote{By inserting an action between $a_s$ and $a_f$, the index of $a_f$ should be incremented by one. For simplicity, we let $a_f$ still refer to the index of the same finish read action.} If there is not an overlapping write with this read operation, $\readop[\tid,\rid]{\data}$ is inserted directly after $a_s$. 
        \item If $\rid$ is regular, we further consider whether $\data$ is the value written by any overlapping write with this read operation. If it is, we insert $\readop[\tid,\rid]{\data}$ directly after the invocation of the first overlapping write that writes $\data$ or directly after $a_s$, whichever comes later. Otherwise, $\readop[\tid,\rid]{\data}$ is inserted immediately after $a_s$.
        \item If $\rid$ is atomic, $\readop[\tid,\rid]{\data}$ is placed immediately after $a_o$.
    \end{itemize}
Finally, $\ftoi$ is defined by taking $\pftoi$ and dropping all occurrences of $\startread[\tid,\rid]$, $\orderread[\tid,\rid]$ and $\finishread[\tid,\rid]{\data}$, for all $\tid\in\TID,\,\rid\in\RID,\,\data\in\Data[\rid]$.

We need some further auxiliary definitions.
Say we have a path $\pi$ from the initial state of $M_{\frRep}$.
Consider $\ftoi(\ell(\pi))$ and specifically the prefix $\ftoi(\ell(\pi))_n$ of its first $n$ actions, for some natural number $n$.
Recall that $\ftoi(\ell(\pi))$ is constructed by dropping all start, order and finish read actions from $\pftoi(\ell(\pi))$. 
Let $\pftoi(\ell(\pi))_n$ be the longest prefix of $\pftoi(\ell(\pi))$ that contains all actions of $\ftoi(\ell(\pi))_n$ but does not contain the $(n+1)^{\text{th}}$ action of $\ftoi(\ell(\pi))$.\footnote{If there is no $(n+1)^{\text{th}}$ action because $\pi$ is finite, then this longest prefix of $\pftoi(\ell(\pi))$ is $\pftoi(\ell(\pi))$ itself.} Thus, the next action of $\pftoi(\ell(\pi))$ after $\pftoi(\ell(\pi))_n$, if not a read action, is the same as the next action of $\ftoi(\ell(\pi))$ after $\ftoi(\ell(\pi))_n$.
We can obtain a prefix of $\ell(\pi)$ by dropping all read actions from $\pftoi(\ell(\pi))_n$; let $\textit{pre}_n(\pi)$ be the prefix of $\pi$ such that $\ell(\textit{pre}_n(\pi))$ is $\pftoi(\ell(\pi))_n$ with all read actions dropped.
Note that it may be that $\ftoi(\ell(\textit{pre}_n(\pi))) \neq \ftoi(\ell(\pi))_n$: if, for example, the last action of $\ftoi(\ell(\pi))_n$ is $\readop[\tid,\rid]{\data}$, then $\pftoi(\ell)\pi))_n$ contains the matching $\startread[\tid,\rid]$ but need not contain the $\finishread[\tid,\rid]{\data}$ and thus $\ftoi(\ell(\textit{pre}_n(\pi))$ may not contain the $\readop[\tid,\rid]{\data}$ action.

Similar to \cref{app:path-itof}, we prove some additional lemmas before the main result.

\begin{lem}\rm\label{lem:frtoir-projections}
    Let $\pi$ be a path from the initial state of $M_{\frRep}$ and let $\pi'_n$ be a path from the initial state of $M_{\irRep}$ such that $\ell(\pi'_n) = \ftoi(\ell(\pi))_n$ for some natural number $n$.
    Let $\pi_n$ be $\textit{pre}_n(\pi)$.
    Let $a$ be the $(n+1)^{\text{th}}$ action of $\ftoi(\ell(\pi))$ and let $\tid=\thrmap{a}$. Then 
    $\proj_{\rid}(\pi_n)$ is equal to  $\proj_{\rid}(\pi'_n)$ with respect to start, order and finish write actions for all $\rid\in\RID$.
\end{lem}
\begin{proof}
    We obtain $\textit{pre}_n(\pi)$, so $\pi_n$, by dropping from $\pftoi(\ell(\pi))_n$ all read actions.
    In turn $\pftoi(\ell(\pi))_n$ is obtained as the longest prefix of $\pftoi(\ell(\pi))$ that contains the first $n$ actions of $\ftoi(\ell(\pi))$ but not the $(n+1)^{\text{th}}$ action.
    Since start, finish and order write actions are unaltered by the transformation $\ftoi$, the claim is trivially true.
\end{proof}

\begin{lem}\rm\label{lem:frtoir-threads}
    Let $\pi$ be a path from the initial state of $M_{\frRep}$ and let $\pi'_n$ with final state $s'=(s'_1,\ldots,s'_k)$ be a path from the initial state of $M_{\irRep}$ such that $\ell(\pi'_n) = \ftoi(\ell(\pi))_n$ for some natural number $n$.
    Let $\pi_n$ be $\textit{pre}_n(\pi)$ and let $s=(s_1,\ldots,s_k)$ be the final state of $\pi_n$. Assume that $s'_{\idx{\tid}} \in\ftoistate(s_{\idx{\tid}})$ for all $\tid\in\TID$.
    Let $a$ be the $(n+1)^{\text{th}}$ action of $\ftoi(\ell(\pi))$ and let $\tid=\thrmap{a}$.
    If a is a thread-local or register interface action, $a$ is enabled in $s'_{\idx{\tid}}$.
\end{lem}
\begin{proof}
    Since $a$ is the $(n+1)^{\text{th}}$ action of $\ftoi(\ell(\pi))$, $a$ is the next action in $\pftoi(\ell(\pi))$ after $\pftoi(\ell(\pi))_n$.
    We do a case distinction on whether $a = \readop[\tid,\rid]{\data}$ for some $\rid\in\RID,\data\in\Data[\rid]$.
    \begin{itemize}
        \item If $a$ is not a read action, then $a$ is also the next action in $\ell(\pi)$ after $\pi_n$ and is thus enabled in $s$. Since $a$ is a thread-local or register interface action, $a$ is enabled in $s_{\idx{\tid}}$. Then it follows from \cref{lem:ftoi-preserve} that $s'_{\idx{\tid}}$ enables $a$.
        \item If $a$ is a read action $\readop[\tid,\rid]{\data}$, then since $a$ is the next action of $\pftoi(\ell(\pi))$ after $\pftoi(\ell(\pi))_n$, $\pftoi(\ell(\pi))_n$ contains the matching $\startread[\tid,\rid]$ but not the matching $\finishread[\tid,\rid]{\data}$, since $\readop[\tid,\rid]{\data}$ is always placed between them. Thus, $\proj_{\tid}(\pi_n)$ ends on $\startread[\tid,\rid]$.
        This means that $s'_{\idx{\tid}}$, by being in $\ftoistate(s_{\idx{\tid}})$ could be either the state before $\readop[\tid,\rid]{\data}$, or after.But since this read action is the $(n+1)^{\text{th}}$ action on $\ftoi(\ell(\pi))$, it cannot have occurred yet in $\proj_{\tid}(\pi'_n)$, so $s'_{\idx{\tid}}$ must be the state where read actions are enabled, including $a$.\qedhere
    \end{itemize}
\end{proof}

\begin{lem}\rm\label{lem:frtoir-registers}
    Let $\pi$ be a path from the initial state of $M_{\frRep}$ and let $\pi'_n$ with final state $s'=(s'_1,\ldots,s'_k)$ be a path from the initial state of $M_{\irRep}$ such that $\ell(\pi'_n) = \ftoi(\ell(\pi))_n$ for some natural number $n$.
    Let $\pi_n$ be $\textit{pre}_n(\pi)$ and $s=(s_1,\ldots,s_k)$ be the final state of $\pi_n$. 
    Then it holds for all $\rid\in\RID$ that $\writers{s_{\idx{\rid}}} = \writers{s'_{\idx{\rid}}}$ and $\pending{s_{\idx{\rid}}}\cap\writers{s_{\idx{\rid}}} = \pending{s'_{\idx{\rid}}}$.
\end{lem}
\begin{proof}
    Consider an arbitrary $\rid\in\RID$.
    From \cref{lem:frtoir-projections}, we know that $\proj_{\rid}(\pi_n)$ is equal to $\proj_{\rid}(\pi'_n)$ with respect to those actions that are tied to write operations. The argument then proceeds the same as the one for \cref{lem:irtofr-registers}; since the same calls to $\uswsym$, $\uowsym$ and $\ufwsym$ are done, the claim indeed holds.
\end{proof}

\begin{lem}\rm\label{lem:fr2ir-support-end}
    Let $\pi$ be a path from the initial state of $M_{\frRep}$ and let $\pi'_i$ be a path from the initial state of $M_{\irRep}$ such that $\ell(\pi'_i) = \ftoi(\ell(\pi))_i$, for some natural number $i$.
    Let $\pi_i = \textit{pre}_i(\pi)$ and $\pi_{i+1} = \textit{pre}_{i+1}(\pi)$.
    Let $s = (s_1,\ldots,s_k)$ be the final state of $\pi_i$, $s^+ = (s^+_1,\ldots,s^+_k)$ the final state of $\pi_{i+1}$ and $s' = (s'_1,\ldots,s'_k)$ the final state of $\pi'_i$.
    Assume that $s'_{\idx{\tid'}} \in \ftoistate(s_{\idx{\tid'}})$ for all $\tid'\in\TID$.
    Assume that it is possible to pick a path $\pi'_{i+1}$ from the initial state of $M_{\irRep}$ such that $\ell(\pi'_{i+1}) = \ftoi(\ell(\pi))_{i+1}$ and $\pi'_i$ is a prefix of $\pi'_{i+1}$.
    We prove that it is possible to pick such a path $\pi'_{i+1}$ with final state $s^{+\prime} = (s^{+\prime}_1,\ldots,s^{+\prime}_k)$ such that 
     $s^{+\prime}_{\idx{\tid'}} \in \ftoistate(s^+_{\idx{\tid'}})$ for all $\tid'\in\TID$.
\end{lem}
\begin{proof}
    Let $a$ be the $(i+1)^{\text{th}}$ action of $\ftoi(\ell(\pi))$ with $\thrmap{a} = \tid$.
    Towards a contradiction, assume there exists a $\tid'\in\tid$ such that we cannot pick $\pi'_{i+1}$ such that $s^{+\prime}_{\idx{\tid}} \in \ftoistate(s^+_{\idx{\tid}})$.
    We do a case distinction on whether $\tid'=\tid$. We use that any $\pi'_{i+1}$ we pick satisfies $\ell(\pi'_{i+1}) = \ftoi(\ell(\pi))_{i+1}$ and has $\pi'_i$ as a prefix.
    \begin{itemize}
        \item If $\tid'\neq\tid$, then note that $s^{+\prime}_{\idx{\tid'}} = s'_{\idx{\tid'}}$, since $a$ is not a transition by $\tid'$ and it is the only difference between $\pi'_i$ and $\pi'_{i+1}$. Since $s'_{\idx{\tid'}} \in \ftoistate(s_{\idx{\tid'}})$ and $s^{+\prime}_{\idx{\tid'}} \notin \ftoistate(s^+_{\idx{\tid'}}$, there must be at least one transition between $s_{\idx{\tid'}}$ and $s^+{(\idx{\tid'}}$ such that $s'_{\idx{\tid'}} \notin \ftoistate(s^+_{\idx{\tid'}})$.
        Transitions labelled with actions by $\tid'$ between $s$ and $s^+$ that are thread-local, start write or finish write actions would also appear in $\ftoi(\ell(\pi))_{i+1}$ after $\ftoi(\ell(\pi))_i$, but the only transition there is labelled with $a$ and $\thrmap{a} \neq \tid'$.  Similarly, if both a start read and finish read action by $\tid'$ had occurred between $s$ and $s^+$, then a read action by $\tid'$ would appear in $\ftoi(\ell(\pi))_{i+1}$ that did not appear in $\ftoi(\ell(\pi))_i$. 
        Thus, the only actions that can occur between $s_{\idx{\tid'}}$ and $s^+_{\idx{\tid'}}$ are start read and finish read actions.
        Consider a transition $(s_b,\startread[\tid',\rid],s_a)$ in the LTS of $\tid'$: if $s'_{\idx{\tid'}} \in \ftoistate(s_b)$ then also $s'_{\idx{\tid'}} \in \ftoistate(s_a)$, since from the definition of $\ftoistate$ in \cref{app:thread-transform}, it follows that $\ftoistate(s_b) \subseteq \ftoistate(s_a)$.
        Thus, a start read transition cannot cause $s'_{\idx{\tid'}} \notin \ftoistate(s^+_{\idx{\tid'}})$.
        Next, consider a transition $(s_b, \finishread[\tid',\rid]{\data},s_a)$ in the LTS of $\tid'$. Then a state in $\ftoistate(s_b)$ could be both a state where $\readop[\tid',\rid]{\data}$ is enabled or the target of that transition. However, since $\ftoi$ inserts the read actions between matching start read and finish read actions, if $\finishread[\tid',\rid]{\data}$ occurs after $s$, then the matching $\readop[\tid',\rid]{\data}$ has to have happened by $s'$.
        Thus, we only need to consider the case that $\ftoistate(s_b)$ maps to the state after the read action, and thus $\ftoistate(s_b) = \ftoistate(s_a)$.
        Thus, a finish read transition also cannot cause $s'_{\idx{\tid'}} \notin \ftoistate(s^+_{\idx{\tid'}})$.
        We conclude that $s^{+\prime}_{\idx{\tid}} = s'_{\idx{\tid'}} \in \ftoistate(s^+_{\idx{\tid'}}$
        
        \item If $\tid'=\tid$, then $s^{+\prime}_{\idx{\tid}}$ is $s'_{\idx{\tid}}$ after a single transition, namely $a$.
        There could potentially be other transitions by $\tid$ between $s$ and $s^+$, specifically there could be a start read action, but by the same reasoning as the $\tid\neq\tid'$ case, this would not further affect $s^{+\prime}_{\idx{\tid}}$.
        If $a$ was a thread-local, start write or finish write action, then that same transition also occurred on the full-write side, and thus by the definition of the mapping we can pick transitions such that $s^{+\prime}_{\idx{\tid}} \in \ftoistate(s^+_{\idx{\tid}})$.
        If $a$ is $\readop[\tid,\rid]{\data}$, then $\pi_{i+1}$ must contain the matching $\startread[\tid,\rid]$ but may or may not contain the matching $\finishread[\tid,\rid]{\data}$. However, since $s^{+\prime}_{\idx{\tid}}$ is contained both in the full-read state before and after the $\finishread[\tid,\rid]{\data}$, it is still in $\ftoistate(s^+_{\idx{\tid}})$.
        \qedhere
    \end{itemize}
\end{proof}

\begin{prop}\rm\label{lem:fr2ir}
    Given a path $\pi$ from the initial state of $M_{\frRep}$, there exists a path $\pi'$ from the initial state of $M_{\irRep}$ such that $\ell(\pi') = \ftoi(\ell(\pi))$, and for all prefixes $\textit{pre}_i(\pi)$ of $\pi$, ending in state $s=(s_1,\ldots,s_k)$, there exists a prefix $\pi'_i$ of $\pi'$, ending in $s'=(s'_1,\ldots,s'_k)$, such that $\ell(\pi_i') = \ftoi(\ell(\pi))_i$ and for all $\tid\in\TID$, $s'_{\idx{\tid}} \in \ftoistate(s_{\idx{\tid}})$.
\end{prop}
\begin{proof}
    Let $\pi$ be a path from the initial state of $M_{\frRep}$.
    We do induction on the first $n$ actions of $\ftoi(\ell(\pi))$.
    For $n=0$, the path consisting of just the initial state of $M_{\irRep}$ exists and exactly contains the first $0$ actions of $\ftoi(\ell(\pi))$.
    Additionally, $\textit{pre}_0(\pi)$ is just its initial state, and from the transformations described in \cref{app:thread-transform} it follows that the initial states of $M_{\irRep}$ and $M_{\frRep}$ are the same. Clearly, $\initstate_{\irRep} \in \ftoistate(\initstate_{\frRep})$.
    
    Consider the first $i$ actions of $\ftoi(\ell(\pi))$ for $i\geq 0$.
    Let $\pi_i = \textit{pre}_i(\pi)$ and let $s = (s_{1}, \ldots, s_{k})$ be the last state of $\pi_i$.
    By induction there exists a path $\pi'_i$ from the initial state of $M_{\irRep}$ ending in a state $s'=(s'_1,\ldots,s'_k)$ such that $\ell(\pi'_i) = \ftoi(\ell(\pi))_i$ and for all $\tid\in\TID$, $s'_{\idx{\tid}} \in \ftoistate(s_{\idx{\tid}})$. Also assume the requirement on all prefixes $\textit{pre}_m(\pi)$ with $m < i$ holds. 
    We must show that the $(i+1)^{\text{th}}$ action of $\ftoi(\ell(\pi))$, $a$, is enabled in $s'$.
    
    We consider what action $a$ is:
    \begin{itemize}
        \item If $a$ is a thread-local action of a thread $\tid$, we need to establish that $a$ is enabled in $s'_{\idx{\tid}}$. From \cref{lem:frtoir-threads}, it follows that this is the case.
        
        \item If $a$ is a register-local action, then it must be $\orderwrite[\tid,\rid]$ for some $\tid,\rid$ since this is the only register-local action that appears in the instant-read models. It then suffices to prove that it is enabled in $s'_{\idx{\rid}}$. Since $a$ is not a read action and it is the next action enabled in $\pftoi(\ell(\pi))$ after $\pftoi(\ell(\pi))_i$, it follows from the construction of $\textit{pre}_i(\pi)$ that $a$ is enabled in $s$ and thus in $s_{\idx{\rid}}$. Hence, from \cref{fig:fullread-procregular,fig:fullread-procatomic} it follows that $\tid\in\writers{s_{\idx{\rid}}}$ and $\tid\in(\pending{s_{\idx{\rid}}}\cap\writers{s_{\idx{\rid}}})$. By \cref{lem:frtoir-registers}, $\tid \in \writers{s'_{\idx{\rid}}}$ and $\tid\in \pending{s'_{\idx{\rid}}}$ and thus it follows from \cref{fig:instantread-procregular,fig:instantread-procatomic} that $a$ is enabled in $s'_{\idx{\rid}}$.

        \item The cases that $a$ is $\startwrite[\tid,\rid]{\data}$ or $\finishwrite[\tid,\rid]$ for some $\tid\rid$ proceed similarly to the previous two cases: that $a$ is enabled in $s'_{\idx{\tid}}$ follows directly from \cref{lem:frtoir-threads}; that $a$ is enabled in $s'_{\idx{\rid}}$ follows from $a$ being enabled in $s_{\idx{\rid}}$ (which is argued in the same way as for the order write case above), \cref{lem:frtoir-registers} and \cref{fig:fullread-procsafe,fig:fullread-procregular,fig:fullread-procatomic,fig:instantread-procsafe,fig:instantread-procregular,fig:instantread-procatomic}.
        
        \item If $a$ is $\readop[\tid,\rid]{\data}$ for some $\tid\in\TID,\rid\in\RID,\data\in\Data[\rid]$, 
        then we need to prove that $a$ is enabled in $s'_{\idx{\tid}}$ and $s'_{\idx{\rid}}$.
        The former follows directly from \cref{lem:frtoir-threads}.
        
        For $s'_{\idx{\rid}}$, we first establish that it enables $\readop[\tid,\rid]{\data'}$ for some $\data'\in\Data[\rid]$.
        By \cref{fig:instantread-procsafe,fig:instantread-procregular,fig:instantread-procatomic}, it is sufficient to establish that $\tid\notin\writers{s'_{,\idx{\rid}}}$.
        
        We have established that $s'_{\idx{\rid}}$ enables $\readop[\tid,\rid]{\data}$. Thus, by \cref{lem:ftoi-read}, $s_{\idx{\tid}}$ enables $\startread[\tid,\rid]$ or $\finishread[\tid,\rid]{\data}$. Because $\pi_i = \textit{pre}_i(\pi)$ and this is the longest prefix that does not contain too many actions, it must contain the $\startread[\tid,\rid]$, meaning that in $s_{\idx{\tid}}$, $\tid$ is actively reading which by \cref{thread-prop-2} of \cref{sec:threads} means that it cannot be actively writing.         
        Thus, $\tid\notin\writers{s_{\idx{\rid}}}$ and by \cref{lem:frtoir-registers}, $\tid\notin\writers{s'_{\idx{\rid}}}$.
        It remains to show that the value $\data$ can be read.
        We do a case distinction on the type of $\rid$:
        \begin{itemize}
            \item If $\rid$ is a safe register and $\writers{s'_{\idx{\rid}}} \neq \emptyset$, then $\rid$ enables $\readop[\tid,\rid]{\data'}$ for all $\data'\in\Data[\rid]$, including $\data$. Hence, that $\readop[\tid,\rid]{\data}$ is enabled by $s'_{\idx{\rid}}$ is trivially true. We proceed under the assumption that $\writers{s'_{\idx{\rid}}} = \emptyset$. Then by our construction of $\pftoi(\ell(\pi))$, and thus $\ftoi(\ell(\pi))$, we know there is no overlapping write with this read operation at all in $\pi$. And since $\finishread[\tid,\rid]{\data}$ appears (for this specific read operation) in $\pi$, it must be that $\data$ is the stored value of the register in $\pi$ when $\finishread[\tid,\rid]{\data}$ occurs. And since there is no overlapping write, it must also be the case that $\data$ is the stored value of the register when $\startread[\tid,\rid]$ occurs. When constructing $\pftoi(\ell(\pi))$, $\readop[\tid,\rid]{\data}$ was inserted right after $\startread[\tid,\rid]$. It follows that $\trueval{s'_{\idx{\rid}}} = \data$ and thus $\readop[\tid,\rid]{\data}$ is indeed enabled.
            \item If $\rid$ is a regular register, then it is sufficient to establish that $\data \in \{\trueval{s'_{\idx{\rid}}}\} \cup \{\data' \mid \exists_{\tid'\in\writers{s'_{\idx{\rid}}}}.\vals{s'_{\idx{\rid}},\tid'} = \data'\}$. Recall that in the construction of $\pftoi(\ell(\pi))$, if there is a write that writes $\data$ and that overlaps this read operation, then $\readop[\tid,\rid]{\data}$ was placed so that when it occurs, this write is active. Hence, in that case, $\data \in \{\data' \mid \exists_{\tid'\in\writers{s'_{\idx{\rid}}}}.\vals{s'_{\idx{\rid}},\tid'} = \data'\}$. If there is no such write, then $\readop[\tid,\rid]{\data}$ was inserted immediately after $\startread[\tid,\rid]$. In the full-read model, a read of a regular register can only return a value that is not written by an overlapping write if this value was the stored value of the register when the read began. Thus, $\trueval{s'_{\idx{\rid}}} = \data$. In both cases, we have shown that $\data$ is in the set of values that can be returned, and thus $\readop[\tid,\rid]{\data}$ is enabled.
            \item If $\rid$ is an atomic register, then $\readop[\tid,\rid]{\data}$ was placed directly after $\orderread[\tid,\rid]$ in $\pftoi(\ell(\pi))$. In the full-read model, when $\finishread[\tid,\rid]{\data}$ occurs in $\pi$ it must be because that was the stored value of the register when $\orderread[\tid,\rid]$ occurred. Thus, $\trueval{s'_{\idx{\rid}}} = \data$ and $\readop[\tid,\rid]{\data}$ is enabled.
        \end{itemize}
    \end{itemize}
    Thus, $a$ is enabled in $s'$ and therefore there exists a path $\pi'_{i+1}$ such that $\ell(\pi'_{i+1}) = \ftoi(\ell(\pi))_{i+1}$.
    Additionally, it follows from \cref{lem:fr2ir-support-end} that we can pick transitions so that the thread states still match.
    Thus, the claim holds.
\end{proof}
\noindent
It is straightforward to observe that the transformation $\ftoi$ does not alter the presence or order of actions $\crit[\tid]$ or $\noncrit[\tid]$ for all $\tid\in\TID$. 
\begin{cor}\rm\label{cor:fr2ir}
    Given a path $\pi$ from the initial state of $M_{\frRep}$, there exists a path $\pi'$ from the initial state of $M_{\irRep}$ such that $\ell^-(\pi') = \ell^-(\pi)$.
\end{cor}

\subsection{Completeness}\label{app:path-complete}
We have established that the concurrency relations are valid for both full-read and instant-read thread-register models (\cref{lem:concT-val,lem:ai-val}), and the complete traces in these models with respect to a concurrency relation $C \in \{T,S,I,A\}$ have been characterised in \cref{pr:T-just paths,pr:S-just paths,pr:I-just paths,pr:A-just paths}. We first prove several lemmas about how the concept of thread-enabledness is preserved by the operations $\itof$ and $\ftoi$.

\begin{lem}\rm\label{lem:never-fr}
    Let $\pi$ be a path from the initial state of $M_{\irRep}$. Consider a path $\pi'$ from the initial state of $M_{\frRep}$ such that $\ell(\pi') = \itof(\ell(\pi))$. Path $\pi'$ thread-enables no actions $a$ of the form $\finishread[\tid,\rid]{\data}$. 
\end{lem}
\begin{proof}
    The transformation $\itof$ replaces every occurrence of an action of the shape $\readop[\tid,\rid]{\data}$ with $\startread[\tid,\rid]\finishread[\tid,\rid]{\data}$ if $\rid$ is safe or regular, or $\startread[\tid,\rid]\orderread[\tid,\rid]\finishread[\tid,\rid]{\data}$ if $\rid$ is atomic.
    For an action of the shape $\finishread[\tid,\rid]{\data}$ to be thread-enabled in $\pi'$, it must be the case that there is an occurrence of $\startread[\tid,\rid]$ in $\pi'$ that is not followed by $\finishread[\tid,\rid]{\data}$ for some $\data$, since by \cref{thread-prop-4} of \cref{sec:threads}, a thread can only await the response of an operation after invoking it. But the transformation $\itof$ cannot have created an occurrence of $\startread[\tid,\rid]$ that was not subsequently followed by $\finishread[\tid,\rid]{\data}$ for some $\data$. Thus, no such action can be thread-enabled in $\pi'$.
\end{proof}

\begin{lem}\rm\label{lem:thread-enabled-irtofr}
    Let $\pi$ be a path from the initial state of $M_{\irRep}$. Consider a path $\pi'$ from the initial state of $M_{\frRep}$ as obtained from \cref{lem:ir2fr}. If $\pi'$ thread-enables an action $a \in \nonblock$ with $a \notin \{\startread[\tid,\rid],\finishread[\tid,\rid]{\data}\mid\tid\in\TID,\rid\in\RID,\data\in\Data[\rid]\}$ then $\pi$ thread-enables $a$. Additionally, if $\pi'$ thread-enables an action $a = \startread[\tid,\rid]$ for some $\tid\in\TID,\rid\in\RID$ then $\pi$ thread-enables an action $b \in \{\readop[\tid,\rid]{\data}\mid\data\in\Data[\rid]\}$.
\end{lem}
\begin{proof}
    Let $\pi$ be a path from the initial state of $M_{\irRep}$ and let $\pi'$ be a path from the initial state of $M_{\frRep}$ such that $\ell(\pi') = \itof(\ell(\pi))$ and for all prefixes of $\pi$, there is a matching prefix of $\pi'$ where all the threads are in states related by the mapping $\itofstate$.
    Assume that $\pi'$ thread-enables an action $a \in \nonblock$ with $\thrmap{a} = \tid$. Then $\pi'$ has a suffix $\pi_s'$ from some state $s'$ such that $\pi_s'$ contains no actions $b$ with $\thrmap{b} = \tid$ and $s'_{\idx{\tid}}$ enables $a$. 

    Consider that path $\pi$ has the same actions as path $\pi'$, except that to construct $\pi'$ we replaced every occurrence of $\readop[\tid',\rid]{\data}$ with $\startread[\tid',\rid]\finishread[\tid',\rid]{\data}$ if $\rid$ is a safe or regular register and $\startread[\tid',\rid]\orderread[\tid',\rid]\finishread[\tid',\rid]{\data}$ if $\rid$ is an atomic register. Since actions by a thread $\tid'$ are replaced by actions also by thread $\tid'$ for all $\tid'\in\TID$, it is the case that if $\pi'$ contains a suffix where no actions $b$ with $\thrmap{b} = \tid$ are enabled, then so does $\pi$. Let $\pi_s$ be such a suffix of $\pi$ and let $s$ be the first state of $\pi_s$.

    Let $\pi_p$ be the prefix of $\pi$ that ends in $s$.
    Then since $\pi'$ was obtained via \cref{lem:ir2fr}, there exists a prefix $\pi'_p$ of $\pi'$ ending in a state $s''$ such that $\ell(\pi'_p) = \itof(\ell(\pi_p))$ and $s''_{\idx{\tid}} \in \itofstate(s_{\idx{\tid}})$ for all $\tid\in\TID$.
    Since $\pi_p$ contains all actions $b$ with $\thrmap{b} = \tid$ that are in all of $\pi$, $\pi'_p$ must contain all actions $b$ with $\thrmap{b} = \tid$ (this follows from the definition of transformation $\itof$).
    Thus, $s''_{\idx{\tid}} = s'_{\idx{\tid}}$. And thus $s'_{\idx{\tid}} \in \itofstate(s_{\idx{\tid}})$.
    Therefore, if $a \notin \{\startread[\tid,\rid],\finishread[\tid,\rid]{\data}\mid\tid\in\TID,\,\rid\in\RID,\,\data\in\Data[\rid]\}$ then $a$ is also enabled in $s_{\idx{\tid}}$ (by \cref{lem:itof-preserve}). If $a \in \{\startread[\tid,\rid]{\data}\mid\tid\in\TID,\,\rid\in\RID\}$ then by \cref{lem:itof-read} a read action is enabled in $s_{\idx{\tid}}$.
\end{proof}

\begin{lem}\rm\label{lem:thread-enabled-frtoir}
    Let $\pi$ be a path from the initial state of $M_{\frRep}$. Consider a path $\pi'$ from the initial state of $M_{\irRep}$ such that $\ell(\pi') = \ftoi(\ell(\pi))$, as obtained by \cref{lem:fr2ir}. If $\pi'$ thread-enables an action $a \in \nonblock$ with $a \notin \{\readop[\tid,\rid]{\data}\mid\tid\in\TID,\rid\in\RID,\data\in\Data[\rid]\}$ then $\pi$ thread-enables $a$ or thread-enables an action of the shape $\finishread[\tid,\rid]{\data}$. Additionally, if $\pi'$ thread-enables an action $a = \readop[\tid,\rid]{\data}$ for some $\tid\in\TID,\rid\in\RID,\data\in\Data[\rid]$ then $\pi$ thread-enables an action of the form $\startread[\tid,\rid]$ or $\finishread[\tid,\rid]{\data}$.
\end{lem}
\begin{proof}
    Let $\pi$ be a path from the initial state of $M_{\frRep}$ and let $\pi'$ be a path from the initial state of $M_{\irRep}$ such that $\ell(\pi') = \ftoi(\ell(\pi))$, as obtained by \cref{lem:fr2ir}. 
    Assume that $\pi'$ thread-enables an action $a \in \nonblock$ with $\thrmap{a} = \tid$. Then $\pi'$ has a suffix $\pi_s'$ from state $s'$ such that $\pi_s'$ contains no actions $b$ with $\thrmap{b} = \tid$ and $s'_{\idx{\tid}}$ enables $a$.

    Consider that path $\pi$ consists of the same actions as path $\pi'$, except that to construct $\pi'$ we replaced every occurrence of $\startread[\tid',\rid]$, $\finishread[\tid',\rid]{\data}$ and (if $\rid$ is atomic) $\orderread[\tid',\rid]$ with $\readop[\tid',\rid]{\data}$ somewhere between the start and finish. Since actions by a thread $\tid'$ are replaced by actions also by thread $\tid'$, it is the case that if $\pi'$ contains a suffix where no actions $b$ with $\thrmap{b} = \tid$ are enabled, then so does $\pi$. Let $\pi_s$ be such a suffix of $\pi$ and let $s$ be the first state of $\pi_s$. 

    Let $\pi_p$ be the prefix of $\pi$ that ends in $s$. Consider the smallest number $i$ such that $\pi_p$ is a prefix of $\textit{pre}_{i}(\pi)$. Let $s''$ be the final state of $\textit{pre}_i(\pi)$.
    Since there are no occurrences of actions $b$ with $\thrmap{b} = \tid$ in $\pi_s$, $s_{\idx{\tid}} = s''_{\idx{\tid}}$.
    From \cref{lem:fr2ir}, there exists a prefix $\pi_i'$ of $\pi'$ ending in $s'''$ such that $\ell(\pi_i') = \ftoi(\ell(\pi))_i$ and $s'''_{\idx{\tid'}} \in \ftoistate(s''_{\idx{\tid'}})$ for all $\tid'\in\TID$. Thus $s'''_{\idx{\tid}} \in \ftoistate(s_{\idx{\tid}})$.
    Since after $\textit{pre}_i(\pi)$, $\pi$ does not contain any actions $b$ with $\thrmap{b} = \tid$, $\pi'$ does not contain any actions $b$ with $\thrmap{b} = \tid$ after $\pi'_i$: the only actions in $\pi'$ that are not in $\pi$ are read actions, and no read by $\tid$ could come after $\pi'_i$ unless at least a finish read by $\tid$ came after $\textit{pre}_i(\pi)$.
    Thus, $\pi_s'$ and the suffix of $\pi'$ starting in $s'''$ are both suffixes of $\pi'$ such that no actions $b$ with $\thrmap{b} = \tid$ occur.
    Thus, $s'_{\idx{\tid}} = s'''_{\idx{\tid'}}$. Therefore, we have that $s'_{\idx{\tid}} \in \ftoistate(s_{\idx{\tid}})$. 

    If $a$ is an action $\readop[\tid,\rid]{\data}$ for some $\rid\in\RID,\data\in\Data[\rid]$ then from $s'_{\idx{\tid}} \in \ftoistate(s_{\idx{\tid}})$, it follows that $s_{\idx{\tid}}$ either enables $\startread[\tid,\rid]$ or $\finishread[\tid,\rid]{\data}$ (by \cref{lem:ftoi-read}).
    If $a \notin \{\readop[\tid,\rid]{\data} \mid \rid\in\RID,\data\in\Data[\rid]\}$, then by \cref{lem:ftoi-preserve}, $s_{\idx{\tid}}$ enables $a$ or some action $\finishread[\tid,\rid]{\data}$.
\end{proof}

Using these lemmas, we now prove weak completed trace equivalence in both directions.

\begin{lem}\rm\label{lem:conc-irtofr}
    Let $C \in \{T,S,I,A\}$. If $\pi$ is a $C$-complete path from the initial state of $M_{\irRep}$, then there exists a $C$-complete path $\pi'$ from the initial state of $M_{\frRep}$ such that $\ell^-(\pi) = \ell^-(\pi')$.
\end{lem}
\begin{proof}
    Let $\pi$ be a $C$-complete path from the initial state of $M_{\irRep}$. From \cref{lem:ir2fr}, we know that there exists a path $\pi'$ from the initial state of $M_{\frRep}$ such that $\ell(\pi') = \itof(\ell(\pi))$ and for every prefix of $\pi$, there is a prefix of $\pi'$ where all the threads are in matching states according to the mapping $\itofstate$. Since $\ell(\pi') = \itof(\ell(\pi))$, also $\ell^-(\pi')=\ell^-(\pi))$.
 It remains to establish that $\pi'$ is $C$-complete in $M_{\frRep}$.
    Towards a contradiction, assume that $\pi'$ is not $C$-complete.
    We do a case distinction on what $C$ is:
    \begin{itemize}
        \item If $C =T$, then by \cref{pr:T-just paths}, it follows that $\pi'$ must thread-enable some action $a \in \nonblock$. By \cref{lem:thread-enabled-irtofr}, it follows that $\pi$ also thread-enables some action $b\in\nonblock$. Thus $\pi$ is not $T$-complete by \cref{pr:T-just paths}. 

        \item If $C =S$, then by \cref{pr:S-just paths}, it follows that $\pi'$ thread-enables some action $a \in \nonblock$ such that either (a) $a \notin \textit{start}(\rid)$ for all $\rid\in\RID$ or (b) $a \in \textit{start}(\rid)$ for some $\rid\in\RID$ but $\pi'$ contains only finitely many occurrences of actions of the form $\startwrite[\tid,\rid]{\data}$. In case (a), $a$ is a thread-local, finish read or finish write action. By \cref{lem:never-fr}, we know it must be a thread-local or finish write action and thus by \cref{lem:thread-enabled-irtofr}, $\pi$ also thread-enables $a$ and thus $\pi$ is not $S$-complete.
        In case (b), $a$ is an action in $\textit{start}(\rid)$ for some $\rid\in\RID$ that is thread-enabled by $\pi'$ with only finitely many occurrences of actions of the form $\startwrite[\tid,\rid]{\data}$ occurring in $\pi'$. Since $a$ is in $\textit{start}(\rid)$, it must be of the form $\startread[\tidtwo,\rid]$ or $\startwrite[\tidtwo,\rid]{\data'}$. In both cases, by \cref{lem:thread-enabled-irtofr} $\pi$ also thread-enables an action in $\textit{start}(\rid)$, namely $\startwrite[\tidtwo,\rid]{\data'}$ or $\readop[\tid,\rid]{\data''}$ for some $\data''$. Since $\pi$ also contains only finitely many occurrences of actions of the shape $\startwrite[\tid,\rid]{\data}$, $\pi$ is not $S$-complete.

        \item If $C = I$, then by \cref{pr:I-just paths}, it follows that $\pi'$ thread-enables some action $a\in\nonblock$ such that either (a) $a \notin \textit{start}(\rid)$ for all $\rid\in\RID$, (b) $a = \startwrite[\tid,\rid]{\data}$ and $\pi'$ contains finitely many occurrences of actions $b \in \textit{start}(\rid)$, or (c) $a = \startread[\tid,\rid]$ and $\pi'$ contains finitely many occurrences of actions of the form $\startwrite[\tidtwo,\rid]{\data}$.
        We cover each case:
        \begin{itemize}
        \item Case (a) proceeds just as case (a) when $C=S$. It follows that $\pi$ is not $I$-complete.
        \item If $\pi'$ thread-enables an action $a=\startwrite[\tid,\rid]{\data}$, then by \cref{lem:thread-enabled-irtofr}, $\pi$ also thread-enabled $a$. In $\pi'$, there are only finitely many occurrences of actions in $\textit{start}(\rid)$. Since the start write actions are also in $\pi$ and the start read actions must come from read actions in $\pi$, which are also in $\textit{start}(\rid)$, the same holds for $\pi$ and so $\pi$ is not $I$-complete.
        \item If $\pi'$ thread-enables an action $a= \startread[\tid,\rid]$, then by \cref{lem:thread-enabled-irtofr}, $\pi$ enables an action $\readop[\tid,\rid]{\data}$ for some $\data\in\Data[\rid]$. By \cref{pr:I-just paths}, this, combined with there being only finitely many occurrences of write actions to $\rid$ in both $\pi$ and $\pi'$, means that $\pi$ is not $I$-complete.
        \end{itemize}

        \item If $C = A$, then by \cref{pr:A-just paths}, it follows that $\pi'$ thread-enables some action $a\in\nonblock$ such that either (a) $a \notin \textit{start}(\rid)$ for all $\rid\in\RID$, or (b) $a \in \textit{start}(\rid)$ for some $\rid\in\RID$ and $\pi'$ contains finitely many occurrences of actions $b \in \textit{start}(\rid)$.
        We cover both cases:
        \begin{itemize}
        \item Case (a) proceeds just as case (a) for $C = S$. It follows that $\pi$ not $A$-complete.
        \item If $\pi'$ thread-enables an action $a\in \textit{start}(\rid)$ for some $\rid\in\RID$, then this is either $\startread[\tid,\rid]$ or $\startwrite[\tid,\rid]{\data}$ for some $\tid\in\TID,\,\data\in\Data[\rid]$. In both cases, $\pi$ also thread-enables an action in $\textit{start}(\rid)$, by \cref{lem:thread-enabled-irtofr} either $\startwrite[\tid,\rid]{\data}$ or $\readop[\tid,\rid]{\data'}$ for some $\data'\in\Data[\rid]$. Additionally, since there are only finitely many occurrences of actions in $\textit{start}(\rid)$ in $\pi'$, the same holds for $\pi$: the start write actions are the same, and the read actions are replaced with start read actions. Thus, $\pi$ is not $A$-complete.
        \end{itemize}
    \end{itemize}
    We have shown that in all cases, $\pi$ is not $C$-complete, which contradicts our choice of $\pi$. Thus, $\pi'$ must be $C$-complete.
\end{proof}

\begin{lem}\rm\label{lem:conc-frtoir}
    Let $C \in \{T,S,I,A\}$. If $\pi$ is a $C$-complete path from the initial state of $M_{\frRep}$, then there exists a $C$-complete path $\pi'$ from the initial state of $M_{\irRep}$ such that $\ell^-(\pi) = \ell^-(\pi')$.
\end{lem}
\begin{proof}
     Let $\pi$ be a $C$-complete path from the initial state of $M_{\frRep}$. From \cref{lem:fr2ir}, we know there exists a path $\pi'$ from the initial state of $M_{\irRep}$ such that $\ell(\pi') = \ftoi(\ell(\pi))$, where prefixes $\textit{pre}_{i}(\pi)$ of $\pi$ can be matched with prefixes from $\pi'$ that end in equivalent states for the threads. From $\ell(\pi') = \ftoi(\ell(\pi))$, it follows that $\ell^-(\pi')=\ell^-(\pi))$. It remains to establish that $\pi'$ is $C$-complete in $M_{\irRep}$.
     Towards a contradiction, assume that $\pi'$ is not $C$-complete.
     We do a case distinction on what $C$ is:
     \begin{itemize}
        \item If $C =T$, by \cref{pr:T-just paths}, it follows that $\pi'$ must thread-enable some action $a \in \nonblock$. By \cref{lem:thread-enabled-frtoir}, it follows that $\pi$ also thread-enables some action $b\in\nonblock$.
        Thus $\pi$ is not $T$-complete.

        \item If $C =S$, then by \cref{pr:S-just paths}, it follows that $\pi'$ thread-enables some action $a \in \nonblock$ such that either (a) $a \notin \textit{start}(\rid)$ for all $\rid\in\RID$ or (b) $a \in \textit{start}(\rid)$ for some $\rid\in\RID$ but $\pi'$ contains only finitely many occurrences of actions of the form $\startwrite[\tid,\rid]{\data}$. In case (a), $a$ is a thread-local or finish write action. By \cref{lem:thread-enabled-frtoir}, $\pi$ also thread-enables $a$ or a finish read action, which is also not in $\textit{start}(\rid)$ for any $\rid\in\RID$, and thus $\pi$ is not $S$-complete.
        In case (b), $a$ is an action in $\textit{start}(\rid)$ for some $\rid\in\RID$ that is thread-enabled by $\pi'$ with only finitely many occurrences of actions of the form $\startwrite[\tid,\rid]{\data}$ occurring in $\pi'$. Since $a$ is in $\textit{start}(\rid)$, it must be of the form $\readop[\tidtwo,\rid]{\data'}$ or $\startwrite[\tidtwo,\rid]{\data'}$. In the latter case, by \cref{lem:thread-enabled-frtoir} $\pi$ either thread-enables an action not in $\textit{start}(\rid)$ for any $\rid\in\RID$, or also thread-enables $a$ and there are still only finitely many write actions to $\rid$ in $\pi$, so $\pi$ is not $S$-complete. If $a$ is $\readop[\tidtwo,\rid]{\data'}$, then by \cref{lem:thread-enabled-frtoir}, $\pi$ thread-enables $\startread[\tidtwo,\rid]$ or $\finishread[\tidtwo,\rid]{\data'}$. In the former case, it is again an action in $\textit{start}(\rid)$ that is thread-enabled; in the latter case it is an action that is not in $\textit{start}(\rid)$ at all. In both cases, $\pi$ is not $S$-complete. 

        \item If $C = I$, then by \cref{pr:I-just paths},
        it follows that $\pi'$ thread-enables some action $a\in\nonblock$ such that either (a) $a \notin \textit{start}(\rid)$ for all $\rid\in\RID$, (b) $a = \startwrite[\tid,\rid]{\data}$ and $\pi'$ contains finitely many occurrences of actions $b \in \textit{start}(\rid)$, or (c) $a = \readop[\tid,\rid]{\data}$ and $\pi'$ contains finitely many occurrences of actions of the form $\startwrite[\tidtwo,\rid]{\data}$.
        We cover each case:
        \begin{itemize}
        \item Case (a) proceeds just as case (a) if $C=S$, thus $\pi$ is not $I$-complete.
        \item If $\pi'$ thread-enables an action $a=\startwrite[\tid,\rid]{\data}$, then by \cref{lem:thread-enabled-frtoir}, $\pi$ thread-enables an action not in $\textit{start}(\rid)$ for any $\rid\in\RID$, or also thread-enabled $a$. In the former case, $\pi$ is trivially not $I$-complete. In the latter case, note that in $\pi'$, there are only finitely many occurrences of actions in $\textit{start}(\rid)$. Since the start write actions are also in $\pi$ and the read actions must come from, amongst others, start read actions in $\pi$, which are also in $\textit{start}(\rid)$, the same holds for $\pi$ and so $\pi$ is not $I$-complete.
        \item If $\pi'$ thread-enables an action $a= \readop[\tid,\rid]{\data}$, then by \cref{lem:thread-enabled-frtoir}, $\pi$ thread-enables $\startread[\tid,\rid]$ or $\finishread[\tid,\rid]{\data}$. In the latter case, this action is not in $\textit{start}(\rid)$ and thus $\pi$ is not $I$-complete. In the former, there being only finitely many occurrences of start write actions to $\rid$ in $\pi'$ and thus also $\pi$ means that $\pi$ is not $I$-complete.
        \end{itemize}

        \item If $C = A$, then by \cref{pr:A-just paths},
        it follows that $\pi'$ thread-enables some action $a\in\nonblock$ such that either (a) $a \notin \textit{start}(\rid)$ for all $\rid\in\RID$, or (b) $a \in \textit{start}(\rid)$ for some $\rid\in\RID$ and $\pi'$ contains finitely many occurrences of actions $b \in \textit{start}(\rid)$.
        We cover both cases:
        \begin{itemize}
        \item Case (a) proceeds just as case (a) for $C=S$, thus $\pi$ is not $A$-complete.
        \item If $\pi'$ thread-enables an action $a\in \textit{start}(\rid)$ for some $\rid\in\RID$, then this is either $\readop[\tid,\rid]{\data}$ or $\startwrite[\tid,\rid]{\data}$ for some $\tid\in\TID,\data\in\Data[\rid]$. 
        If it is $\readop[\tid,\rid]{\data}$, then by \cref{lem:thread-enabled-frtoir} $\pi$ thread-enables $\startread[\tid,\rid]$ or $\finishread[\tid,\rid]{\data}$. In the latter case, it thread-enables an action not in $\textit{start}(\rid)$ so $\pi$ is not $A$-complete. In the former case, or if $a$ is $\startwrite[\tid,\rid]{\data}$, $\pi$ thread-enables an action in $\textit{start}(\rid)$. Additionally, since there are only finitely many occurrences of actions in $\textit{start}(\rid)$ in $\pi'$, the same holds for $\pi$: the start write actions are the same, and the start read actions are replaced with read actions. Thus, $\pi$ is not $A$-complete.
        \end{itemize}
    \end{itemize}
    We have shown that in all cases, $\pi$ is not $C$-complete, which contradicts our choice of $\pi$. Thus, $\pi'$ must be $C$-complete.
\end{proof}

\begin{thm}\rm\label{thm:wct-ir-fr}
    $M_{\frRep} =^C_\WCT M_{\irRep}$ for all $C \in \{T,S,I,A\}$.
\end{thm}
\begin{proof}
    This follows directly from \cref{lem:conc-irtofr,lem:conc-frtoir,def:wct}.
\end{proof}

\section{Register models that avoid overlap}\label{app:avoidoverlap}

We do our verification on atomic registers with different models of blocking by using the instant-read atomic register model with four different concurrency relations.
It is not immediately obvious that the resulting verifications are indeed valid for memory models where certain operations cannot overlap, since our instant-read atomic register model (\cref{fig:instantread-procatomic}) does allow read and write operations to overlap.
In this appendix, we show that the verification results we obtain are the same as if we had modelled operations that truly do not overlap.
For this, we use \cref{def:wct}.

We will argue that our verifications are valid by first considering models that truly avoid overlap, and then proving weak complete trace equivalence with our instant-read models.
Since the models that truly avoid overlap are more closely related to the full-read models, we use those as an intermediate step.

To accurately capture the memory model of blocking reads and writes, we need not only the concurrency relation $\conc_{\!A}$, but also a register model that does not allow any read or write to start when another read or write to the same register is in progress. This can be achieved by changing, in \cref{fig:fullread-procatomic}, the condition $\tid\notin(\readers{s}\cup\writers{s})$, occurring in the first and fourth lines, into
$(\readers{s}\cup\writers{s}) = \emptyset$.

Similarly, in the blocking model with concurrent reads, where reads and writes have to await the completion of in-progress writes, but only writes have to await the completion of in-progress reads, these two lines of \cref{fig:fullread-procatomic} become
  \[\begin{array}{lr@{~\then~}l}
            & (\tid\mathbin{\notin}\readers{s} \wedge \writers{s}=\emptyset) &
                 \startread[\tid,\rid]\co\Reg[\frRep,\atoRep](\rid,\usr{s,\tid}) \\
          + & ((\readers{s}\cup\writers{s}) = \emptyset) &
                 \startwrite[\tid,\rid]{\data}\co\Reg[\frRep,\atoRep](\rid,\usw{s,\tid,\data})
  \end{array}\]

To capture the model with blocking writes and non-blocking reads, we make the same alterations as above, except that writes do not need to wait for in-progress reads. Hence, these two lines become 
\[\begin{array}{lr@{~\then~}l}
            & (\tid\mathbin{\notin}\readers{s} \wedge \writers{s}=\emptyset) &
                 \startread[\tid,\rid]\co\Reg[\frRep,\atoRep](\rid,\usr{s,\tid}) \\
          + & (\writers{s} = \emptyset) &
                 \startwrite[\tid,\rid]{\data}\co\Reg[\frRep,\atoRep](\rid,\usw{s,\tid,\data})
  \end{array}\]
Note that this only models the blocking behaviour described in \cref{sec:introduction}; we do not model that a write causes a read to be aborted and subsequently resumed.

For a given mutual exclusion algorithm, let $M_{\frRep}$ be the LTS of its full-read thread-register model, using the full-read atomic register model from \cref{fig:fullread-procatomic} for all registers, and let $M_{\irRep}$ be its instant-read equivalent. Moreover, let $M_A$ be the variant of $M_{\frRep}$ that employs the modified register model for blocking reads and writes; let $M_I$ be the variant for the blocking model with concurrent reads, and let $M_S$ be the one for the model of blocking writes and non-blocking reads.

Now, using $A$, $I$ and $S$ as abbreviations for the completeness criteria $\justact{\conc_{\!A}}{\block}$, $\justact{\conc_I}{\block}$ and $\justact{\conc_S}{\block}$, we will show that $M_A =^A_\WCT M_{\frRep}$, $M_I =^I_\WCT M_{\frRep}$ and $M_S =^S_\WCT M_{\frRep}$. This implies that it makes no difference whether, for the verifications based on the (partly) blocking memory models, we use the register models fine-tuned for that memory model as described above, or the model for full-read atomic registers from \cref{fig:fullread-procatomic}.

In the subsequent argument, let $C \in \{A, I, S\}$.

First, we must establish that $\conc_C$ is a valid concurrency relation for $M_C$. That the relations are disjoint with respect to $\actequivsym$ solely depends on the definitions of these relations and set actions of an LTS, and hence still clearly holds within the context of $M_A, M_I$ and $M_S$.
It is less immediately obvious that the \hyperlink{second}{second property of concurrency relations} still holds, particularly because the fine-tuned models are not thread-consistent: it is for example possible for $\startread[\tid,\rid]$ to be disabled by the occurrence of $\startwrite[\tidtwo,\rid]{\data}$ in all three models, even when $\tid \neq \tidtwo$.
However, in all three cases, if an action $a$ is disabled by the occurrence of an action $b$, then $a \nconc_C b$:
\begin{itemize}
    \item In $M_A$, $\startread[\tid,\rid]$ and $\startwrite[\tid, \rid]{\data}$ can be disabled by the occurrence of $\startread[\tidtwo, \rid]$ or $\startwrite[\tidtwo, \rid]{\data'}$. Accordingly, in $\conc_{\!A}$ start read actions and start write actions both interfere with both start read and start write actions on the same register.
    \item In $M_I$, $\startread[\tid,\rid]$ can be disabled by the occurrence of $\startwrite[\tidtwo,\rid]{\data}$. Accordingly, $\conc_I$ allows start write actions to interfere with start read actions on the same register. Additionally, $\startwrite[\tid,\rid]{\data}$ can be disabled by the occurrence of $\startread[\tidtwo,\rid]$ or $\startwrite[\tidtwo,\rid]{\data'}$. Indeed, $\conc_I$ allows start reads and start writes to interfere with start writes on the same register.
    \item In $M_S$, $\startread[\tid,\rid]$ and $\startwrite[\tid, \rid]{\data}$ can both be disabled by an occurrence of $\startwrite[\tidtwo, \rid]{\data'}$. Accordingly, $\conc_S$ allows start writes to interfere with both start reads and start writes on the same register.
\end{itemize}
Thus, in all three register models an action can only be disabled by the occurrence of an action that interferes with it according to the appropriate concurrency relation.\vspace{-2pt} Consequently, if there is a transition $s \xrightarrow{b} s'$ such that an action $a$ is enabled in $s$ and $a \conc_C b$, then $a$ is enabled in $s'$.\footnote{Recall that for the actions present in these models $\actequivsym$ equals $\mathit{id}$.} By induction, similarly to the proof of \cref{lem:concT-val} in \cref{app:concTproof}, \hyperlink{second}{the second property of concurrency relations} is satisfied, and we can apply these concurrency relations to their respective models.

Next, we will show that the most relevant results of \cref{app:characterise} apply also to the register models introduced above.
First observe that \cref{lem:1,lem:2} apply equally well, with the same proofs, when reading $M_A$, $M_I$ or $M_S$ for $M_{\frRep}$, since the only difference between these models is when start actions are enabled, and these lemmas refer only to order and finish actions.

\begin{lem}\rm\label{lem:register enablings}
Let $R_r$ be the LTS of any full-read register $r$, either defined as in one of the
\cref{fig:fullread-procsafe,fig:fullread-procregular,fig:fullread-procatomic}, or using one of the three modified register models defined at the beginning of this appendix.
Then any state $s_r$ of $R_r$ enables either
(i) $\startread[\tid,\rid]$ and $\startwrite[\tid,\rid]{d}$ for all $t\in\TID$ and all $d \in \Data$, or
(ii) $\orderread[\tid,\rid]$ or $\orderwrite[\tid,\rid]$ for some $t\in\TID$, or
(iii) $\finishwrite[\tid,\rid]$ or $\finishread[\tid,\rid]{d}$ for some $t\in\TID$ and $d \in \Data$.
\end{lem}
\begin{proof}
In case  $\readers{s_r} \cup \writers{s_r} = \emptyset$, option (i) applies.
This can be seen from \cref{fig:fullread-procsafe,fig:fullread-procregular,fig:fullread-procatomic}, and from the description of the three register models defined at the beginning of this appendix.  Likewise, in case $\readers{s_r} \cup \writers{s_r} \neq \emptyset$, one of (ii) or (iii) must apply.  (In the case of regular registers, we use the easily checked invariant that if $t \in \readers{s}$ then $\posval{s,t}\neq\emptyset$.)
\end{proof}

\noindent We further refine this for the model $M_S$.

\begin{lem}\rm\label{lem:register enablings S}
Let $R_r$ be the LTS of any register $r$, as defined at the beginning of this appendix for the register model with blocking writes and non-blocking reads, i.e.\ $M_S$.
Then any state $s_r$ of $R_r$ enables either
(i) $\startwrite[\tid,\rid]{d}$ for all $t\in\TID$ and all $d \in \Data$, or
(ii) $\orderwrite[\tid,\rid]$ for some $t\in\TID$, or
(iii) $\finishwrite[\tid,\rid]$ for some $t\in\TID$.
If case (i) applies then, for each $t'\in\TID$, $s_r$ enables either
(a) $\startread[\tid',\rid]$, or
(b) $\orderread[\tid',\rid]$, or
(c) $\finishread[\tid',\rid]{d}$ for some $d \in \Data$.
\end{lem}
\begin{proof}
In case $\writers{s_r} = \emptyset$, option (i) applies.
This can be seen from the description of this register model.  Likewise, in case $\writers{s_r} \neq \emptyset$, one of (ii) or (iii) must apply.

Assuming $\writers{s_r} \mathbin= \emptyset$, if $t'\mathbin{\notin}\readers{s_r}$, option (a) applies and otherwise (b) or (c) apply.\!\!\!
\end{proof}

\begin{prop}\rm\label{pr:just paths}
The characterisations of the $\block$-$\conc_C$-\hyperlink{just}{just} paths of $M_{\frRep}$ from \cref{pr:A-just paths,pr:I-just paths,pr:S-just paths} apply equally well to $M_C$.
\end{prop}

\begin{proof}
Below we give a proof for the case $C=A$, most of which applies to all three cases $C \in \{A,I,S\}$. The parts that are specific to the case $C=A$ are coloured {\color{highlightColour}{\colourname}}, and afterwards we will give the case $C=S$ with the differences between it and $C=A$ similarly highlighted. The case $C=I$ is treated at the end.

We must prove that a path $\pi$ starting in the initial state of $M_A$, is $\block$-$\conc_A$-just if, and only if, a) $\pi$ thread-enables no actions $a \in \nonblock$ other than actions from $\textit{start}(r)$ for some $r \in \RID$, and b) if, for some $r \in \RID$, an action $a \in \textit{start}(r)$ is thread-enabled by $\pi$, then $\pi$ contains infinitely many occurrences of actions $b \in \textit{start}(r)$.

Suppose $\pi$ thread-enables an action $a \in \nonblock$, say with $t= \thrmap{a}$ and $r= \regmap{a}$, such that {\color{highlightColour}if $a \in \textit{start}(r)$ then $\pi$ contains only finitely many actions $b\in\textit{start}(r)$. In the latter case $\pi$ contains only finitely many actions $b$ with $\regmap{b}=r$},
because actions $c$ with $\thrmap{c}=t'$ and $\regmap{c}=r$ must occur strictly in the order $\startwrite[\tid',\rid]{d}$ -- $\orderwrite[\tid',\rid]$ -- $\finishwrite[\tid',\rid]$.

We have to show that $\pi$ is \hyperlink{just}{not $\block$-$\conc_{C}$-just}. Let $\pi'$ be a suffix of $\pi$ in which no actions $b$ with $\thrmap{b}=t$ occur; {\color{highlightColour}in case $a \in \textit{start}(r)$ we moreover choose $\pi'$ such that it contains no actions $b$ with $\regmap{b}=r$.} Let $s$ be the initial state of $\pi'$. So $s_{\idx{t}} = \finish_t(\pi)$ enables $a$.

In case $r \mathbin=\undefsymb$, as $s_{\idx{t}}$ enables $a$ and $a$ does not require synchronisation with any register, also $s$ enables $a$. As $\pi'$ does not contain actions $b$ with $a \nconc_{C} b$, the path $\pi$ is \hyperlink{just}{not just}.

So assume that $r \in \RID$. We proceed with a case distinction on the action $a$, which must be of the form $\startread[\tid,\rid]$, $\startwrite[\tid,\rid]{d}$, $\finishread[\tid,\rid]{d}$ or $\finishwrite[\tid,\rid]$, since it is an action that occurs in a thread LTS and $r \neq \undefsymb$.

First assume that $a = \finishwrite[\tid,\rid]$ or $a= \finishread[\tid,\rid]{d}$ for some $d \in \Data$. By \cref{lem:2}, either $\orderwrite[\tid,\rid]$ or $\orderread[\tid,\rid]$ or $\finishwrite[\tid,\rid]$ or $\finishread[\tid,\rid]{d'}$ for some $d'\in\Data$ is enabled by state $s_{\idx{r}}$ in $R_r$.
In case $s_{\idx{r}}$ enables $c=\orderread[\tid,\rid]$ or $c=\orderwrite[\tid,\rid]$, also $s$ enables $c$, since $c$ does not require synchronisation with thread $t$. As $\pi'$ contains no actions $b$ with $c \nconc_{C} b$, using that $\thrmap{c}=t$, the path $\pi$ is \hyperlink{just}{not just}.
In case $s_{\idx{r}}$ enables an action $c=\finishread[\tid,\rid]{d}$ or $c=\finishwrite[\tid,\rid]$, by \cref{lem:1} also state $s$ enables $c$. As $\pi'$ contains no actions $b$ with $c \nconc_C b$, the path $\pi$ is \hyperlink{just}{not just}.

Thus we may restrict attention to the case that $a = \startread[\tid,\rid]$ or $a = \startwrite[\tid,\rid]{d}$ for some $d \in \Data$, that is, $a\in \textit{start}(r)$. So henceforth we may use that $\pi'$ {\color{highlightColour}contains no actions $b$ with $\regmap{b}=r$.}

{\color{highlightColour}We proceed by making a case distinction on the three possibilities described by \cref{lem:register enablings} for actions enabled by $s_{\idx{r}}$.
In case (i) of \cref{lem:register enablings}, $s_{\idx{r}}$ enables $\startread[\tid',\rid]$ and $\startwrite[\tid',\rid]{d'}$ for all $t'\in\TID$ and all $d' \in \Data$, and thus in particular the action $a$.} Consequently, also  $s$ enables $a$. As $\pi'$ does not contain actions $b$ with $\thrmap{b}=t$ {\color{highlightColour}or actions $b \in \textit{start}(r)$}, it contains no action $b$ with $a \nconc_{C} b$. Hence, the path $\pi$ is \hyperlink{just}{not just}.

In case (ii) $s_{\idx{r}}$, and hence also $s$, enables {\color{highlightColour}$\orderread[\tid',\rid]$ or $\orderwrite[\tid',\rid]$} for some $t'\mathbin\in\TID$. Thus, by \cref{lem:1}, $s_{\idx{t'}}$ enables an action {\color{highlightColour}$\finishwrite[\tid',\rid]$ or $\finishread[\tid',\rid]{d}$}. By the assumptions on thread behaviour from \cref{sec:threads} (\cref{thread-prop-2,thread-prop-3,thread-prop-4}), this implies that $s_{\idx{t'}}$ only enables {\color{highlightColour}actions $b$ with $\regmap{b}=r$}. Since $\pi'$ does not contain such $b$, it cannot contain any action $c$ with $\thrmap{c}=t'$ either.
Consequently, $\pi$ is not just.

In case (iii) $s_{\idx{r}}$ enables an action {\color{highlightColour}$c=\finishwrite[\tid',\rid]$ or $c=\finishread[\tid',\rid]{d}$ for some $t'\in\TID$ and $d \in \Data$}.
Thus, by \cref{lem:1}, $c$ is enabled by $s_{\idx{t'}}$ as well as by $s$.
The rest of the argument proceeds just as for case (ii) above.

For the other direction, note that $M_C$ has the same states as $M_{\frRep}$ and can be obtained from $M_{\frRep}$ by leaving out some transitions. This implies that any path $\pi$ in $M_C$ is also a path in $M_{\frRep}$, and moreover, if $\pi$ is 
\hyperlink{just}{$\block$-$\conc_{C}$-unjust} in $M_C$ then it certainly is 
\hyperlink{just}{$\block$-$\conc_{C}$-unjust} in $M_{\frRep}$. Using this, the reverse direction of this proof is direct consequence of \cref{pr:A-just paths,pr:I-just paths,pr:S-just paths}.
\vspace{2ex}

\noindent
In the case $C = S$, we must prove that a path $\pi$ starting in the initial state of $M_S$, is $\block$-$\conc_S$-just if, and only if, a) $\pi$ thread-enables no actions $a \in \nonblock$ other than actions from $\textit{start}(r)$ for some $r \in \RID$, and b) if an action $a \in \textit{start}(r)$ is thread-enabled by $\pi$, then $\pi$ contains infinitely many occurrences of actions $b$ of the form $\startwrite[t',r]{d}$ for some $t' \in \TID$ and $d' \in \Data$. 
The proof is identical to the $C=A$ case, save for the {\colourname} phrases.

Suppose $\pi$ thread-enables an action $a \in \nonblock$, say with $t= \thrmap{a}$ and $r= \regmap{a}$, such that {\color{highlightColour}if $a \in \textit{start}(r)$ then $\pi$ contains only finitely many actions $b$ of the form $\startwrite[\tid',\rid]{d}$ for some $\tid'\in\TID$ and $d'\in\Data$. In the latter case $\pi$ also contains only finitely many actions $b$ of the form $\orderwrite[\tid',\rid]$ or $\finishwrite[\tid',\rid]$},
because actions $c$ with $\thrmap{c}=t'$ and $\regmap{c}=r$ must occur strictly in the order $\startwrite[\tid',\rid]{d}$ -- $\orderwrite[\tid',\rid]$ -- $\finishwrite[\tid',\rid]$.

We have to show that $\pi$ is \hyperlink{just}{not $\block$-$\conc_{C}$-just}. Let $\pi'$ be a suffix of $\pi$ in which no actions $b$ with $\thrmap{b}=t$ occur; {\color{highlightColour}in case $a \in \textit{start}(r)$ we moreover choose $\pi'$ such that it contains no actions $b$ of the form $\startwrite[\tid',\rid]{d}$, $\orderwrite[\tid',\rid]$ or $\finishwrite[\tid',\rid]$ for some $\tid'\in\TID$ and $d'\in\Data$.} Let $s$ be the initial state of $\pi'$. So $s_{\idx{t}} = \finish_t(\pi)$ enables $a$.

In case $r \mathbin=\undefsymb$, as $s_{\idx{t}}$ enables $a$ and $a$ does not require synchronisation with any register, also $s$ enables $a$. As $\pi'$ does not contain actions $b$ with $a \nconc_{C} b$, the path $\pi$ is \hyperlink{just}{not just}.

So assume that $r \in \RID$. We proceed with a case distinction on the action $a$, which must be of the form $\startread[\tid,\rid]$, $\startwrite[\tid,\rid]{d}$, $\finishread[\tid,\rid]{d}$ or $\finishwrite[\tid,\rid]$, since it is an action in a thread LTS and $r \neq \undefsymb$.

First assume that $a = \finishwrite[\tid,\rid]$ or $a= \finishread[\tid,\rid]{d}$ for some $d \in \Data$. By \cref{lem:2}, either $\orderwrite[\tid,\rid]$ or $\orderread[\tid,\rid]$ or $\finishwrite[\tid,\rid]$ or $\finishread[\tid,\rid]{d'}$ for some $d'\in\Data$ is enabled by state $s_{\idx{r}}$ in $R_r$.
In case $s_{\idx{r}}$ enables $c=\orderread[\tid,\rid]$ or $c=\orderwrite[\tid,\rid]$, also $s$ enables $c$, since $c$ does not require synchronisation with thread $t$. As $\pi'$ contains no actions $b$ with $c \nconc_{C} b$, using that $\thrmap{c}=t$, the path $\pi$ is \hyperlink{just}{not just}.
In case $s_{\idx{r}}$ enables an action $c=\finishread[\tid,\rid]{d}$ or $c=\finishwrite[\tid,\rid]$, by \cref{lem:1} also state $s$ enables $c$. As $\pi'$ contains no actions $b$ with $c \nconc_C b$, the path $\pi$ is \hyperlink{just}{not just}.

Thus we may restrict attention to the case that $a = \startread[\tid,\rid]$ or $a = \startwrite[\tid,\rid]{d}$ for some $d \in \Data$, that is, $a\in \textit{start}(r)$. So henceforth we may use that $\pi'$ {\color{highlightColour}contains no actions $b$ of the form $\startwrite[\tid',\rid]{d}$, $\orderwrite[\tid',\rid]$ or $\finishwrite[\tid',\rid]$ for some $\tid'\in\TID$ and $d'\in\Data$.}

{\color{highlightColour}We proceed by making a case distinction on the three possibilities described by \cref{lem:register enablings S} for actions enabled by $s_{\idx{r}}$.
In case (i) of \cref{lem:register enablings S}, we make a further case distinction, depending on whether $a = \startread[\tid,\rid]$ or $a = \startwrite[\tid,\rid]{d}$. In the latter case,
$s_{\idx{r}}$ enables $\startwrite[\tid',\rid]{d'}$ for all $t'\in\TID$ and all $d' \in \Data$, and thus in particular the action $a$. In the former case, options (b) and (c) for $t'=t$ of \cref{lem:register enablings S} are ruled out, because in those cases \cref{lem:1} would imply that an action $\finishread[\tid,\rid]{d}$ would be enabled by $s_\idx{t}$, but by assumption on thread behaviour (see \cref{sec:threads}, specifically \cref{thread-prop-2,thread-prop-4}) this cannot happen in a state enabling $\startread[\tid,\rid]$. Thus (a) applies, and also $s_{\idx{r}}$ enables action $a$.} Consequently, also  $s$ enables $a$. As $\pi'$ does not contain actions $b$ with $\thrmap{b}=t$ {\color{highlightColour}or actions $b\in\textit{start}(r)$ with $\iswrite{b}$}, it contains no action $b$ with $a \nconc_{C} b$. Hence, the path $\pi$ is \hyperlink{just}{not just}.

In case (ii) $s_{\idx{r}}$, and hence also $s$, enables {\color{highlightColour}$\orderwrite[\tid',\rid]$} for some $t'\mathbin\in\TID$. Thus, by \cref{lem:1}, $s_{\idx{t'}}$ enables an action {\color{highlightColour}$\finishwrite[\tid',\rid]$}. By the assumptions of \cref{sec:threads}, this implies that $s_{\idx{t'}}$ only enables {\color{highlightColour}the action $b=\finishwrite[\tid',\rid]$}. Since $\pi'$ does not contain such $b$, it cannot contain any action $c$ with $\thrmap{c}=t'$ either.
Consequently, $\pi$ is not just.

In case (iii) $s_{\idx{r}}$ enables an action {\color{highlightColour}$c=\finishwrite[\tid',\rid]$ for some $t'\in\TID$}.
Thus, by \cref{lem:1}, $c$ is enabled by $s_{\idx{t'}}$ as well as by $s$.
The rest of the argument proceeds just as for case (ii) above.

For the other direction, note that $M_C$ has the same states as $M_{\frRep}$ and can be obtained from $M_{\frRep}$ by leaving out some transitions. This implies that any path $\pi$ in $M_C$ is also a path in $M_{\frRep}$, and moreover, if $\pi$ is 
\hyperlink{just}{$\block$-$\conc_{C}$-unjust} in $M_C$ then it certainly is 
\hyperlink{just}{$\block$-$\conc_{C}$-unjust} in $M_{\frRep}$. Using this, the reverse direction of this proof is direct consequence of \cref{pr:A-just paths,pr:I-just paths,pr:S-just paths}.
\vspace{2ex}

\noindent
Finally, for the case $C=I$ we must prove that a path $\pi$ starting in the initial state of $M_I$ is $\block$-$\conc_I$-just if, and only if, a) $\pi$ thread-enables no actions $a \in \nonblock$ other than actions from $\textit{start}(r)$ for some $r \in \RID$, and b) if an action $\startwrite[\tid,\rid]{\data}$ is thread-enabled by $\pi$, then $\pi$ contains infinitely many occurrences of actions $b \in \textit{start}(r)$, and c) if an action $\startread[\tid,\rid]$ is thread-enabled by $\pi$, then $\pi$ contains infinitely many occurrences of actions $b$ of the form $\startwrite[\tidtwo,\rid]{\data}$.
Let $\pi$ be a path in $M_I$ that thread-enables an action $ a\in \nonblock$ such that one of these three conditions is violated. In the case that the violated condition is a) or b), the proof proceeds just as in the case $C=A$; in the case it is condition c) the proof proceeds just as in the case $C=S$. The reverse direction goes exactly as in the case $C=A$.
\end{proof}

\begin{lem}\rm\label{lem:MCtoM}
    Let $C \in \{A, I, S\}$. If $\pi$ is a $C$-complete path in $M_C$ starting in the initial state of $M_C$, then it is also a $C$-complete path in $M_{\frRep}$ starting in the initial state of $M_{\frRep}$.
\end{lem}
\begin{proof}
    Let $\pi$ be a $C$-complete path in $M_C$ starting in its initial state. 
    Since for all three variants, $M_C$ has stricter conditions for actions being enabled than $M_{\frRep}$ does, $\pi$ is guaranteed to exist in $M_{\frRep}$. It remains to show that $\pi$ is $C$-complete in $M_{\frRep}$. This, however, is an immediate consequence of \cref{pr:just paths}, from which it follows that a path starting in the initial state of $M_C$ is $\block$-$\conc_C$-just exactly when that same path starting in the initial state of $M_{\frRep}$ is.
\end{proof}

We also wish to prove the other direction. However, it is not necessarily the case that any $C$-complete path in $M_{\frRep}$ is also a $C$-complete path in $M_C$, if only because not every path in $M_{\frRep}$ exists in $M_C$. Thus, we need a way to convert a path in $M_{\frRep}$ into a path that is guaranteed to exist in $M_C$.
We define a rewrite relation on paths, where $\pi \leadsto \pi'$ denotes that $\pi$ is rewritten into $\pi'$. Namely $\pi \leadsto \pi'$ if, and only if, there exists paths $\sigma$ and $\rho$, states $s$ and $s'$, thread $t\in\TID$, and an action $b$ with $\thrmap{b}\neq \tid$, such that
\begin{itemize}
\item either $\pi = \sigma ~\startread[\tid,\rid]~s~b~\rho$ and $\pi' = \sigma~b~s'~\startread[\tid,\rid]~\rho$,
\item or $\pi = \sigma ~\startwrite[\tid,\rid]{d}~s~b~\rho$ and $\pi' = \sigma~b~s'~\startwrite[\tid,\rid]{d}~\rho$,
\item or $\pi = \sigma ~b~s~\finishread[\tid,\rid]{d}~\rho$ and $\pi' = \sigma~\finishread[\tid,\rid]{d}~s'~b~\rho$,
\item or $\pi = \sigma ~b~s~\finishwrite[\tid,\rid]~\rho$ and $\pi' = \sigma~\finishwrite[\tid,\rid]~s'~b~\rho$.
\end{itemize}
Thus, this rewrite relation moves any start read or start write action forwards by swapping it with an action $b$ of another thread, and in the same manner moves any finish read or finish write action backwards. 
\cref{prop:swapping} together with \cref{obs:trivial-cong} show that whenever a path $\pi$ of any of the above forms exists in $M_{\frRep}$, the corresponding $\pi'$ also exists in $M_{\frRep}$.\footnote{For convenience, we here identify paths whose states differ merely up to $\cong$, as defined in \cref{app:additionalresults}.}
Since paths may be infinite, this rewrite system need not terminate, but in the limit it converts any path $\pi$ into a normal form $\norm{\pi}$ in which each action $\orderread[\tid,\rid]$ is immediately preceded by the corresponding $\startread[\tid,\rid]$ and immediately followed by the corresponding $\finishread[\tid,\rid]{d}$, and likewise for write actions. To convergence to this limit, one should give priority to rewrite steps that apply closer to the beginning of the path.
This normal form also exists in $M_A$, $M_I$ and $M_S$.

\begin{lem}\rm\label{lem:MtoMC}
    Let $C \in \{A, I, S\}$. If $\pi$ is a $C$-complete path in $M_{\frRep}$ starting in the initial state of $M_{\frRep}$, then $\norm{\pi}$ is a $C$-complete path in $M_C$ starting in the initial state of $M_C$.
\end{lem}
\begin{proof}
    Let $\pi$ be a $C$-complete path in $M_{\frRep}$ starting in its initial state. Trivially, $\norm{\pi}$ exists in $M_C$. It remains to show that $\norm{\pi}$ is $C$-complete in $M_C$.
    Note that $M_C$ is a subgraph of $M_{\frRep}$. Hence, if $\norm{\pi}$ is $C$-complete in $M_{\frRep}$, then it is also $C$-complete in $M_C$. 
    It therefore suffices to prove that $\norm{\pi}$ is $C$-complete in $M_{\frRep}$.
This is an immediate consequence of \cref{pr:A-just paths,pr:I-just paths,pr:S-just paths}, given that $\pi$ and $\norm{\pi}$ have the same sets of thread-enabled actions and $\textit{occ}_\pi = \textit{occ}_{\norm{\pi}}$ (i.e., each action occurs as often in $\norm{\pi}$ as it occurs in $\pi$).
\end{proof}

\begin{thm}\label{thm:avoid-fullread}\rm
$M_A =^A_\WCT M_{\frRep}$, $M_I =^I_\WCT M_{\frRep}$ and $M_S =^S_\WCT M_{\frRep}$.
\end{thm}
\begin{proof}
We prove that $\WCT_C(M_{\frRep}) = \WCT_C(M_C)$ for all $C \in \{A, I, S\}$. We prove this by mutual set inclusion.

First, let $\pi$ be a $C$-complete path from the initial state of $M_C$, such that $\ell^-(\pi) \mathbin\in \WCT_C(M_C)$. By \cref{lem:MCtoM}, $\pi$ is also a $C$-complete path from the initial state of $M_{\frRep}$. Thus, $\ell^-(\pi) \in \WCT_C(M_{\frRep})$.

Second, let $\pi$ be a $C$-complete path from the initial state of $M_{\frRep}$, such that $\ell^-(\pi) \mathbin\in \WCT_C(M_{\frRep})$. Then by \cref{lem:MtoMC}, $\norm{\pi}$ is a $C$-complete path in $M_C$. Note that $\ell^-(\norm{\pi}) = \ell^-(\pi)$, hence $\ell^-(\pi) \in \WCT_C(M_C)$.

We have proven $M_C =^C_\WCT M_{\frRep}$ for all $C \in \{A, I, S\}$.
\end{proof}
\noindent
And thus, applying \cref{thm:wct-ir-fr}:
\begin{cor}\rm
    $M_A =^A_\WCT M_{\irRep}$, $M_I =^I_\WCT M_{\irRep}$ and $M_S =^S_\WCT M_{\irRep}$.
\end{cor}

\section{mCRL2 models}\label{app:mcrl2}
In this section, we comment on some of the details of how we translate our models of threads and registers as presented in this paper to the mCRL2 language. For information on the mCRL2 language, we refer to \cite{mCRL2language}.
e modelled both the full-read and instant-read register models, but the verifications were all done with the instant-read models.\footnote{The verifications in \cite{GlabbeekLS25}, however, were done with the full-read models.}
We made several alterations while translating the process-algebraic definitions of the registers (see \cref{app:registers}) to mCRL2 processes.
Some of these were necessary to obtain valid mCRL2 models, others were employed to reduce the state space of the models.
All models are available as supplementary material \cite{GLSzenodo}.

The most important differences between the definitions from \cref{app:registers} and the mCRL2 models are as follows:
\begin{itemize}
    \item Since the safe, regular and atomic models do not require exactly the same information from the status object, and tracking unused information unnecessarily increases the state space of the model, we separate the status object into different variants for each register type, that only track the required information for that type.
    \item We reset values to a pre-defined default whenever we know that the value is no longer relevant. For example, the safe register model only needs to know the value that a thread intends to write if this write operation does not encounter an overlapping write; if there are multiple concurrent writes, it will not matter what their intended values were. Hence, in the safe register model we only track a single value instead of a mapping from threads to values for $\valssym$, and reset this value to its default whenever we observe zero or more than one active writer.
    \item For full-read regular registers, instead of computing on the spot what the values of all active writes are whenever a read starts, we use a multiset to keep track of those values as writes start and finish; this is more straightforward and, since it adds no new information, does not expand the state-space.
    \item Similarly for the full-read regular registers, we further reduce the state space by adding the value of a new writer only to the $\posvalsym$ of active readers, rather than all threads. Since $\posvalsym$ is reset for a thread whenever it starts a read, this does not affect the behaviour of the model.
    \item We have moved summations inwards whenever possible, so that $\Data$ is only summed over when the resulting value is actually relevant.
\end{itemize}

In addition to the registers, the threads must also be modelled in mCRL2. 
We already mentioned the relevant choices made here in \cref{sec:modelmutex}.
It remains to discuss how the parallel composition as defined in \cref{sec:preliminaries} and employed in \cref{sec:thrregmodels} is modelled in mCRL2. In order to get the right communication between threads and registers, we create ``sending'' and ``receiving'' versions of all register interface actions, and define a communication function so that the two versions together form the correct action.
For instance, to start a write of value $\data$ to register $\rid$, thread $\tid$ does the action $\mathit{start\_write\_s}(\tid, \rid, \data)$. The register simultaneously does the action $\mathit{start\_write\_r}(\tid,\rid,\data)$, and this communication appears in the model as $\mathit{start\_write}(\tid,\rid,\data)$. 
The register-local and thread-local actions do not need such a modification, since they are only performed by a single component.

It is worth noting that the mCRL2 version of the modal $\mu$-calculus does not support quantification over actions directly, but does allow quantification over the data parameters of actions. In order to express the formulae we need (see \cref{app:mucalc}), we therefore add the action $\mathit{label}$ to every action in the model, creating multi-actions. We give the $\mathit{label}$ action a parameter over a set of labels $\mathbb{L}$, where each label corresponds to one of the action names used in \cref{app:registers}. We then hide the original actions, so that only $\mathit{label}$ is visible in the model. This way, we can refer to the labels in our formulae, which circumvents the issue of not being able to quantify over actions. This approach is based on \cite{bouwman2020off}.
We capture the equivalence relation $\actequivsym$ on actions by removing the data parameter from the labels used for the read actions in the instant-read model. Thus, $\mathit{read}(\tid,\rid,\data)$ and $\mathit{read}(\tid,\rid,\data')$ both get the label $\readopempty[\tid,\rid]$.

\section{Modal \texorpdfstring{$\mu$}{mu}-calculus formulae}\label{app:mucalc}
In \cite{spronck2024progress}, we presented template formulae that can be instantiated to capture liveness properties that fit Dwyer, Avrunin and Corbett's Property Specification Patterns (PSP) \cite{dwyer1999patterns}, incorporating a variety of completeness criteria, including justness.

In this appendix, we recap the three correctness properties we verified, and give the modal $\mu$-calculus formulae for these properties. 
To be able to use the template formulae with justness for the liveness properties, we must show how they fit into PSP.
We do not thoroughly explain the modal $\mu$-calculus here; instead, we refer to \cite{spronck2024progress} for more information on how these formulae should be understood.
We do not use the mCRL2 modal $\mu$-calculus syntax here, since it results in longer and less readable formulae. However, the files are all given in the supplementary material \cite{GLSzenodo}.

\paragraph{Mutual exclusion.}
The mutual exclusion property says that at any given time, at most one thread will be in its critical section.
In our models, a thread $\tid$ accessing its critical section is represented by the action $\crit[\tid]$.
We reformulate the mutual exclusion property for our models as: for all threads $\tid$ and $\tidtwo$ such that $\tid \neq \tidtwo$, at any time it is impossible that both $\crit[\tid]$ and $\crit[\tidtwo]$ are enabled.
This is captured by the following modal $\mu$-calculus formula:
\begin{equation}
    \bigwedge_{\tid, \tidtwo \in \TID}( (\tid \neq \tidtwo) \imps \boxm{\clos{\allact}}\neg(\diam{\crit[\tid]}\tp \land \diam{\crit[\tidtwo]}\tp))
\end{equation}

\paragraph{Deadlock freedom.}
The deadlock freedom property says that whenever at least one thread is running its entry protocol, eventually some thread will enter its critical section.
Since we want to apply the results from \cite{spronck2024progress}, we need to first show how this can be represented in PSP, using the actions in our model.
This property fits into the global response pattern of \cite{dwyer1999patterns}: whenever the trigger occurs, in this case one thread being in its entry protocol, we want the response to occur eventually, in this case some thread executing its critical section.
We cannot directly apply the results from \cite{spronck2024progress}, however, since there we require the trigger to be a set of actions and a thread ``being in the entry protocol'' cannot simply be captured by the occurrence of an action. 
Instead, a thread $\tid$ is in its entry protocol between the execution of $\noncrit[\tid]$ and the subsequent execution of $\crit[\tid]$. 
Fortunately, this requires only a minor deviation from the template given in \cite{spronck2024progress}, namely by allowing the trigger to be a regular expression over sets of actions instead.

Assume that for $C \in \{T, S, I, A\}$ and all actions $a$ in our model, we have a function $\elimf{C}{a}$ that maps an action to all actions that interfere with it according to concurrency relation $C$.
We have encoded these functions in our mCRL2 models.
For clarity, let $\crit[\mathit{all}] = \bigcup_{\tid \in \TID}\{\crit[\tid]\}$.

In the modal $\mu$-calculus formula, we capture that there \emph{does not} exist a path from the initial state of the model that violates deadlock freedom and is just. This is equivalent to checking that deadlock freedom is satisfied on all just paths.
The formula \[\nu X.(\bigwedge_{a \in \nonblock}( \diam{a}\tp \imps \diam{\clos{\comp{\crit[\mathit{all}]}}}(\diam{\elimf{C}{a} \setminus \crit[\mathit{all}]}X)))\] holds in a state $s$ iff there is a path starting in $s$ that (i) is just, in the sense that any state enabling an action $a \notin \block$ is followed by an action from $\elimf{C}{a}$, and (ii) does not contain any action from  $\crit[\mathit{all}]$.
In other words, this formula says that from a state $s$, there is a just path where the response never occurs. We then merely need to prepend the formula for an occurrence of the trigger, namely an occurrence of $\noncrit[\tid]$, for some $\tid\in\TID$, that is not (yet) followed by $\crit[\tid]$, and negate the resulting formula to capture that the initial state of the model does not admit just paths that violate deadlock freedom.

The final formula for deadlock freedom under $\justact{\conc_{C}}{\block}$ is:
\begin{equation}
    \neg\diam{\clos{\allact}}\bigvee_{\tid \in \TID}(\diam{\noncrit[\tid] \co \clos{\comp{\crit[\tid]}}} \nu X.(\bigwedge_{a \in \nonblock}( \diam{a}\tp \imps \diam{\clos{\comp{\crit[\mathit{all}]}}}(\diam{\elimf{C}{a} \setminus \crit[\mathit{all}]}X))))
\end{equation}

\paragraph{Starvation freedom.}
Starvation freedom says that whenever a thread leaves its non-critical section, it will eventually enter its critical section.
Unlike deadlock freedom, starvation freedom does fit into the global response pattern using only sets of actions: the trigger is the occurrence of $\noncrit[\tid]$ for some $\tid \in \TID$, and the response is the occurrence of $\crit[\tid]$.
Using the same functions $\elimf{C}{a}$ for $C \in \{T, S, I, A\}$ and actions $a$ as the deadlock freedom formula, we get the following modal $\mu$-calculus formula for starvation freedom under $\justact{\conc{_C}}{\block}$:
\begin{equation}
    \bigwedge_{\tid \in \TID}(\neg \diam{\clos{\allact} \co \noncrit[\tid]}\nu X. (\bigwedge_{a \in \nonblock}(\diam{a}\tp \imps \diam{\clos{\comp{\crit[\tid]}}}(\diam{\elimf{C}{a} \setminus \crit[\tid]}X))))
\end{equation}

\section{Additional pseudocode}\label{app:pseudocode}

In this appendix, we give the pseudocode for the algorithms mentioned in \cref{tab:results} that are not already given in \cref{sec:verification}.
We also include full pseudocode for the alternate versions of algorithms that are suggested in \cref{sec:verification}.

\subsection{Anderson's algorithm}\label{app:anderson}
Anderson's algorithm is presented in \cite{anderson1993fine}. It only works for 2 threads, and has different code for both threads.
For clarity, we break from our established shorthand of $j = 1\mathop{-}i$ for just this algorithm, and present the algorithms for $i = 0$ (\cref{alg:anderson-0}) and $j = 1$ (\cref{alg:anderson-1}) separately.
There are six Booleans total: $\varidx{p}{i}, \varidx{p}{j}, \varidx{q}{i}, \varidx{q}{j}, \varidx{t}{i}$ and $\varidx{t}{j}$. All are initialised to $\true$.

\begin{table}[ht!]
\vspace{-1em}
\noindent\begin{minipage}[t]{0.45\textwidth}
\begin{algorithm}[H]
\caption{Anderson's algorithm for $i = 0$}\label{alg:anderson-0}
\begin{algorithmic}[1]
    \State{$\varidx{p}{i} \writeop \false$}
    \State{$\varidx{q}{i} \writeop \false$}
    \State{$\varname{x} \writeop \varidx{t}{j}$}
    \State{$\varidx{t}{i} \writeop \varname{x}$}
    \If{$\varname{x} = \true$}
        \State{$\varidx{p}{i} \writeop \true$}
        \State{\textbf{await} $\varidx{p}{j} = \true$}
    \Else
        \State{$\varidx{q}{i} \writeop \true$}
        \State{\textbf{await} $\varidx{q}{j} = \true$}
    \EndIf
    \State{\textbf{critical section}}
    \State{$\varidx{p}{i} \writeop \true$}
    \State{$\varidx{q}{i} \writeop \true$}
\end{algorithmic}
\end{algorithm}
\end{minipage}
\hfill
\begin{minipage}[t]{0.45\textwidth}
\begin{algorithm}[H]
\caption{Anderson's algorithm for $j = 1$}\label{alg:anderson-1}
\begin{algorithmic}[1]
    \State{$\varidx{p}{j} \writeop \false$}
    \State{$\varidx{q}{j} \writeop \false$}
    \State{$\varname{x} \writeop \neg\varidx{t}{i}$}
    \State{$\varidx{t}{j} \writeop \varname{x}$}
    \If{$\varname{x} = \true$}
        \State{$\varidx{q}{j} \writeop \true$}
        \State{\textbf{await} $\varidx{p}{i} = \true$}
    \Else
        \State{$\varidx{p}{j} \writeop \true$}
        \State{\textbf{await} $\varidx{q}{i} = \true$}
    \EndIf
    \State{\textbf{critical section}}
    \State{$\varidx{p}{j} \writeop \true$}
    \State{$\varidx{q}{j} \writeop \true$}
\end{algorithmic}
\end{algorithm}
\end{minipage}
\end{table}

\subsection{Aravind's BLRU algorithm, alternate version}\label{app:aravind}
In \cref{sec:aravind}, we proposed an alternate version of Aravind's BLRU algorithm that does satisfy starvation freedom under $\conc_S$. We give the full pseudocode in \cref{alg:aravind-alt}.c
It is still the case that every thread has three registers: $\varname{flag}$ and $\varname{stage}$ (Booleans, initially $\false$), and $\varname{date}$ (natural number, initially the thread's own id).

\begin{minipage}[t]{0.8\textwidth}
\begin{algorithm}[H]
\caption{Aravind's BLRU algorithm, alternate version}\label{alg:aravind-alt}
\begin{algorithmic}[1]
    \State{$\varidx{flag}{i} \writeop \true$}
    \Repeat
        \State{$\varidx{stage}{i} \writeop \false$}
        \State{$\textbf{await } \forall_{j \neq i}:\varidx{flag}{j} = \false \lor (\varidx{date}{i} < \varidx{date}{j} \land \varidx{stage}{j} = \false)$}
        \State{$\varidx{stage}{i} \writeop \true$}
    \Until{$\forall_{j \neq i}:$ $\varidx{stage}{j} = \false$}
    \State{\textbf{critical section}}
    \State{$\varidx{date}{i} \writeop \max(\varidx{date}{0}, ..., \varidx{date}{N-1}) + 1$}
    \If{$\varidx{date}{i} \geq 2N-1$}
        \State{$\forall_{j \in [0...N-1]}: \varidx{date}{j} \writeop j$}
    \EndIf
    \State{$\varidx{stage}{i} \writeop \false$}
    \State{$\varidx{flag}{i} \writeop \false$}
\end{algorithmic}
\end{algorithm}
\end{minipage}

\subsection{Attiya-Welch's algorithm}\label{app:attiya-welch}
In \cref{sec:attiya-welch}, we suggest alternate versions for both versions of the Attiya-Welch algorithm. We give the pseudocode of these versions here.
The registers used are the same as what is stated in \cref{sec:attiya-welch}, except that in the variant version we add a local variable $t$.

\begin{table}[ht!]
\vspace{-1em}
\noindent\begin{minipage}[t]{0.47\textwidth}
\begin{algorithm}[H]
\caption{Attiya-Welch algorithm, orig. alt.}\label{alg:attiya-welch-alt}
\begin{algorithmic}[1]
        \State{$\varidx{flag}{i} \writeop \false$}
        \State{\textbf{await} $\varidx{flag}{j} = \false \lor \varname{turn} = j$}
        \State{$\varidx{flag}{i} \writeop \true$}
        \If{$\varname{turn} = i$}
            \If{$\varidx{flag}{j} = \true$}
                \State{\textbf{goto} \cref{attiya-welch-entry2}}
            \EndIf
        \Else
            \State{\textbf{await} $\varidx{flag}{j} = \false$}
        \EndIf
        \State{\textbf{critical section}}
        \If{$\varname{turn} \neq i$}
            \State{$\varname{turn} \writeop i$}
        \EndIf
        \State{$\varidx{flag}{i} \writeop \false$}
\end{algorithmic}
\end{algorithm}
\end{minipage}
\hfill
\begin{minipage}[t]{0.47\textwidth}
\begin{algorithm}[H]
\caption{Attiya-Welch algorithm, var. alt.}\label{alg:attiya-welch-var-alt}
\begin{algorithmic}[1]
    \Repeat
    \State{$\varidx{flag}{i} \writeop \false$}
    \State{\textbf{await} $\varidx{flag}{j} = \false \lor \varname{turn} = j$}
    \State{$\varidx{flag}{i} \writeop \true$}
    \State{$t \gets \varname{turn}$}
    \Until{$t = j \lor \varidx{flag}{j} = \false$}
    \If{$t = j$}
    \State{\textbf{await} $\varidx{flag}{j} = \false$}
    \EndIf
    \State{\textbf{critical section}}
    \If{$\varname{turn} \neq i$}
            \State{$\varname{turn} \writeop i$}
        \EndIf
    \State{$\varidx{flag}{i} \writeop \false$}
\end{algorithmic}
\end{algorithm}
\end{minipage}
\end{table}

\subsection{Burns-Lynch's algorithm and Lamport's 1-bit algorithm}\label{app:burns-lynch}
We give the Burns-Lynch algorithm as \cref{alg:burns-lynch} and Lamport's 1-bit algorithm as \cref{alg:lamport1bit}.
For the sake of easy comparison, we name the one Boolean that both algorithms use $\varidx{flag}{i}$ for each thread $i$ in both algorithms, although Lamport uses the name $\varidx{x}{i}$ and Burns and Lynch use $\varidx{Flag}{i}$.
In both cases, it is initially $\false$.

\begin{table}[ht!]
\vspace{-1em}
\noindent\begin{minipage}[t]{0.47\textwidth}
\begin{algorithm}[H]
\caption{The Burns-Lynch algorithm}\label{alg:burns-lynch}
\begin{algorithmic}[1]
    \Repeat
        \State{$\varidx{flag}{i} \writeop \false$}\label{bl-3}
        \State{\textbf{await} $\forall_{j < i}: \varidx{flag}{j} = \false$}
        \State{$\varidx{flag}{i} \writeop \true$}
    \Until{$\forall_{j < i}:$ $\varidx{flag}{j} = \false$}
    \State{\textbf{await} $\forall_{j > i}: \varidx{flag}{j} = \false$}
    \State{\textbf{critical section}}
    \State{$\varidx{flag}{i} \writeop \false$}
\end{algorithmic}
\end{algorithm}
\end{minipage}
\hfill
\begin{minipage}[t]{0.47\textwidth}
\begin{algorithm}[H]
\caption{Lamport's 1-bit algorithm}\label{alg:lamport1bit}
\begin{algorithmic}[1]
    \State{$\varidx{flag}{i} \writeop \true$}\label{lamport1-l}
    \For{$j$ \textbf{from} $0$ \textbf{to} $i - 1$}
        \If{$\varidx{flag}{j} = \true$}
            \State{$\varidx{flag}{i} \writeop \false$}
            \State{\textbf{await} $\varidx{flag}{j} = \false$}
            \State{\textbf{goto} \cref{lamport1-l}}
        \EndIf
    \EndFor
    \For{$j$ \textbf{from} $i + 1$ \textbf{to} $N - 1$}
        \State{\textbf{await} $\varidx{flag}{j} = \false$}
    \EndFor
    \State{\textbf{critical section}}
    \State{$\varidx{flag}{i} \writeop \false$}
\end{algorithmic}
\end{algorithm}
\end{minipage}
\end{table}

\subsection{Dekker's algorithm variants}\label{app:dekkers}
Two variants of Dekker's algorithm are discussed in \cref{sec:dekker}: the alternate version of the original algorithm we propose, and the RW-safe version from \cite{buhr2016dekker}.
We give the pseudocode for these variants here.

\begin{table}[ht!]
\vspace{-1em}
\noindent\begin{minipage}[t]{0.45\textwidth}
\begin{algorithm}[H]
\caption{Dekker's algorithm alt.}\label{alg:dekker-alt}
\begin{algorithmic}[1]
    \State{$\varidx{flag}{i} \writeop \true$}
    \While{$\varidx{flag}{j} = \true$}
        \If{$\varname{turn} = j$}
            \State{$\varidx{flag}{i} \writeop \false$}
            \State{\textbf{await} $\varname{turn} = i$}
            \State{$\varidx{flag}{i} \writeop \true$}
        \EndIf
    \EndWhile
    \State{\textbf{critical section}}
    \If{$\varname{turn} \neq j$}
        \State{$\varname{turn} \writeop j$}
    \EndIf
    \State{$\varidx{flag}{i} \writeop \false$}
\end{algorithmic}
\end{algorithm}
\end{minipage}
\hfill
\begin{minipage}[t]{0.5\textwidth}
    \begin{algorithm}[H]
\caption{Dekker's algorithm RW-safe}\label{alg:dekker-RW-safe}
\begin{algorithmic}[1]
    \State{$\varidx{flag}{i} \writeop \true$}
    \While{$\varidx{flag}{j} = \true$}
        \If{$\varname{turn} = j$}
            \State{$\varidx{flag}{i} \writeop \false$}
            \State{\textbf{await} $\varname{turn} = i \lor \varidx{flag}{j} = \false$}
            \State{$\varidx{flag}{i} \writeop \true$}
        \EndIf
    \EndWhile
    \State{\textbf{critical section}}
    \If{$\varname{turn} \neq j$}
        \State{$\varname{turn} \writeop j$}
    \EndIf
    \State{$\varidx{flag}{i} \writeop \false$}
\end{algorithmic}
\end{algorithm}
\end{minipage}
\end{table}

\subsection{Dijkstra's algorithm variant}\label{app:dijkstra}
As stated in \cref{sec:dijkstra}, the deadlock freedom violation under $\conc_S$ can be prevented by ensuring that certain writes to the $c$ variables only occur when they would actually change the value.
The resulting pseudocode is given in \cref{alg:dijkstra-alt}.

\begin{table}[ht!]
\vspace{-1em}
\noindent\begin{minipage}[t]{0.43\textwidth}
\begin{algorithm}[H]
\caption{Dijkstra's algorithm}\label{alg:dijkstra-alt}
\begin{algorithmic}[1]
    \State{$\varidx{b}{i} \writeop \false$}\label{dijkstra-Li0-alt}
    \If{$\varname{k} \neq i$}\label{dijkstra-Li1-alt}
        \If{$\varidx{c}{i} \neq \true$}
        \State{$\varidx{c}{i} \writeop \true$}\label{dijkstra-Li2-alt}
        \EndIf
        \If{$\varidx{b}{k} = \true$}\label{dijkstra-Li3-alt}
            \State{$\varname{k} \writeop i$}
        \EndIf
        \State{\textbf{goto} \cref{dijkstra-Li1-alt}}
    \Else
        \If{$\varidx{c}{i} \neq \false$}
        \State{$\varidx{c}{i} \writeop \false$}\label{dijkstra-Li4-alt}
        \EndIf
        \For{$j$ \textbf{from} $0$ \textbf{to} $N - 1$}
            \If{$j \neq i \land \varidx{c}{j} = \false$}
                \State{\textbf{goto} \cref{dijkstra-Li1-alt}}
            \EndIf
        \EndFor
    \EndIf
    \State{\textbf{critical section}}
    \State{$\varidx{c}{i} \writeop \true$}
    \State{$\varidx{b}{i} \writeop \true$}
\end{algorithmic}
\end{algorithm}
\end{minipage}
\end{table}

\subsection{Lamport's 3-bit algorithm}\label{app:lamport}
Lamport's 3-bit algorithm works for arbitrary $N$.
Every thread $i$ has three Booleans: $\varidx{x}{i}, \varidx{y}{i}$ and $\varidx{z}{i}$; the former two are initialised to $\false$, the latter is arbitrary.

This algorithm uses some auxiliary operations. 
The operation $\mathrm{ORD}(S)$ returns for a set $S \subseteq \TID$ a list of all elements in $S$ ordered from smallest to largest. Note that such an ordering is defined for $\TID$ since, \hyperlink{threadorder}{as we stated in} \cref{sec:verification}, we use natural numbers to represent thread identifiers. This list is formalised as an increasing function $\gamma:\{1,\dots,M\} \rightarrow S$, where $M=|S|$. We write $\textrm{domain}(\gamma)$ for $\{1,\ldots,M\}$ and $\textrm{range}(\gamma)$ for $S$.
Additionally, there is the Boolean function $\mathrm{CG}(v, l)$, where $v: \{1,\dots,M\} \rightarrow \mathbbm{B}$ is a Boolean function mapping an index in $\{1,\dots,M\}$ (denoting an element of $S$) to either $\true$ or $\false$, and $l \in \{1,\dots,M\}$.
\begin{align*}
   \mathrm{CG}(v,l) \defeq&~  (v(l) = \mathrm{CGV}(v, l)) \\
   \mathrm{CGV}(v, l) \defeq&\begin{cases}
        \neg v(l{-}1) &\text{if $l > 1$}\\
        v(M) &\text{if $l = 1$}
        \end{cases}
\end{align*}

We write ``\textbf{for} $i$ \textbf{from} $j$\ \textbf{cyclically to} $k$'' to mean an iteration that begins with $i = j$, then increments $i$ by 1, taking the result modulo $N$. The iteration terminates when $i = k$, without executing the loop with $i = k$.
We use $\oplus$ for addition modulo $N$.

The definitions of $\mathrm{ORD}$, $\mathrm{CG}$ and $\mathrm{CGV}$ given above differ from those given by Lamport in \cite{Lamport86Mutex2}. He works with \emph{cycles}, rather than lists, and defines $\mathrm{CG}$ and $\mathrm{CGV}$ based on indexed elements of those cycles, rather than using the indices directly as we do. We believe that our presentation is equivalent to the one given by Lamport, but find it more straightforward to explain. This presentation also aligns better with how the algorithm would be implemented, and indeed aligns more closely with our mCRL2 model.

In addition to these changes in presentation, we also made one change to the algorithm itself, as suggested in \cref{sec:lamport-3bit}: \cref{l-new} was added by us to emphasise that the mapping $\zeta$ is a snapshot of the variables $\varidx{z}{i}$ for all $i \in \textrm{range}(\gamma)$, and that the registers themselves do not get repeatedly read on \cref{lamport-3bit-functioncall}.
For this, we introduce the operation $\mathrm{SAVE_{\varname{z}}^{\gamma}}$ that creates a mapping from $\textrm{domain}(\gamma)$ to the Booleans, such that for all $h \in \textrm{domain}(\gamma)$, $\mathrm{SAVE^{\gamma}_{\varname{z}}}(h) = \varidx{z}{\gamma(h)}$.
Lamport does not state explicitly in \cite{Lamport86Mutex2} that the $\varname{z}$ registers must not be read repeatedly when computing the minimum, but if the registers are non-atomic and re-read during the computation it is possible that no minimum will be found.

\begin{table}[ht!]
\vspace{-1em}
\noindent\begin{minipage}[t]{0.64\textwidth}
\begin{algorithm}[H]
\caption{Lamport's 3-bit algorithm}
\label{alg:lamport3bit}
\begin{algorithmic}[1]
  \State{$\varidx{y}{i} \writeop \true$} 
  \State{$\varidx{x}{i} \writeop \true$} \label{l1}
  \State{$\gamma \writeop \mathrm{ORD}\{j\mid \varidx{y}{j} = \true\}$} \label{l2}
  \State{$\zeta \writeop \mathrm{SAVE}^{\gamma}_{\varname{z}}$}\label{l-new}
  \State{$f\writeop \gamma(\mathrm{minimum}\{h\in\textrm{domain}(\gamma)\mid \mathrm{CG}(\zeta,h)=\true\})$}\label{lamport-3bit-functioncall}
  \For{$j$ \textbf{from} $f$\ \textbf{cyclically to} $i$}
  \If{$\varidx{y}{j} = \true$}
    \If{$\varidx{x}{i} = \true$} {$\varidx{x}{i} \writeop \false$}\EndIf
  \State{\textbf{goto} \cref{l2}}
  \EndIf
  \EndFor
  \If{$\varidx{x}{i} = \false$} {\textbf{goto} \cref{l1}} \EndIf
  \For{$j$ \textbf{from} $i\oplus 1$ \textbf{cyclically to} $f$}
  \If{$\varidx{x}{j} = \true$} {\textbf{goto} \cref{l2}}\EndIf
  \EndFor
  \State{\textbf{critical section}}
  \If{$\varidx{z}{i} = \true$}
    \State{$\varidx{z}{i} \writeop \false$}
  \Else 
    \State{$\varidx{z}{i} \writeop \true$}
  \EndIf
  \State{$\varidx{x}{i} \writeop \false$}
  \State{$\varidx{y}{i} \writeop \false$}
 \end{algorithmic}
\end{algorithm}
\end{minipage}
\end{table}

\subsection{Szymanski's flag algorithm, alternate version}\label{app:szymanski-flag}
In \cref{sec:szymanski-flag} we observed that by changing the exit protocol of Szymanski's flag algorithm with Booleans slightly, we could have it satisfy the same properties as the integer version of the same algorithm.
The resulting algorithm is given as \cref{alg:Szy-flag-bits-alt}.

\begin{table}[ht!]
\vspace{-1em}
\noindent\begin{minipage}[t]{0.85\textwidth}
\begin{algorithm}[H]
\caption{Szymanski's flag algorithm implemented with Booleans, alt.}\label{alg:Szy-flag-bits-alt}
\begin{algorithmic}[1]
  \State{$\varidx{intent}{i} \writeop \true$}
  \State{\textbf{await} $\forall_j:\ \varidx{intent}{j} = \false \lor \varidx{door\_in}{j} = \false$}
  \State{$\varidx{door\_in}{i}\writeop \true$}
  \If{$\exists_j:\ \varidx{intent}{j} = \true \land \varidx{door\_in}{j} = \false$}
    \State{$\varidx{intent}{i} \writeop \false$}
    \State{\textbf{await} $\exists_j:\ \varidx{door\_out}{j} = \true$}
  \EndIf
  \If{$\varidx{intent}{i} = \false$} {$\varidx{intent}{i} \writeop \true$}\EndIf
  \State{$\varidx{door\_out}{i} \writeop \true$}
  \State{\textbf{await} $\forall_{j < i}:\ \varidx{door\_in}{j} = \false$}
  \State{\textbf{critical section}}
  \State{\textbf{await} $\forall_{j>i}:\ \varidx{door\_in}{j} = \false \lor \varidx{door\_out}{j} = \true$}
  \State{$\varidx{door\_out}{i}\writeop \false$}
  \State{$\varidx{intent}{i} \writeop \false$}
  \State{$\varidx{door\_in}{i} \writeop \false$}
\end{algorithmic}
\end{algorithm}
\end{minipage}
\end{table}

\subsection{Szymanski's 3-bit linear wait algorithm, alternate}\label{app:szymanski-3bit}
In \cref{sec:Szymanski}, we give a slightly altered version of Szymanski's 3-bit linear wait algorithm for which, with just two threads, all three properties are satisfied with both non-atomic and atomic registers.
This algorithm is given as \cref{alg:szymanski-3bit-alt}.

\begin{table}[ht!]
\vspace{-1em}
\noindent\begin{minipage}[t]{0.8\textwidth}
\begin{algorithm}[H]
    \caption{Szymanski's 3-bit linear wait algorithm, alt.}
    \label{alg:szymanski-3bit-alt}
    \begin{algorithmic}[1]
      \State{$\varidx{a}{i} \writeop \true$}
      \lFor{$j$ \textbf{from} $0$\ \textbf{to}\ $N{-}1$} {\textbf{await} $\varidx{s}{j} = \false$}
      \State{$\varidx{w}{i} \writeop \true$}
      \State{$\varidx{a}{i} \writeop \false$}
      \While{$\varidx{s}{i} = \false$}
        \State{$j \writeop 0$}
        \lWhile{$j < N \land \varidx{a}{j} = \false$}{$j\writeop j+1$}
          \If{$j=N$}
            \State{$\varidx{s}{i} \writeop \true$}
            \State{$j\writeop 0$}         
            \lWhile{$j<N\land \varidx{a}{j} = \false$}{$j\writeop j+1$}
            \If{$j< N$} {$\varidx{s}{i} \writeop \false$}
           \Else
             \State{$\varidx{w}{i} \writeop \false$}
             \lFor{$j$ \textbf{from} $0$\ \textbf{to} $N-1$}{\textbf{await} $\varidx{w}{j} = \false$}
           \EndIf
         \EndIf
         \If{$j<N$}
           \State{$j\writeop 0$}
           \lWhile{$j<N \land (\varidx{s}{j} = \false \lor \varidx{w}{j} = \true)$}{$j\writeop j+1$}
         \EndIf
         \If{$j \neq i \land j<N$}
         \State{$\varidx{s}{i} \writeop \true$}
         \State{$\varidx{w}{i} \writeop \false$}
         \EndIf
     \EndWhile
      \lFor{$j$ \textbf{from} $0$ \textbf{to} $i-1$}{\textbf{await} $\varidx{s}{j} = \false$}
      \State{\textbf{critical section}}
      \State{$\varidx{s}{i} \writeop \false$}
     \end{algorithmic}
  \end{algorithm}
\end{minipage}
\end{table}

\end{document}